\UseRawInputEncoding
\documentclass[a4paper,11pt]{article}

\usepackage[left=2.5cm,right=2.5cm,top=3cm,bottom=3cm]{geometry}
\usepackage[english]{babel}
\usepackage{amsmath,amssymb,amsthm}
\usepackage{mathrsfs}
\usepackage{enumitem}

\newtheorem{theorem}{Theorem}
\numberwithin{theorem}{section}
\newtheorem{corollary}[theorem]{Corollary}
\newtheorem{lemma}[theorem]{Lemma}
\newtheorem{proposition}[theorem]{Proposition}

\theoremstyle{definition}
\newtheorem{definition}[theorem]{Definition}
\newtheorem{example}[theorem]{Example}
\newtheorem{remark}[theorem]{Remark}
\newtheorem{notation}[theorem]{Notation}
\newtheorem{fact}[theorem]{Fact}

\newcommand{\sep}{\colon\hspace{0.4mm}}

\newcommand{\metp}{\mathsf{d}_{|\hspace{0.3mm}\cdot\hspace{0.3mm}|_p}}

\newcommand{\ttL}{\hspace{0.3mm}\mathtt{L}}
\newcommand{\ttM}{\hspace{0.3mm}\mathtt{M}}
\newcommand{\ttN}{\hspace{0.3mm}\mathtt{N}}

\newcommand{\ple}{l_{\mathrm{P}}}

\newcommand{\densr}{\mathcal{D}_r(\Hp)}
\newcommand{\denss}{\mathcal{D}_s(\Hp)}
\newcommand{\densu}{\mathcal{D}_1(\Hp)}
\newcommand{\densops}{\mathcal{D}(\Hp)}

\newcommand{\hphps}{\widehat{\phi\psi}}
\newcommand{\hpsph}{\widehat{\psi\phi}}
\newcommand{\hpsps}{\widehat{\psi\psi}}
\newcommand{\hefj}{\widehat{\hspace{0.9mm}\osj\fj}}

\newcommand{\inj}{\tau_{\Hp}}
\newcommand{\tinj}{\widetilde{\tau}_{\Hp}}

\newcommand{\sta}{\Omega}
\newcommand{\stext}{\sta_{\mathrm{ext}}}
\newcommand{\states}{\mathcal{S}(\Hp)}
\newcommand{\statex}{\mathcal{S}_{\mathrm{ext}}(\Hp)}
\newcommand{\tcst}{\mathcal{T}_{\mathrm{st}}(\Hp)}
\newcommand{\trstates}{\mathcal{S}_{\mathrm{tr}}(\Hp)}
\newcommand{\ztc}{\mathcal{T}_{\mathrm{sa}}(\Hp)_0}
\newcommand{\ztcs}{\widehat{\mathcal{T}}_{\mathrm{sa}}(\Hp)_0}
\newcommand{\ztcr}{\mathcal{T}_{\mathrm{sa}}(\Hp)_0^r}
\newcommand{\sts}{\widehat{\mathcal{T}}_{\mathrm{st}}(\Hp)}

\newcommand{\cl}{\mathrm{cl}\hspace{0.3mm}}

\newcommand{\pdj}{\pi_j}

\newcommand{\tpdk}{\widetilde{\pi}_k}
\newcommand{\prs}{\varpi_0}
\newcommand{\psy}{\upsilon_0}
\newcommand{\de}{\mathrm{d}}
\newcommand{\spm}{P_{\hspace{-0.5mm}A}}

\newcommand{\scj}{\mathfrak{sc}}
\newcommand{\acj}{\mathfrak{ac}}

\newcommand{\indf}{1_{\hspace{-0.2mm}J}}
\newcommand{\indfb}{1_{\hspace{-0.3mm}K}}

\newcommand{\baj}{\mathscr{D}_{\hspace{-0.2mm}J}}
\newcommand{\bak}{\mathscr{D}_{\hspace{-0.2mm}K}}

\newcommand{\cspanq}{\overline{\mathrm{span}}_{\hspace{0.3mm}\mathbb{Q}_p}\hspace{-0.4mm}}

\newcommand{\acop}{\mathrm{aco}_{\hspace{0.3mm}\mathbb{Q}_p}\hspace{-0.4mm}}
\newcommand{\cacop}{\overline{\mathrm{aco}}_{\hspace{0.3mm}\mathbb{Q}_p}\hspace{-0.4mm}}
\newcommand{\acoset}{\mathscr{A}}
\newcommand{\suse}{\mathscr{X}}
\newcommand{\coset}{\mathscr{C}}
\newcommand{\affset}{\mathscr{H}}
\newcommand{\lis}{\mathscr{L}}
\newcommand{\cop}{\mathrm{co}_{\hspace{0.3mm}\mathbb{Q}_p}\hspace{-0.4mm}}
\newcommand{\ccop}{\overline{\mathrm{co}}_{\hspace{0.3mm}\mathbb{Q}_p}\hspace{-0.4mm}}
\newcommand{\aff}{\mathrm{aff}_{\hspace{0.3mm}\mathbb{Q}_p}\hspace{-0.4mm}}
\newcommand{\caff}{\overline{\mathrm{aff}}_{\hspace{0.3mm}\mathbb{Q}_p}\hspace{-0.4mm}}

\newcommand{\osj}{e_j}

\newcommand{\osk}{e_k}
\newcommand{\ds}{\widetilde{e}}
\newcommand{\dsj}{\ds_j}
\newcommand{\tfj}{\widetilde{f}_j}
\newcommand{\dsk}{\ds_k}
\newcommand{\tfm}{\widetilde{f}_m}
\newcommand{\fj}{f_j}

\newcommand{\sj}{\sigma_j}
\newcommand{\pro}{\odot}

\newcommand{\xm}{\xi^{\mbox{\tiny $(m)$}}}

\newcommand{\hsp}{\rangle_{\hspace{-0.2mm}\mbox{\tiny $\Tp$}}}
\newcommand{\bhsp}{\hspace{-0.2mm}\big{\rangle}_{\hspace{-0.5mm}\mbox{\tiny $\Tp$}}\hspace{-0.4mm}}
\newcommand{\inpr}{\langle\hspace{0.3mm}\cdot\hspace{0.6mm},\cdot\hspace{0.3mm}\hsp}
\newcommand{\scap}{\langle\hspace{0.3mm}\cdot\hspace{0.6mm},\cdot\hspace{0.3mm}\rangle}
\newcommand{\scapx}{\langle x,\cdot\hspace{0.4mm}\rangle}
\newcommand{\scapsi}{\langle\psi,\cdot\hspace{0.4mm}\rangle}
\newcommand{\scaphi}{\langle\phi,\cdot\hspace{0.4mm}\rangle}
\newcommand{\scaphii}{\langle\phi_i,\cdot\hspace{0.4mm}\rangle}
\newcommand{\scaphij}{\langle\phi_j,\cdot\hspace{0.4mm}\rangle}
\newcommand{\scaphin}{\langle\phi_n,\cdot\hspace{0.4mm}\rangle}
\newcommand{\scapcx}{\big\langle\overline{\argo},x\big\rangle}
\newcommand{\scapeta}{\langle\eta,\cdot\hspace{0.4mm}\rangle}

\DeclareMathOperator{\opE}{\mathnormal E}
\newcommand{\jkE}{\hspace{-0.7mm}\sideset{^{jk\hspace{-0.5mm}}}{^{\hspace{0.4mm}\Phi}}{\opE}}
\newcommand{\jkEmn}{\hspace{-0.4mm}\sideset{^{jk\hspace{-0.5mm}}}{^{\hspace{0.4mm}\Phi}_{mn}}{\opE}}
\newcommand{\rsE}{\hspace{-0.7mm}\sideset{^{rs\hspace{-0.5mm}}}{^{\hspace{0.4mm}\Phi}}{\opE}}

\DeclareMathOperator{\opT}{\mathnormal T}
\newcommand{\kT}{\hspace{-0.7mm}\sideset{^k}{}{\opT}}
\newcommand{\kTmn}{\hspace{-0.4mm}\sideset{^k}{_{mn}}{\opT}}
\newcommand{\lT}{\hspace{-0.7mm}\sideset{^l}{}{\opT}}
\newcommand{\lTmn}{\hspace{-0.4mm}\sideset{^l}{_{mn}}{\opT}}

\DeclareMathOperator{\opC}{\mathnormal C}
\newcommand{\kC}{\hspace{-0.7mm}\sideset{^{k\hspace{-0.3mm}}}{}{\opC}}
\newcommand{\kCmn}{\hspace{-0.4mm}\sideset{^{k\hspace{-0.3mm}}}{_{mn}}{\opC}}

\newcommand{\defi}{\mathrel{\mathop:}=}
\newcommand{\ifed}{=\mathrel{\mathop:}}

\newcommand{\iffdef}{\overset{\mathrm{def}}{\iff}}

\newcommand{\id}{\mathrm{Id}}

\newcommand{\nat}{\mathbb{N}}
\newcommand{\enne}{\mathsf{N}}
\newcommand{\senne}{\mathbb{N}_{\mbox{\tiny $\le$}\enne}}
\newcommand{\menne}{\}_{m=1}^{\enne}}
\newcommand{\nenne}{\}_{n=1}^{\enne}}
\newcommand{\mnat}{\}_{m\in\nat}}
\newcommand{\nnat}{\}_{n\in\nat}}
\newcommand{\ka}{{\hspace{0.4mm}\mathtt{K}}}

\newcommand{\eps}{\mathscr{E}_\psi}
\newcommand{\ephps}{\mathscr{E}_{\phi,\psi}}

\newcommand{\smnor}{\mbox{\tiny $\|\hspace{-0.5mm}\cdot\hspace{-0.5mm}\|$}}
\newcommand{\nperp}{\overset{\smnor}{\perp}}

\newcommand{\ipp}{\mathscr{I}\hspace{-0.3mm}}
\newcommand{\nop}{\mathscr{N}\hspace{-0.3mm}}

\newcommand{\op}{\mathrm{op}\hspace{0.3mm}}
\newcommand{\opPhi}{\mathrm{op}_\Phi}
\newcommand{\opPsi}{\mathrm{op}_\Psi}
\newcommand{\card}{\mathrm{card}}
\newcommand{\isin}{\nu_{p,\mu}}
\newcommand{\tr}{\mathrm{tr}}

\newcommand{\dphi}{\phi^\prime}
\newcommand{\dA}{A^\prime}
\newcommand{\Wf}{W_{\hspace{-0.3mm}\Phi}}
\newcommand{\supp}{\mathrm{supp}}

\newcommand{\argo}{(\cdot)}

\DeclareMathOperator{\md}{mod}
\DeclareMathOperator{\dom}{dom}
\DeclareMathOperator{\ran}{ran}

\newcommand{\calg}{\mathfrak{A}}
\newcommand{\calgs}{\mathfrak{A}_{\mathrm{sa}}}
\newcommand{\salge}{\mathfrak{C}}
\newcommand{\stalg}{\mathfrak{S}(\mathfrak{A})}
\newcommand{\trast}{\mathfrak{S}_{\mathrm{tr}}(\mathfrak{A})}

\newcommand{\cst}{\mathfrak{S}(\mathcal{B}(\mathcal{K}))}
\newcommand{\densK}{\mathcal{D}(\mathcal{K})}
\newcommand{\proK}{\mathcal{P}(\mathcal{K})}
\newcommand{\bsa}{\mathcal{B}_{\mathrm{sa}}(\mathcal{K})}

\newcommand{\linop}{\mathrm{Lin}\hspace{0.3mm}(\Hp)}

\newcommand{\bas}{\mathfrak{B}}
\newcommand{\uba}{\mathcal{B}(\Hp)_1}
\newcommand{\usp}{\mathcal{B}(\Hp)_{[1]}}
\newcommand{\Hba}{\mathcal{H}_1}

\newcommand{\Qba}{\Q_1}
\newcommand{\pint}{\mathbb{Z}_{p}}
\newcommand{\mide}{\mathfrak{P}_{p}}
\newcommand{\vari}{\mathfrak{V}_{p,\hspace{0.5mm}\mu}}

\newcommand{\sephi}{\mathscr{S}_{\hspace{-0.7mm}\phi}}
\newcommand{\Minf}{\mathsf{M}_{\infty}(\Q)}

\newcommand{\Qp}{\mathbb{Q}_p}
\newcommand{\Qpa}{\mathbb{Q}_p^{\ast}}
\newcommand{\Qtra}{\mathbb{Q}_3^{\ast}}
\newcommand{\Qca}{\mathbb{Q}_5^{\ast}}
\newcommand{\Q}{\mathbb{Q}_p(\sqrt{\mu})}
\newcommand{\Qa}{\mathbb{Q}_p(\sqrt{\mu})^\ast}

\newcommand{\Bp}{\mathcal{B}(\mathcal{H})}
\newcommand{\Bpa}{\mathcal{B}_{\mathrm{ad}}(\Hp)}
\newcommand{\Bps}{\mathcal{B}_{\mathrm{sa}}(\Hp)}

\newcommand{\Up}{\mathcal{U}(\Hp)}
\newcommand{\Tp}{\mathcal{T}(\Hp)}
\newcommand{\Tps}{\mathcal{T}_{\mathrm{sa}}(\Hp)}
\newcommand{\Tw}{\mathcal{T}_{\mathrm{w}}(\Hp)}

\newcommand{\Cp}{\mathcal{C}(\mathcal{H})}
\newcommand{\Cpa}{\mathcal{C}_{\mathrm{ad}}(\mathcal{H})}

\newcommand{\Fp}{\mathscr{F}(\mathcal{H})}
\newcommand{\blfi}{\mathscr{B}_\Phi}

\newcommand{\rcf}{\mathbb{F}_p}

\newcommand{\Hp}{\mathcal{H}}
\newcommand{\Hpu}{\mathcal{H}_{p,\mu}}
\newcommand{\cV}{\mathcal{V}}
\newcommand{\dH}{\Hp^\prime}
\newcommand{\dX}{X^\prime}
\newcommand{\ddX}{X^{\prime\prime}}
\newcommand{\ddH}{\Hp^{\prime\prime}}
\newcommand{\IX}{\mathcal{I}_{X}}
\newcommand{\IH}{\mathcal{I}_{\mathcal{H}}}
\newcommand{\JH}{\mathcal{J}_{\mathcal{H}}}
\newcommand{\ellphi}{\mathcal{L}_{\Phi}}

\begin{document}

\title{Trace class operators and states in $p$-adic quantum mechanics}

\author{Paolo Aniello,$^{1,2}$\thanks{Email: paolo.aniello@na.infn.it} \hspace{1.2mm}
Stefano Mancini,$^{3,4}$\thanks{Email: stefano.mancini@unicam.it} \hspace{0.6mm}
and Vincenzo Parisi$\hspace{0.4mm}^{3,4}$\thanks{Email: vincenzo.parisi@unicam.it}}

\date{
{\normalsize
$^1$\textit{Dipartimento di Fisica  ``Ettore Pancini'', Universit\`a  di Napoli ``Federico II",\\
Complesso Universitario di Monte S.~Angelo, via Cintia, I-80126 Napoli, Italy}\\	
\vspace{3mm}
$^2$\textit{Istituto Nazionale di Fisica Nucleare, Sezione di Napoli,\\
Complesso Universitario di Monte S.~Angelo, via Cintia, I-80126 Napoli, Italy}\\
\vspace{3mm}
$^3$\textit{School of Science and Technology, University of Camerino,\\
Via Madonna delle Carceri, 9, Camerino, I-62032, Italy}\\
\vspace{3mm}
$^4$\textit{Istituto Nazionale di Fisica Nucleare, Sezione di Perugia,\\
via A.~Pascoli, I-06123 Perugia, Italy}
}	
}
	
\maketitle
	
\begin{abstract}
Within the framework of quantum mechanics over a quadratic extension of the non-Archimedean field of $p$-adic numbers, we provide
a definition of a quantum state relying on a general algebraic approach and on a $p$-adic model of probability theory. As in the
standard complex case, a distinguished set of physical states are related to a notion of trace for a certain class of
bounded operators and, in fact, we show that one can define a suitable space of trace class operators in the non-Archimedean setting, as well.
The analogies --- but also the several (highly non-trivial) differences --- with respect to the case of standard quantum mechanics
in a complex Hilbert space are analyzed.
\end{abstract}


 	
\section{Introduction}

Quantum mechanics and general relativity are undoubtedly the most successful physical theories of the past century.
On the one hand, they have transformed our understanding of physical reality by showing
that microscopic systems have an intrinsic indeterministic character and that gravity
can be described as an effect of the curvature of space-time. On the other hand, they have shown that
the investigation and understanding of fundamental physical phenomena should rely on a completely different approach
with respect to that used to develop the theories of earlier centuries. In this regard, Dirac's enlightening words capture
this radical change in physical inquiry: ``I learnt to distrust all physical concepts as the basis for a theory. Instead one should
put one's trust in a mathematical scheme, even if the scheme does not appear at first sight to be connected with physics.
[\,\textellipsis] The basic equations of the theory where worked out before their physical meaning was obtained.
The physical meaning had to follow behind the mathematics''~\cite{dirac1978mathematical}.

Nowadays, many physicists believe that quantum mechanics and general relativity are unfit to describe a wide range of physical phenomena
related to the ultimate structure of matter and space-time at a scale comparable
to Planck's length ($\ple=\sqrt{\hbar\, G/c^3}\sim 10^{-35}m$)~\cite{peres1960quantum}. It seems, therefore, necessary
to find new theoretical schemes and, according to Dirac's view of modern physics, it is likely that the identification
of the right mathematical framework for these new physical models will play a major role.

At the end of the last century, Volovich and Vladimirov~\cite{volovich1987number,volovich2010number,vladimirov1989p} envisaged a description
of microscopic phenomena based on a new picture of space-time at Planck's scale. This description stemmed from the fundamental observation that,
if Planck's length is assumed to be the smallest measurable length, then space-time should possess a \emph{non-Archimedean} character.
Pursuing this idea to its logical conclusions, one is led to the further observation that the only complete non-Archimedean field
one can construct starting from the field of rational numbers is, up to isomorphisms, the field of $p$-adic numbers
$\mathbb{Q}_p$~\cite{Cassels,gouvea1997p,schikhof2007ultrametric,theclassfields2007,robert2013course}, where $p$ is a generic prime number.
It is then natural to attempt at developing $p$-adic models of quantum
theory~\cite{volovich2010number,vladimirov1989p,ruelle1989quantum,meurice1989quantum,khrennikov1990representation,khrennikov1991p,albeverio1996p,Vourdas,Vourdas_book}
and formulating field theories on $\mathbb{Q}_p$~\cite{rammal1986ultrametricity,parisi1988p}.

Actually, the mere possibility of constructing a quantum theory based on $\mathbb{Q}_p$ had already been contemplated
in the early 1970s (see the paper of Beltrametti and Cassinelli~\cite{beltrametti1972quantum}, and references therein).
As observed by Beltrametti and Cassinelli, the problem of adopting a certain number field occurs
at two different levels in the formulation of quantum mechanics: first, concerning the values of space-time coordinates
--- and this is precisely the layer Volovich's hypothesis refers to --- and, second,
concerning the carrier vector space of physical states. In principle, the second level must not depend on the first one;
hence, there are (at least) two possibilities that may be investigated separately. In their seminal paper~\cite{beltrametti1972quantum},
Beltrametti and Cassinelli considered the second aspect only, finding inconsistencies related to the usual lattice structure of quantum mechanics.
These inconsistencies essentially originate from the lack of a nontrivial involutive automorphism of $\mathbb{Q}_p$.

Clearly, this analysis \emph{does not} rule out the possibility of constructing a sensible physical theory
of microscopic phenomena based on $\mathbb{Q}_p$. In fact, on the one hand one cannot stipulate, in advance,
that such a theory should rely on a certain given lattice structure (e.g., classical and quantum mechanics \emph{do not}
share the same lattice structure~\cite{varadarajan1968geometry,beltrametti81}); on the other hand, precisely as one
passes from the real to the complex numbers, one can consider a suitable \emph{quadratic extension} of $\mathbb{Q}_p$,
which instead does admit a nontrivial involutive automorphism~\cite{schikhof2007ultrametric,robert2013course,serre2012course}.

Both the possibility of building a quantum theory relying on `wave functions' of the form, say,
$\psi\colon\mathbb{Q}_p\rightarrow \mathbb{C}$ and the `more radical' idea of involving a non-Archimedean field
in the \emph{second} layer characterizing a quantum theory --- the carrier space of physical states ---
have extensively been investigated by Khrennikov, Albeverio and their collaborators
(see~\cite{khrennikov1991p,albeverio1996p,khrennikov1990representations,khrennikov1990mathematical,khrennikov1990quantum,khrennikov1991real,
khrennikov1993statistical,albeverio1997representation,albeverio1997spectrum,albeverio1997fourier,albeverio1998regularization,albeverio2009p},
and references therein), paving the way to a new and intriguing line of research.

In our opinion, however, at least one remarkable aspect of a $p$-adic theory of microscopic phenomena has not been
fully investigated yet. Assuming the point of view according to which the carrier space itself of physical states
should be non-Archimedean, it is not immediately clear how states (and observables) should be defined. In principle,
a theory based entirely on $p$-adic numbers may have a substantially different character with respect to standard quantum mechanics;
e.g., as a consequence of a different interpretation of the mathematical entities of the theory or,
even more drastically, of a completely different mathematical behaviour that may emerge.

The present paper focuses on this aspect, trying to provide the basic mathematical tools for an abstract definition
of a physical state. We follow two general guidelines: The first is to define states using an algebraic approach;
i.e., to describe `$p$-adic states' as linear functionals on some algebra of `p-adic operators'. The second is to adopt
a $p$-adic model of probability theory~\cite{khrennikov1993p,khrennikov2002interpretations}, coherently with our aim
at developing a theory entirely based on $\mathbb{Q}_p$.

As a result, we are led to the conclusion that $p$-adic states should be defined as suitable linear functionals on a
$\ast$-algebra of (bounded) observables over a quadratic extension $\Q$ of the field of $p$-adic numbers, where $\mu$ is
a non-quadratic element of $\mathbb{Q}_p$.

This approach seems to follow the same route as the standard algebraic formulation of quantum mechanics~\cite{Emch,strocchi2008introduction},
but one should not push the analogies between the $p$-adic and the standard complex case too far. E.g., our analysis shows that
\begin{itemize}

\item To start with, the relevant carrier vector space $\Hp\equiv\Hpu$ is \emph{not} a Hilbert space in the usual sense, being modeled
on the Banach space $c_0(\Q)$ of zero-convergent sequences in $\Q$. In particular, it is neither isomorphic to its dual $\dH$ nor reflexive,
in the infinite-dimensional case. (Our subsequent claims in the list refer, in general, to an infinite-dimensional setting.)

\item As in the complex case, one can define the Banach algebra $\Bp$ of all bounded operators, \emph{but}
this is not, in a natural way, a $\ast$-algebra.

\item It is instead the smaller algebra $\Bpa\subset\Bp$ of \emph{adjointable bounded operators} to possess a natural structure
of a Banach $\ast$-algebra (that may be thought of as the algebra of physical observables).

\item As in the complex case, one can define a \emph{trace class} $\Tp\subset\Bpa$, that induces a distinguished class of physical states
--- the so-called \emph{trace-induced states} --- but these states are not required to be \emph{positive}, simply because no natural notion
of positivity can be defined within the mathematical framework we adopt here.

\item The trace class $\Tp$ is a two-sided $\ast$-ideal in $\Bpa$, \emph{but} differently from the complex case,
it does not coincide with the whole class of bounded operators for which a \emph{global} trace functional is well
defined. Moreover, $\Tp$ alone plays
the role, simultaneously, of all \emph{trace ideals} of compact operators that can be defined in a complex
Hilbert space~\cite{Reed,Simon}.

\end{itemize}

The paper is organized as follows. In Section~\ref{sec2}, we recall some basic facts about the field of $p$-adic
numbers $\Qp$ and its quadratic extensions $\Q$. Section~\ref{sec3} is devoted to introducing a notion
of a $p$-adic Hilbert space. We begin by discussing $p$-adic normed and Banach spaces over $\Q$. Then, we consider
a suitable notion of a inner product in a $p$-adic Banach space, thus getting to a definition of a $p$-adic Hilbert space
suitable for our purposes. In Section~\ref{sec4}, we study the bounded and  the adjointable bounded operators in a $p$-adic Hilbert space.
We use extensively the notion of a \emph{matrix operator}, which turns out to be very convenient in the $p$-adic setting.
Indeed, in Section~\ref{sec5}, the $p$-adic unitary operators $\mathcal{U}(\Hp)$ are introduced as matrix operators,
whose complete characterization we then provide in Theorem~\ref{th.u3}. Section~\ref{sec6} is devoted to elaborating
a definition of a trace class operator suitable for a $p$-adic setting (where there is no natural notion of \emph{positive} operator).
We next prove that the linear space $\Tp$ of all trace class operators is a left ideal in the space of bounded operators $\Bp$,
and a two sided $\ast$-ideal in the Banach $\ast$-algebra of bounded adjointable operators $\Bpa$ (Theorem~\ref{th.t3}).
Eventually, we show that $\Tp$ is, in a natural way a a $p$-adic Hilbert space, the so-called $p$-adic Hilbert-Schmidt space
(see Theorem~\ref{theohs}). In Section~\ref{sec7}, we study the physical states in the $p$-adic setting. As anticipated,
their definition is algebraic and adapted to a $p$-adic model of probability theory. Finally, in Section~\ref{conclu},
a few conclusions are drawn, followed by a glance at future prospects.


\section{Basics of $p$-adic numbers}
\label{sec2}

The aim of this section is to recall some fundamental facts regarding the field of $p$-adic numbers $\mathbb{Q}_p$ and its quadratic
extensions~\cite{volovich2010number,Cassels,gouvea1997p,schikhof2007ultrametric,robert2013course,serre2012course,vladimirov1994p,Anashin,shimura2010arithmetic,folland2016course}.

Let $p\in\mathbb{N}$ be a prime number. According to the unique factorization theorem, every nonzero rational number
$x\in\mathbb{Q}$ can be written uniquely in the form $x=p^k m/n$, for some $k,m,n\in\mathbb{Z}$, and with $p\nmid m,n$. The
so-called $p$-adic \emph{absolute value} is defined as a map $|\cdot|_p:\mathbb{Q}\rightarrow\mathbb{R}^+$, with $|0|_p\equiv 0$ and
\begin{equation}
|x|_p\defi p^{-k} \quad (x=p^k m/n\neq 0).
\end{equation}
This map satisfies all the defining properties of an absolute value, or \emph{valuation}; i.e., it is strictly positive on
$\mathbb{Q}^\ast\equiv\mathbb{Q}\setminus\{0\}$, it factorizes under the product of two elements in $\mathbb{Q}$ and
satisfies the triangle inequality. However, it also satisfies
a more stringent condition, the so-called \emph{ultrametric inequality} (or strong triangle inequality), namely,
\begin{equation}
|x+y|_p\leq \max\{|x|_p,|y|_p\},\quad \forall x,y\in\mathbb{Q}.
\end{equation}
This inequality represents the main difference with respect to the standard absolute value on $\mathbb{Q}$.
A valuation on a certain field is called \emph{non-Archimedean} or \emph{ultrametric} if, like $|\cdot|_p$,
it satisfies the strong triangle inequality; otherwise (e.g., in the case of the standard absolute value on $\mathbb{Q}$),
it is called \emph{Archimedean}.

A remarkable theorem due to Ostrowski shows that the standard absolute value $|\cdot|$ and the
$p$-adic absolute value $|\cdot|_p$ --- with $p$ ranging over the prime numbers--- exhaust all
possible mutually \emph{inequivalent} valuations on $\mathbb{Q}$~\cite{schikhof2007ultrametric,theclassfields2007,robert2013course,vladimirov1994p}.
This fact implies that every \emph{ultrametric} absolute value on $\mathbb{Q}$ is equivalent,
for some prime number $p$, to the $p$-adic valuation.

Let
\begin{equation}
\metp(x,y)\defi |x-y|_p, \quad x,y\in\mathbb{Q},
\end{equation}
be the metric induced by $|\cdot|_p$.  Obviously, the strong triangle inequality entails that $\metp$ satisfies
\begin{equation}
\metp(x,y)\leq \max\big\{\metp(x,z),\metp(z,y)\big\},\quad x,y,z\in\mathbb{Q}.
\end{equation}
In the mathematical literature, one refers to a space endowed with such a metric as an \textit{ultrametric space}.
Although the metric $\metp$, as well as the $p$-adic absolute value, differs from the standard metric on $\mathbb{Q}$
essentially for the ultrametric inequality, the consequences of this fact are noteworthy, e.g., from a topological point of
view~\cite{robert2013course,vladimirov1994p}.

For every prime number $p$, by means of a standard procedure~\cite{Cassels,gouvea1997p,theclassfields2007}, the $p$-adic numbers
$\mathbb{Q}_p$ can be defined as the field completion of $\mathbb{Q}$ w.r.t.\ the metric $\metp$, and then $\mathbb{Q}$
can be regarded as a dense subfield of the complete field $\Qp$. We will denote by $\Qpa\equiv\Qp\setminus\{0\}$ the multiplicative group
of $\mathbb{Q}_p$.

\begin{theorem}[\cite{theclassfields2007,robert2013course,vladimirov1994p,Anashin,folland2016course}]
Every $x\in\Qpa$ admits a unique decomposition as a convergent series of the form
\begin{equation}
x=\sum_{i=0}^{\infty}x_i\,p^{i+k}=p^k(x_0+x_1\,p+x_2\,p^2+\cdots),
\quad k\in\mathbb{Z},\ x_i\in\{0,1,...,p-1\},\ x_0\neq 0,
\end{equation}
and, conversely, every series of this form converges to some non-zero element of $\Qp$.
The continuous extension of the $p$-adic absolute value $|\cdot|_p$ on $\mathbb{Q}$ to $\Qp$ ---
extension which we still denote by the same symbol --- is an ultrametric valuation on $\Qp$.
Explicitly, we have:
\begin{equation}
|x|_p=\big|{\textstyle\sum_{i=0}^{\infty}}x_i\,p^{i+k}\big|_p= p^{-k}, \quad
\forall x\in\Qpa.
\end{equation}
\end{theorem}

The field of $p$-adic numbers $\Qp$ --- being endowed with an ultrametric valuation --- is called \emph{ultrametric}
or \emph{non-Archimedean}. The topological peculiarities associated with ultrametricity
justify the use of $p$-adic numbers when describing physics on length scales comparable to Planck's length
$\ple$~\cite{peres1960quantum,vladimirov1989p,vladimirov1994p}.

The so-called \emph{valuation ring} --- w.r.t.\ $|\cdot|_p$ --- of the non-Archimedean field $\Qp$ is the ring
of \emph{$p$-adic integers} $\pint\defi\{\lambda\in\Qp \sep |\lambda|_p\le 1\}=\{\sum_{i=0}^{+\infty}a_ip^i \sep 0\le a_i<p\}$,
i.e., a subring of $\Qp$~\cite{gouvea1997p,theclassfields2007}. The set $\mide\defi\{\lambda\in\Qp \sep |\lambda|_p< 1\}=p\hspace{0.4mm}\pint\subset\pint$
is a maximal ideal in $\pint$ (actually, its unique maximal ideal) --- the so-called \emph{valuation ideal} of $\Qp$ w.r.t.\ $|\cdot|_p$ ---
and every element of $\pint\setminus\mide$ is invertible. The quotient $\pint/\mide$ is called the \emph{residue class field}
of $\Qp$ w.r.t.\ $|\cdot|_p$ (recall that the quotient of a ring by a maximal ideal is always a field); specifically,
$\pint/\mide=\pint/p\hspace{0.4mm}\pint$ is isomorphic to the finite field $\rcf=\mathbb{Z}/p\hspace{0.4mm}\mathbb{Z}$.

We will not discuss the properties of $\mathbb{Q}_p$ any further. A detailed account of the topological properties
of this field can be found, for instance, in~\cite{robert2013course,vladimirov1994p}, while applications to functional analysis
are discussed in~\cite{narici71,rooij1978non,bosch1984non,Perez-Garcia,diagana2016non}. In what follows, we will focus on the description of the
\emph{quadratic extensions} of $\mathbb{Q}_p$.

Every field with characteristic $0$ admits a quadratic extension obtained by adjoining the square root
of a non-quadratic element~\cite{Cassels,gouvea1997p,robert2013course}. Therefore, to classify all the (inequivalent) quadratic extensions
of $\mathbb{Q}_p$, we first need to characterize the non-quadratic elements of this field. We start with the following fact:

\begin{proposition}[\cite{vladimirov1994p}]\label{squareinratp}
A $p$-adic number $x=\sum_{i=0}^{\infty}x_i\,p^{i+k}\in\mathbb{Q}_p^*$, $k\in\mathbb{Z}$, $x_i\in\{0,1,...,p-1\}$, $x_0\neq 0$,
is a quadratic element --- i.e., $x\in(\Qpa)^2$ --- iff the following conditions are satisfied:
\begin{itemize}

\item $k$ is even;

\item if $p\neq2$, the equation $j^2\equiv x_0\pmod p$ admits a solution $j\in\mathbb{Z}$ --- i.e., $x_0\in\{1,...,p-1\}$
is a quadratic residue modulo $p$ --- whereas, if $p=2$, $x_1=x_2=0$.

\end{itemize}
\end{proposition}

Let us denote by $\eta\in\mathbb{Q}_{p}^*$ a normalized element which is not a square, i.e., any $p$-adic number $\eta$ such that
$|\eta|_p=1$ and $\eta\notin(\mathbb{Q}_{p}^*)^2$. The first condition implies that $\eta=\eta_0+\eta_1 \,p^1+\eta_2\,p^2+\cdots\neq p$.
For $p\neq 2$ we can choose, in particular, any of the --- exactly, $(p-1)/2$ --- quadratic non-residues ($\md p$) $\eta$ contained in
$\{2,\cdots,p-1\}$; e.g., for $p\equiv 3\pmod 8$, or for $p\equiv 5\pmod 8$, one can take $\eta=2$~\cite{vladimirov1994p}.
For $p=2$, one must take $\eta$ of the form $\eta=1+\eta_1 2+\eta_2 2^2+\cdots$, with $\eta_1,\eta_2\in\{0,1\}$ and $\eta_1\eta_2\neq 0$.
At this point, by Proposition~\ref{squareinratp}, it is clear that $p$ and $\eta p$ do not belong to $(\mathbb{Q}_{p}^*)^2$, as well.
In fact, $p=1\, p^1$ entails that $k=1$ does not satisfy the first condition therein. Similarly, for $\eta p=\eta_0 \,p^1+\eta_1 \,p^2+\cdots$.

Therefore, we have the following consequence of Proposition~\ref{squareinratp}:

\begin{corollary}\label{corrs2}
The following facts hold true:
\begin{itemize}

\item[\tt{(a)}] for $p\neq 2$, there is some $\eta\in\mathbb{Q}_p$ such that $\eta\not\in(\mathbb{Q}_{p}^{*})^2$ and $|\eta|_p=1$,
and $\eta p$ and $p$ are not squares too;

\item[\tt{(b)}]  every $\mu\in\{2,\eta,2\eta\}$ --- with $\eta=3,5,7$ --- is not a quadratic element of $\mathbb{Q}_2$; namely,
every $\mu\in\{2,3,5,6,7,10,14\}$ is not a square in $\mathbb{Q}_2$.

\end{itemize}
\end{corollary}

\begin{remark} \label{squacla}
The quotient group $\Qpa/(\Qpa)^2$ consists of four `square classes', for $p\neq 2$, with representatives $\{1,\eta,\,p,\,\eta p\}$,
where $\eta$ is any normalized non-quadratic element of $\Qp$; whereas, for $p=2$, it consists of eight square classes, with
representatives $\{1,2,3,5,6,7,10,14\}$ (or, equivalently, $\{\pm 1,\pm 2,\pm 3,\pm 6\}$)~\cite{vladimirov1994p,shimura2010arithmetic}.
Therefore, Corollary~\ref{corrs2} provides a classification of the square classes of $\Qpa$ different from $(\Qpa)^2$.
\end{remark}

We can now define and characterize the inequivalent quadratic extensions of $\mathbb{Q}_p$. We do not give a formal treatment here,
because the following definition closely mimics the definition of the field of complex numbers $\mathbb{C}$, regarded as a quadratic extension
of $\mathbb{R}$. For a formal approach see, e.g., Chapt.~{2} of~\cite{robert2013course} and Chapt.~{13} of~\cite{Dummit}.

\begin{definition}
Let $\mu\in\mathbb{Q}_p$ be a non-quadratic $p$-adic number, i.e., $\mu\notin(\mathbb{Q}_{p}^*)^2$. Introducing the symbol $\sqrt{\mu}$,
the quadratic extension $\Q$ of $\mathbb{Q}_p$ associated with $\mu$ is defined as the set
\begin{equation}
\Q\defi\big\{x+y\sqrt{\mu}\sep x,y\in\mathbb{Q}_p\big\}.
\end{equation}
Therefore, $\Q$ is a vector space over $\Qp$ of dimension $\big[\Q:\Qp\big]=2$.
\end{definition}

The elements of $\Q$ can be added and multiplied following the usual rules, with the additional convention that $(\sqrt{\mu})^2=\mu$.
Moreover, one observes that
\begin{equation}
x+y\sqrt{\mu}=0\quad\iff\quad x=y=0,
\end{equation}
and every non-null element $x+y\sqrt{\mu}$ admits a unique inverse, which is given by~\cite{vladimirov1994p}
\begin{equation}
(x+y\sqrt{\mu})^{-1}=\frac{x}{x^2-\mu y^2}-\frac{y}{x^2-\mu y^2}\hspace{0.5mm}\sqrt{\mu},
\end{equation}
where the denominator $x^2-\mu y^2$ is not zero, otherwise $\mu$ should be a square in $\mathbb{Q}_p$.

As in the complex case, on $\Q$ it is possible to define a \emph{conjugation}, which is given by
\begin{equation}\label{conj}
z=x+y\sqrt{\mu}\mapsto\bar{z}=x-y\sqrt{\mu},
\end{equation}
so that
\begin{equation}
z\overline{z}=x^2-\mu y^2\in\mathbb{Q}_p.
\end{equation}

For every $z=x+y\sqrt{\mu}\in\Q$, we call
\begin{equation}
x=\scj(z)\equiv(z+\overline{z})/2=\scj(\overline{z}) \quad \mbox{and} \quad
y=\acj(z)\equiv(z-\overline{z})/2\sqrt{\mu}=-\acj(\overline{z})
\end{equation}
the \emph{selfconjugate} and the \emph{anticonjugate} coordinate of $z$, respectively.

The $p$-adic absolute value $|\cdot|_p$ can be extended (in a unique way) to a valuation $|\cdot|_{p,\hspace{0.5mm}\mu}$ on $\Q$,
that is given explicitly by~\cite{Cassels,albeverio1996p,Anashin}
\begin{equation} \label{extval}
|z|_{p,\hspace{0.5mm}\mu}=\sqrt{|z\overline{z}|_p},
\end{equation}
and that, for the sake of simplicity, henceforth we will simply denote by $|\cdot|$.

Recalling Corollary~\ref{corrs2} and Remark~\ref{squacla}, it can be shown that there is a natural one-to-one correspondence between
the square classes of $\Qpa$ --- except the square class containing the identity, i.e., $(\Qpa)^2$ --- and the quadratic extensions of
$\Qp$ (up to isomorphisms):

\begin{proposition}[\cite{vladimirov1994p,shimura2010arithmetic}]\label{corfinal}

The quadratic extensions of $\mathbb{Q}_p$ are classified as follows:
\begin{itemize}

\item[\tt{(a)}] if $p\neq 2$, there are precisely three non-isomorphic quadratic extensions of $\mathbb{Q}_p$,
i.e., $\mathbb{Q}_p(\sqrt{\mu})$, with $\mu\in\{\eta,\,p,\,\eta p\}$, for any $\eta\notin(\mathbb{Q}_p^*)^2$ such that $|\eta|_p=1$;

\item[\tt{(b)}] if $p=2$,  there are precisely seven non-isomorphic quadratic extensions of $\mathbb{Q}_p$, i.e.,
$\mathbb{Q}_p(\sqrt{\mu})$, with $\mu\in\{2,\eta,2\eta\}$ --- for $\eta=3,5,7$ --- thus, with
$\mu=2,3,5,6,7,10,14$.

\end{itemize}
\end{proposition}

\begin{remark} \label{remrami}
Putting $\Qa\equiv\Q\setminus\{0\}$, the set $|\Qa|\defi\{|\alpha|\sep\alpha\in\Qa\}$ is a \emph{discrete subgroup} of
the multiplicative group of all positive reals, called the \emph{valuation group} of $\Qp$. By relation~\eqref{extval}, it is clear that
$|\Qa|\subset\{p^{k/2}\}_{k\in\mathbb{Z}}$. If we have a \emph{strict} containment --- i.e., if
$|\Qa|=\{p^{k}\}_{k\in\mathbb{Z}}=|\Qp|\setminus\{0\}$ --- then $\Q$ is said to be an \emph{unramified extension} of
$\Qp$; otherwise (i.e., if $|\Qa|=\{p^{k/2}\}_{k\in\mathbb{Z}}$), the quadratic extension $\Q$ is called (totally) \emph{ramified}.
It can be shown that, for every prime number $p$, there is --- up to isomorphisms --- exactly one unramified quadratic extension of $\Qp$.
E.g., $\mathbb{Q}_2(\sqrt{5})$ is the only unramified quadratic extension of $\mathbb{Q}_2$ (up to isomorphisms).
See Chapt.~{7} of~\cite{Cassels}, Chapt.~{5} of~\cite{gouvea1997p} and Chapt.~{2} of~\cite{robert2013course}.
\end{remark}


\section{$p$-adic Banach and Hilbert spaces}
\label{sec3}

Until today, there seems to be no universally accepted model of a Hilbert space over the field of $p$-adic
numbers or its quadratic extensions~\cite{albeverio1996p,khrennikov1990mathematical,kalisch1947p,aguayo2007non,albeverio1999non}.
Having in mind applications to quantum mechanics, in this section we introduce notions of p-adic Banach and Hilbert spaces
that are suitable for our purposes.


\subsection{$p$-adic Banach spaces}
\label{subsec3.1}

We start with the following:

\begin{definition}
 By a \emph{normed vector space} over $\Q$ we mean a pair $(X,\|\cdot\|)$, where $X$ is a vector space over $\Q$ and $\|\cdot\|$
 is an \emph{ultrametric norm} defined on $X$; i.e., a map $\|\cdot\|: X\rightarrow \mathbb{R}^+$ such that
\begin{itemize}
\item $\|x\|=0$ iff $x =0$;
\item $\|\alpha x\|=|\alpha|\;\|x\|$, for  all $\alpha\in\Q$ and all $x\in X$;
\item $\|x+y\|\leq \max(\|x\|,\|y\|)$, for all $x,y\in X$.
\end{itemize}
\end{definition}

\begin{remark}\label{remnor}
In the literature~\cite{robert2013course,khrennikov1990mathematical,narici71,diagana2016non,albeverio1999non},
the pair $(X,\|\cdot\|)$ would be called an \emph{ultrametric} (or non-Archimedean) normed space. Putting
\begin{equation}
\|X\|\defi\{\|x\|\sep x\in X\}\quad\mbox{and}\quad |\Q|\defi\{|\alpha|\sep \alpha\in\Q\},
\end{equation}
the sets $\|X\|$ and $|\Q|$ are not, in general, related by any inclusion relation. However, the inclusion $\|X\|\subset|\Q|$
entails the existence of unit vectors in $X$ and also implies the reverse inclusion, in such a way that, actually, $\|X\|=|\Q|$.
In the case where this condition is satisfied, $\|X\|\setminus\{0\}$ coincides with the valuation group $|\Q^\ast|$ (recall
Remark~\ref{remrami}). Since the valuation group $|\Q^\ast|$ is discrete, then, by Theorem~{3} in~\cite{natarajan2019sequence},
the field $\Q$ is \emph{spherically complete}, i.e., every nest of closed balls in $\Q$ has a non-empty intersection.
\end{remark}

\begin{remark}\label{resepa}
In the following, we will mainly deal with \emph{separable} ($p$-adic) normed and Banach spaces, and we will consider
separable ($p$-adic) Hilbert spaces only. In this regard, note that, since $\Q$ is separable, we do not need to use the
--- in this case equivalent --- notion of a normed space \emph{of countable type}~\cite{rooij1978non,Perez-Garcia}.
We will consider some non-separable ultrametric Banach spaces in Subsection~\ref{linops}.
\end{remark}

\begin{remark}\label{redomvec}
One can easily check that, given $x,y\in X$, $\|x\|>\|y\|\implies\|x+y\|=\|x\|$~\cite{narici71,natarajan2019sequence}.
\end{remark}

Any $p$-adic normed space is a (ultra-)metric space. Thus, it can be completed, so resulting into a \emph{$p$-adic Banach space}
(i.e., an ultrametric Banach space over $\Q$).

\begin{proposition}[\cite{schikhof2007ultrametric,rooij1978non,diagana2016non}]\label{sumlemma}
Let $(X,\|\cdot\|)$ be a $p$-adic Banach space. A series $\sum_ix_i$ in $X$ is convergent if and only if $\lim_i x_i=0$.
In particular, (regarding $\Q$ as a complete $p$-adic normed space) a series $\sum_i x_i$ in $\Q$ converges if and only if $\lim_{i}x_i=0$.
\end{proposition}

We now consider a class of $p$-adic Banach spaces that will be central for our purposes.

Let $I$ be a \emph{countable} index set (in the case where this set is finite, we will put $I=\{1,2,\ldots,n\}$, for some $n\in\mathbb{N}$;
otherwise we put $I=\mathbb{N}$), and let $X$ be a $p$-adic Banach space. We introduce the space $c_0(I,X)$ of \emph{zero-convergent}
--- in the case where $I=\mathbb{N}$ --- sequences in $X$:
\begin{equation}
c_0(I,X)\defi\big\{x=\{x_i\}_{i\in I} \sep x_i\in X,\;\lim_{i}\|x_i\|= 0\big\}.
\end{equation}
In particular, with $X=\Q$ (regarded as a one-dimensional vector space, endowed with the norm $|\cdot|$), we obtain the sequence space
\begin{equation}\label{eqt.15}
c_0(I,\Q)\defi\big\{x=\{x_i\}_{i\in I}\sep x_i\in\Q,\;\lim_i |x_i|=0\big\}.
\end{equation}

\begin{remark}
In order to include the case where $I$ is finite, here and in the following we set: $\lim_i x_i\equiv 0$ and $\lim_i\|x_i\|\equiv 0$, for $I$ finite.
\end{remark}

The space $c_0(I,X)$, endowed with the \emph{sup-norm} $\|\cdot\|_{\infty}$ defined by
\begin{equation}
\|x\|_\infty\defi\sup_{i\in I}\|x_i\|=\max_{i\in I}\|x_i\|,\quad \forall x\in c_0(I,X),
\end{equation}
is a (ultrametric) normed space over $\Q$. Moreover, it is possible to prove that $c_0(I,X)$ is complete  w.r.t.\ the sup-norm, and,
thus, the pair $(c_0(I,X), \|\cdot\|_\infty)$ is a $p$-adic Banach space~\cite{robert2013course}.

\begin{remark}\label{rem.3.6}
The $p$-adic Banach space $(c_0(I,\Q),\|\cdot\|_\infty)$ is a particular case of the ultrametric Banach space $c_0(I,\mathbb{K})$
--- where  $\mathbb{K}$ is a complete, non-trivially valued, non-Archimedean
field~\cite{schikhof2007ultrametric,albeverio1999non,natarajan2019sequence,natarajan2014introduction} and
\begin{equation}
c_0(I,\mathbb{K})\defi\big\{x=\{x_i\}_{i\in I}\sep x_i\in\mathbb{K},\;\lim_i |x_i|_\mathbb{K}=0\big\}
\end{equation}
--- endowed with the norm
\begin{equation}
\|x\|_{\infty}\defi\sup_{i\in I}|x_i|_{\mathbb{K}}=\max_{i\in I}|x_i|_\mathbb{K}.
\end{equation}
\end{remark}

Next, we introduce the notion of norm-orthogonal system of vectors in a $p$-adic normed
space~\cite{schikhof2007ultrametric,rooij1978non,Perez-Garcia,diagana2016non,narici2005non}:

\begin{definition}\label{def.3.6}
Let $(X,\|\cdot\|)$ be a $p$-adic normed space. Two vectors $x,y\in X$ are said to be (mutually) \emph{norm-orthogonal} if,
for every $\alpha\in\Q$, $\|x\|\leq\|x+\alpha y\|$, or, equivalently, if $\|\alpha x+\beta y\|=\max\{\|\alpha x\|,\|\beta y\|\}$,
for all $\alpha,\beta\in\Q$. More generally, a finite set $\{x_1,\dots,x_n\}$ in $X$ is said to be \emph{norm-orthogonal} if
\begin{equation}
\left\|\sum_{i=1}^n\alpha_i x_i\right\|=\max_i|\alpha_i|\,\|x_i\|,\;\;\text{for all}\;\; \{\alpha_1,\dots,\alpha_n\}\subset\Q.
\end{equation}
 An arbitrary set $\bas\subset X$ is called \emph{norm-orthogonal} if every finite subset of $\bas$ is. In particular,
 a norm-orthogonal set $\bas\subset X$ is said to be \emph{normal} if $\|x\|=1$, for all $x\in\bas$.
 \end{definition}

\begin{definition}\label{normbasis}
Let $(X,\|\cdot\|)$ be a (separable) $p$-adic Banach space. A countable subset $\bas$ of $X\setminus\{0\}$ is said to be a
\emph{norm-orthogonal (normal)} basis if
\begin{enumerate}[label=\tt{(B\arabic*)}]

\item $\bas$ is a norm-orthogonal (normal) set;

\item for each $x\in X$, there exists a map $c_x:\bas\rightarrow \Q$ such that
\begin{equation}
x=\sum_{b\in \bas}c_x(b)\, b.
\end{equation}

\end{enumerate}
\end{definition}

Note that a normal basis in $X$ is, in particular, a (normalized) Schauder basis~\cite{bosch1984non,narici2005non};
also see point~\ref{propiii} in the proposition below. We will denote such a basis by
$\boldsymbol{e}\equiv\{e_i\}_{i\in I}$.

\begin{proposition}[\cite{schikhof2007ultrametric}]\label{prop.3.29}
Let $(X,\|\cdot\|)$ be a $p$-adic Banach space, let $\{e_i\}_{i\in I}$ be a normal basis and let $x=\sum_{i\in I}\alpha_ie_i$,
with $\alpha_1,\alpha_2,\dots\in\Q$. Then, the following facts hold:
\begin{enumerate}[label=\rm{(\roman*)}]

\item in the case where $I=\mathbb{N}$, $\lim_i \alpha_i=0$;

\item $\|x\|=\max_{i\in I}|\alpha_i|$;

\item \label{propiii} if, for some $\lambda_1,\lambda_2,\dots\in\Q$,
$\sum_{i\in I}\lambda_ie_i=x$, then $\alpha_i=\lambda_i$, $\forall i\in I$;
namely the expansion of every vector in $X$ w.r.t.\ the basis $\{e_i\}_{i\in I}$
is unique.

\end{enumerate}
\end{proposition}

\begin{remark}
By the unconditional convergence of a series in  a $p$-adic Banach space, every permutation of a normal basis
is a normal basis too.
\end{remark}

\begin{theorem}\label{sepban}
Let $(X,\|\cdot\|)$ be a (separable) $p$-adic Banach space over $\Q$. Then, it admits a norm-orthogonal basis.
Moreover, $X$ admits a normal basis if and only if $\|X\|=|\Q|$ {\rm (equivalently, iff $\|X\|\subset|\Q|$;
see Remark~\ref{remnor})}. If the last condition is satisfied, the mapping
\begin{equation}\label{eqiso}
c_0(I,\Q)\ni\{x_i\}_{i\in I}\mapsto \sum_{i\in I}x_ie_i\in X
\end{equation}
defines a surjective isometry of $c_0(I,\Q)$ onto $X$.
\end{theorem}

\begin{proof}
The first assertion of the theorem follows from Theorem $50.8$ in~\cite{schikhof2007ultrametric},
taking into account the fact that every finite extension of $\mathbb{Q}_p$ is locally compact.
Alternatively, one can use the fact that $\Q$ is spherically complete, $X$ is separable (equivalently, of countable type)
and Lemma $5.5$ in~\cite{rooij1978non}. The second assertion is (the separable version of) the Monna-Fleischer Theorem;
see Sect.~$4.4.5$ in~\cite{robert2013course}, taking into account the fact that the valuation group $|\Q^*|$ is
a discrete subgroup of the multiplicative group of all positive reals. For the final assertion of the theorem,
see Proposition~$3$ in Section~$4.4.2$ of~\cite{robert2013course}.
\end{proof}

In the light of the previous result, we set the following:

\begin{definition}
We say that a $p$-adic Banach space $(X,\|\cdot\|)$ over $\Q$ is \emph{normal} if $\|X\|=|\Q|$;
equivalently, if it admits a normal basis.
\end{definition}

\begin{remark}
Let $(X, \|\cdot\|)$ be a (separable) $p$-adic Banach space. We define $\dim(X)$ as the countable cardinality
of any norm-orthogonal basis in $X$. In the finite-dimensional case, $\dim(X)$ coincides with the algebraic dimension
of $X$. In the infinite-dimensional case, we simply put $\dim(X)=\infty$.
\end{remark}


\subsection{Inner product $p$-adic Banach spaces}
\label{ipbs}

We now aim at introducing a suitable notion of \emph{$p$-adic Hilbert space} over $\Q$.
The first step is to provide a convenient notion of $p$-adic inner product (Banach) space:

\begin{definition}\label{innprodu}
Let $(X, \|\cdot\|)$ be a $p$-adic Banach space over $\Q$. By a \emph{non-Archimedean inner product} on $X$ we mean a map
$\scap: X\times X\rightarrow \Q$ such that
\begin{itemize}

\item[(i)] $\scap$ is a \emph{sesquilinear form}, i.e., it is linear in its second argument and conjugate-linear
in its first argument (w.r.t.\ the conjugation in $\Q$ introduced in Section~\ref{sec2});

\item[(ii)] $\scap$ is \emph{Hermitian}, i.e., $\langle x,y\rangle=\overline{\langle y,x\rangle}$;

\item[(iii)] the \emph{Cauchy-Schwarz inequality} holds, i.e., $|\langle x,y\rangle|\leq \|x\|\;\|y\|$.

\end{itemize}
We call the triple $(X, \|\cdot\|, \scap)$, where $\scap$ is a non-Archimedean inner product,
an \textit{inner product $p$-adic Banach space}.

Given an inner product $p$-adic Banach space $(X, \|\cdot\|, \scap)$, we say that $\scap$
is \emph{non-degenerate} if, moreover, $\langle x,y\rangle=0$, for all $y\in X$, implies that $x=0$.
\end{definition}

\begin{remark}
Note that here, in general, $\|x\|\neq \sqrt{|\langle x,x\rangle|}$.
\end{remark}

\begin{remark}\label{rem.cont}
It is worth noting that the Cauchy-Schwarz inequality immediately implies that a non-Archimedean inner product $\scap$
on a $p$-adic Banach space is \emph{continuous} w.r.t.\ both its arguments (separately), and also \emph{jointly} continuous (i.e., as a map
from $X\times X,$ endowed with the product topology, into $\Q$), where the topology on $X$ is the one induced by the norm.
\end{remark}

\begin{example}\label{canoinner}
 Let us  provide an example of an inner product $p$-adic Banach space. Given a normal $p$-adic Banach space $X$,
 let us consider the (non-degenerate, Hermitian) sesquilinear form defined by
\begin{equation}\label{defsp}
X\times X\ni (x,y)\mapsto\langle x,y\rangle\defi \sum_{i\in I}\overline{x_i}\, y_i,
\end{equation}
where, for some \emph{normal basis} $\{e_i\}_{i\in I}$, $x=\sum_{i\in I} x_ie_i$ and $y=\sum_{i\in I} y_ie_i$. Clearly, we have that
\begin{align}
|\langle x,y\rangle|=\bigg{|}\sum_{i\in I}\overline{x_i}\, y_i\bigg{|}\leq \max_{i\in I}|x_i|\,|y_i|\leq
\max_{i\in I} |x_i|\, \max_{j\in I}|y_j|=\|x\|_{\infty}\|y\|_{\infty},
\end{align}
i.e., the Cauchy-Schwarz inequality is satisfied. We call this inner product the \textit{canonical inner product} in $X$
associated with the normal basis $\{e_i\}_{i\in I}$.
\end{example}

\begin{remark}\label{nulvec}
A normal inner product $p$-adic Banach space $(X,\|\cdot\|,\scap)$ --- with $\dim(X)\geq 2$,
and even assuming that $\scap$ is non-degenerate --- may contain \emph{isotropic vectors},
i.e., nonzero vectors $x$ such that $\langle x,x\rangle=0$. Let us construct an example of such a vector.
Suppose that $\scap$ is the canonical inner product, associated with a normal basis $\{e_i\}_{i\in I}$,
of Example \ref{canoinner}. For $p\equiv 1\pmod 4$ and for a non-quadratic element $\mu$ of $\mathbb{Q}_p$,
assuming that $X$ is a vector space over $\Q$ and taking into account
the fact that $-1$ is a square in $\mathbb{Q}_p$, let $x\in X$ be given by
$x=\big(\alpha+\beta\sqrt{\mu}\big)e_1 + \sqrt{-1}^p\big(\gamma+\delta\sqrt{\mu}\big)e_2$,
where $\alpha,\beta,\gamma,\delta\in\Qp$ --- with $|\alpha|+|\beta|+|\gamma|+|\delta|\neq 0$
and $\big(\alpha^2-\gamma^2\big)-\mu\big(\beta^2-\delta^2\big)=0$ (e.g., with $\alpha^2=\gamma^2\neq 0$ and $\beta^2=\delta^2$) ---
and $\sqrt{-1}^p\in\Qp$ is any of the two $p$-adic square roots of $-1$; i.e., the components $x_i\in\Q$ of the vector $x$
--- w.r.t.\ the fixed normal basis $\{e_i\}_{i\in I}$ of $X$ --- are given by  $x_1=\alpha+\beta\sqrt{\mu}$,
$x_2=\sqrt{-1}^p\big(\gamma+\delta\sqrt{\mu}\big)$, and $x_i=0$, for $i\geq 3$. Then, we have that $x\neq 0$ and
$\langle x,x\rangle=\overline{x_1}\,x_1+\overline{x_2}\,x_2=\big(\alpha^2-\mu\beta^2\big)-\big(\gamma^2-\mu\delta^2\big)=0$,
namely, $x$ is a isotropic vector.
\end{remark}

In an inner product $p$-adic Banach space $(X, \|\cdot\|, \scap)$, we have two (distinct) natural notions of orthogonality:
the --- previously introduced --- norm-orthogonality and the \emph{inner-product-orthogonality} (IP-orthogonality). Clearly, we say that
two vectors $x,y\in X$ are IP-orthogonal --- in symbols, $x\perp y$ --- if $\langle x,y\rangle=0$.

\begin{definition}\label{orbasis}
Let $(X, \|\cdot\|, \scap)$ be a normal inner product $p$-adic Banach space. A (finite or denumerable) sequence
of vectors $\Phi\equiv\{\phi_i\}_{i\in I}$ is said to be an \emph{orthonormal basis} in $X$, if the following conditions hold:
\begin{enumerate}[label=\tt{(O\arabic*)}]

\item $\Phi$ is a normal basis in $(X, \|\cdot\|)$;

\item $\langle \phi_i,\phi_j\rangle=\delta_{ij}$, for all $i,j\in I$.

\end{enumerate}
\end{definition}

Let $(X,\|\cdot\|, \scap)$ be a normal inner product $p$-adic Banach space over $\Q$, and suppose that $X$ admits an
orthonormal basis $\Phi\equiv\{\phi_i\}_{i\in I}$. By providing explicit examples, we now show that the construction of a \emph{new}
orthonormal basis $\Psi\equiv\{\psi_i\}_{i\in I}$ in $X$, starting from the given orthonormal basis $\{\phi_i\}_{i\in I}$,
is not a trivial task as it would be, say, in an ordinary separable complex Hilbert space.

Assume, at first, that $\dim(X)=2$. In order to construct a new orthonormal basis in $X$, suppose that $z\in\Q$ is such that
$z\overline{z}=2$ and $|z|=\sqrt{|z\overline{z}|_p}=1$ ($\hspace{-1.2mm}\iff |2|_p=1\iff p\neq 2$).

\begin{example}\label{ex.02}
Let us give a few explicit examples where, for $p\neq 2$, the condition $z\overline{z}=2$ is realized.
\begin{enumerate}

\item Let us take $p=3$ and $\mu=5$ (note that $5=2+1\cdot 3$ is a quadratic non-residue $\md 3$).
Then, $z=\sqrt{7}^3-\sqrt{5}$ --- where $\sqrt{7}^3$ is one of the $3$-adic square roots of $7=1+2\cdot 3\in(\Qtra)^2$, i.e.,
\begin{equation}
\sqrt{7}^3=
\begin{cases}
1+1\cdot 3+1\cdot 3^2+0\cdot 3^3+2\cdot 3^4+\cdots\\
2+1\cdot 3+1\cdot 3^2+2\cdot 3^3+ 0\cdot 3^4+\cdots
\end{cases}
\end{equation}
--- verifies the condition $z\overline{z}=7-5=2$, and $|z|=\sqrt{|2|_3}=1$.

\item Let $p=3$ and $\mu=2\not\in(\Qtra)^2$. Then, $z=2+\sqrt{2}$ is such that $z\overline{z}=4-2=2$, where $|2|_3=1$.

\item Let $p=5$ and $\mu=3$ ($3$ is a quadratic non-residue $\md 5$). We set $z=\sqrt{29}^5+3\sqrt{3}$, where
$\sqrt{29}^5$ is one of the $5$-adic square roots of $29=4+0\cdot 5+1\cdot 5^2\in(\Qca)^2$, i.e.,
\begin{equation}
\sqrt{29}^5=
\begin{cases}
2+0\cdot 5+4\cdot 5^2+3\cdot 5^3+4\cdot 5^4+\cdots\\
3+4\cdot 5+0\cdot 5^2+1\cdot 5^3+ 0\cdot 5^4+\cdots
\end{cases}.
\end{equation}
Then, we have that $z\overline{z}=29-27=2$, where $|2|_5=1$.

\item Let us take $p=7$ (and, say, $\mu=7$). We set $z=\sqrt{2}^7$ ($2\equiv3^2\pmod7$), with
\begin{equation}
\sqrt{2}^7=
\begin{cases}
3+1\cdot 7+2\cdot 7^2+6\cdot 7^3+1\cdot 7^4+\cdots\\
4+5\cdot 7+4\cdot 7^2+0\cdot 7^3+ 5\cdot 7^4+\cdots\;,
\end{cases}
\end{equation}
so that $z\overline{z}=2$, where $|2|_7=1$.
\end{enumerate}
\end{example}

Now, given an orthonormal basis $\{\phi_1,\phi_2\}$ in $X$ ($\dim(X)=2$, $p\neq 2$), the set $\{\psi_1,\psi_2\}$, where
\begin{equation}\label{eq.03}
\psi_1=\frac{1}{z}(\phi_1+\phi_2),\;\;\;\;\psi_2=\frac{1}{z}(\phi_1-\phi_2),
\end{equation}
with $z\in\Q$ being chosen as above, is an orthonormal basis for $X$. Indeed, we have that $\langle\psi_1,\psi_2\rangle=0$ and
$\langle\psi_1,\psi_1\rangle=2/z\overline{z}=1=\langle\psi_2,\psi_2\rangle$. Moreover: $\|\psi_1\|=\|\psi_2\|=1/|z|=1$. \emph{But},
it remains to show that $\{\psi_1,\psi_2\}$ is a \emph{norm-orthogonal} set, as well.  To clarify this point, let us first prove the following:

\begin{fact}\label{lem.1}
If $x_1,x_2\in\Q$ and $p\neq 2$ --- equivalently, $|2|_p=1$ --- then
\begin{equation}\label{eq.05}
\max\{|x_1|,|x_2|\}=\max\{|x_1+x_2|, |x_1-x_2|\}.
\end{equation}
\end{fact}

\begin{proof}
We can suppose, without loss of generality, that $|x_1|\geq|x_2|$. Then, we have:
\begin{equation}
|x_1|=|2x_1|=|(x_1+x_2)+(x_1-x_2)|\leq \max\{|x_1+x_2|, |x_1-x_2|\}\leq \max\{|x_1|,|x_2|\}=|x_1|,
\end{equation}
so that~\eqref{eq.05} holds true.
\end{proof}

We can now conclude that the vectors in~\eqref{eq.03} form an orthonormal set. Indeed, let $x=x_1\psi_1+x_2\psi_2$ be any vector
in $X=\textrm{span}\{\psi_1,\psi_2\}=\mathrm{span}\{\phi_1,\phi_2\}$. We have:
\begin{equation}
x=\frac{1}{z}(x_1+x_2)\phi_1+\frac{1}{z}(x_1-x_2)\phi_2.
\end{equation}
Since $\{\phi_1,\phi_2\}$ is a norm-orthogonal set and $|z|=1$, we have that
\begin{equation}
\|x\|=\max\left\{\frac{1}{|z|}|x_1+x_2|,\frac{1}{|z|}|x_1-x_2|\right\}=\max\{|x_1+x_2|, |x_1-x_2|\}=\max\{|x_1|,|x_2|\},
\end{equation}
where the last equality holds by~\eqref{eq.05}. Therefore, the set
\begin{equation}
\big\{\psi_1=z^{-1}(\phi_1+\phi_2),\,\psi_2=z^{-1}(\phi_1-\phi_2)\big\}
\end{equation}
is norm-orthogonal too. Clearly, if $\dim(X)>2$, then
\begin{equation}
\psi_1=\frac{1}{z}(\phi_1+\phi_2),\quad\psi_2=\frac{1}{z}(\phi_1-\phi_2),\quad\psi_3=\phi_3,\quad\ldots
\end{equation}
is again an orthonormal basis in $X$.

Let us now consider the case where $\dim(X)=\infty$ and $p\neq 2$. If $z\in\Q$ is such that $z\overline{z}=2$,
then $\Psi\equiv\{\psi_1,\psi_2,\psi_3,\ldots\}$ --- with
\begin{equation}\label{eq.15}
\psi_1=\frac{1}{z}(\phi_1+\phi_2),\quad\psi_2=\frac{1}{z}(\phi_1-\phi_2),\quad\psi_3=
\frac{1}{z}(\phi_3+\phi_4),\quad\psi_4=\frac{1}{z}(\phi_3-\phi_4),\quad\ldots
\end{equation}
--- is an orthonormal basis. Indeed, if
\begin{equation}
x=\sum_{j\in\mathbb{N}}x_j\psi_j=\frac{1}{z}\sum_{j\,\,\text{odd}}((x_j+x_{j+1})\phi_j+(x_j-x_{j+1})\phi_{j+1}),
\end{equation}
then, since $|z|=1$,
\begin{equation}
\|x\|=\max_{j\,\,\text{odd}}\{|x_j+x_{j+1}|,\,|x_j-x_{j+1}|\}.
\end{equation}
 Hence, by~\eqref{eq.05},
\begin{equation}
\|x\|=\max_{j\; \mathrm{odd}}\{|x_j|,\,|x_{j+1}|\}=\max_{j\in\mathbb{N}}\{|x_j|\},
\end{equation}
so that $\Psi$ is a norm-orthogonal (and IP-orthogonal) set, and an orthonormal basis in $X$,
because $\mathrm{span}\{\psi_j\}_{j\in\nat}=\mathrm{span}\{\phi_j\}_{j\in\nat}$ so that
$\overline{\mathrm{span}\{\psi_j\}_{j\in\nat}}^{\,\|\cdot\|}=\overline{\mathrm{span}\{\phi_j\}_{j\in\nat}}^{\,\|\cdot\|}=X$.


\subsection{$p$-adic Hilbert spaces}

To the best of our knowledge, the existence of an orthonormal basis in a generic inner product $p$-adic Banach space is
not guaranteed (even assuming that the inner product is non-degenerate). Therefore, it is natural to set the following:

\begin{definition}
Let $(X, \|\cdot\|,\scap)$ be an inner product $p$-adic Banach space (over $\Q$). We say that $X$
is a \emph{$p$-adic Hilbert space} if it admits an orthonormal basis $\{\phi_i\}_{i\in I}$
(in the sense of Definition~\ref{orbasis}). We will typically denote (the carrier space of) a $p$-adic Hilbert space by $\Hp$.
\end{definition}

Let $x\in \Hp$, and let $\Phi\equiv \{\phi_i\}_{i\in I}$ be an orthonormal basis in $\Hp$. By the first condition in Definition~\ref{orbasis},
we can express $x$  --- in a unique way --- as $x=\sum_{i\in I}x_i\phi_i$, for some set of coefficients $\{x_i\}_{i\in I}$ in $\Q$.
Moreover, by the second condition in the same definition, and taking into account the continuity, w.r.t.\ each of its arguments,
of the non-Archimedean inner product (see Remark~\ref{rem.cont}), we have that
\begin{equation}
\langle \phi_j,x\rangle=\langle\phi_j,{\textstyle\sum_{i\in I}}x_i\phi_i\rangle=
\sum_{i\in I}x_i\,\langle\phi_j,\phi_i\rangle=x_j,\quad\forall j\in I.
\end{equation}
 Thus, we see that any $x\in\Hp$ is expressed --- w.r.t.\ the fixed orthonormal basis $\Phi$ in $\Hp$ --- as
\begin{equation}\label{eq.Hs1}
x=\sum_{i\in I}\langle \phi_i,x\rangle\phi_i,
\end{equation}
from which we deduce the \emph{non-Archimedean Parseval} identity
\begin{equation}
\|x\|=\max_{i\in I} |\langle \phi_i,x\rangle|.
\end{equation}

\begin{proposition}\label{nondeginpro}
Let $\Hp$ be a $p$-adic Hilbert space over $\Q$. Then, $\Hp$ is normal --- i.e., $\|\Hp\|=|\Q|$ --- and the non-Archimedean inner product
$\scap$ defined on it is non-degenerate, i.e.,
\begin{equation}
\langle x,y\rangle=0,\;\;\forall y\in\Hp\;\;\implies\;\; x=0.
\end{equation}
\end{proposition}

\begin{proof}
Since $\Hp$ admits a (ortho-)normal basis, then, by the second assertion of Theorem~\ref{sepban}, $\|\Hp\|=|\Q|$.
Let $\Phi\equiv \{\phi_i\}_{i\in I}$ be an orthonormal basis of $\Hp$. If $\langle x,y\rangle=0$, $\forall y\in \Hp$,
it must be true that $\langle x,\phi_i\rangle=0$, $\forall i\in I$. But then, by~\eqref{eq.Hs1} and by the uniqueness
of the decomposition of a vector w.r.t.\ a (ortho-)normal basis in $\Hp$, it follows that $x=0$.
\end{proof}

\begin{example}\label{examp324}
Let us consider the normal $p$-adic Banach space $(c_0(I,\Q), \|\cdot\|_{\infty})$ (see Subsection~\ref{subsec3.1}),
where $\|x\|_{\infty}=\max_{i\in I}|x_i|$ ($x=\{x_i\}_{i\in I}$). Let $\boldsymbol{e}\equiv\{e_i\}_{i\in I}$ be the set
of all sequences in $c_0(I,\Q)$ whose elements are of the form
\begin{equation}
e_1=(1,0,0,\cdots),\quad e_2=(0,1,0,\cdots),\quad e_3=(0,0,1,\cdots),\quad\ldots\quad .
\end{equation}
Clearly, the set $\{e_i\}_{i\in I}$ is a normal basis --- the \emph{standard basis} --- and $\dim(c_0(I,\Q))=\card(I)$.
Let us endow $(c_0(I,\Q), \|\cdot\|_{\infty})$ with the \emph{canonical inner product}, associated with $\{e_i\}_{i\in I}$,
introduced in Example~\ref{canoinner}:
\begin{equation}\label{eq.Hs2}
c_0(I,\Q)\times c_0(I,\Q)\ni (x,y)\mapsto\langle x,y\rangle\defi \sum_{i\in I}\overline{x_i}\, y_i\in \Q.
\end{equation}
By construction, we have that
\begin{equation}
\langle e_i,e_j\rangle=\delta_{ij},\quad \forall i,j\in I,
\end{equation}
i.e., $\{e_i\}_{i\in I}$ is an orthonormal basis for $c_0(I,\Q)$. Therefore, the $p$-adic Banach space $(c_0(I,\Q), \|\cdot\|_{\infty})$,
endowed with the inner product in~\eqref{eq.Hs2}, and admitting the orthonormal basis $\{e_i\}_{i\in I}$, is a $p$-adic Hilbert space.
In the literature~\cite{khrennikov1990mathematical,diagana2016non}, this $p$-adic Hilbert space is sometimes called \emph{coordinate $p$-adic Hilbert space},
and denoted by $\mathbb{H}(I)$. More generally, given a normal $p$-adic Banach space $X$ and a normal basis $\{e_i\}_{i\in I}$ in $X$,
we can endow this space with the sesquilinear form defined by~\eqref{defsp}, so that $\{e_i\}_{i\in I}$ becomes an orthonormal basis and $X$ a $p$-adic Hilbert space.
\end{example}

In the light of Remark~\ref{nulvec} about the existence of isotropic vectors in an inner product $p$-adic Banach space, we set the following:

\begin{definition}
For every quadratic extensions $\Q$ of $\mathbb{Q}_p$, we define the \emph{isotropy index} $\isin\in\mathbb{N}$ as
\begin{equation}
\isin\defi \min\big\{\card(\supp(x))\sep x\in c_0(\nat,\Q)\setminus\{0\},\; \langle x,x\rangle=0\big\},
\end{equation}
where $\supp(x)\subset\nat$ denotes the \emph{support} of the sequence $x=\{x_i\}_{i\in\nat}\in c_0(\nat,\Q)$; namely,
\begin{equation}
\supp(x)\defi\{i\in\nat\sep x_i\neq 0\}.
\end{equation}
\end{definition}

\begin{proposition}\label{prop.326}
Given a quadratic extension $\Q$ of $\mathbb{Q}_p$, $\isin\in \{2,3\}$.
\end{proposition}

\begin{proof}
It is clear that $\isin>1$. Let us show that, in particular, either $\isin=2$ or $\isin=3$. Assume, at first, that $p\neq 2$.
By Lemma $54.6$ in~\cite{theclassfields2007}, there exist numbers $\alpha,\beta,\gamma\in \mathbb{Q}_p$ such that $\alpha\neq 0\neq\beta$
and $\alpha^2+\beta^2+\gamma^2=0$. Therefore, putting $x=\alpha e_1+\beta e_2+\gamma e_3$, where $\{e_i\}_{i\in\nat}$ is the standard basis
in $c_0(\mathbb{N},\Q)$, we have that $\langle x,x\rangle=\alpha^2+\beta^2+\gamma^2=0$; i.e., $\isin\in\{2,3\}$. For $p=2$,
the same result can be achieved by means of a direct calculation. E.g., for $\mu=2$, we can put $x=(1+\sqrt{2})e_1+e_2$
(so that $\langle x,x\rangle=(1-2)+1=0$). For $\mu=3$, we can take $x=(1+\sqrt{3})e_1+e_2+e_3$. For $\mu=5$, we take
$x=(1+\sqrt{5})e_1+2e_2$. The remaining cases ($p=2$ and $\mu=6,7,10,14$) are similar and, once again, it turns out that $\isin \in\{2,3\}$.
\end{proof}

It is worth observing that the coordinate $p$-adic Hilbert space $\mathbb{H}(I)$ (Example~\ref{examp324})
plays a role analogous to the role played by $\ell^2(I)$ for the (separable) complex Hilbert spaces: There exists an isomorphism
of  $p$-adic Hilbert spaces between  $\Hp$ and $\mathbb{H}(I)$, where $\dim(\Hp)=\card(I)$. Here, we are assuming the following:

\begin{definition} \label{ishil}
Given $p$-adic Hilbert spaces $(\Hp, \|\cdot\|, \scap)$ and $(\mathcal{K}, \|\cdot\|, \scap)$ over the same quadratic extension $\Q$ of $\Qp$,
a linear map $W\colon \Hp\rightarrow \mathcal{K}$ is called an \emph{isomorphism of $p$-adic Hilbert spaces}
(an \emph{automorphism}, in the case where $\Hp=\mathcal{K}$) if
\begin{enumerate}[label=\tt{(I\arabic*)}]

\item $W$ is an isometry: $\|Wx\|=\|x\|$, $\forall x\in\Hp$;

\item $W$ is surjective: $W\Hp=\mathcal{K}$;

\item $\langle Wx, Wy\rangle=\langle x,y\rangle$, $\forall x,y\in\Hp$.

\end{enumerate}
Otherwise stated, $W$ is an isomorphism of $p$-adic Banach spaces --- a surjective isometry --- preserving the inner product.
\end{definition}

Then, let $\Hp$ be a $p$-adic Hilbert space, and let $\Phi\equiv\{\phi_i\}_{i\in I}$ be an orthonormal basis in $\Hp$.
By Theorem~\ref{sepban}, the map
\begin{equation}\label{isohilb}
\Wf:\Hp\ni x=\sum_{i\in I}\langle\phi_i,x\rangle\,\phi_i\mapsto\breve{x}\equiv\{\langle\phi_i,x\rangle\}_{i\in I}\in \mathbb{H}(I)
\end{equation}
is an isomorphism of $\Hp$ onto $\mathbb{H}(I)$, since it is a surjective isometry and, by the continuity of the inner product (see Remark~\ref{rem.cont}),
\begin{align}
\langle x,y\rangle&=\langle{\textstyle\sum_{i\in I}}\langle\phi_i,x\rangle\,\phi_i,
{\textstyle\sum_{j\in I}}\langle\phi_j,y\rangle\,\phi_j\rangle
\nonumber\\
&=\sum_{i\in I}\sum_{j\in I}\overline{\langle\phi_i,x\rangle}\,\langle\phi_j,y\rangle\,\langle \phi_i,\phi_j\rangle
\nonumber\\ \label{reside}
&=\sum_{i\in I}\overline{\langle\phi_i,x\rangle}\,\langle\phi_i,y\rangle=\langle\breve{x},\breve{y}\rangle,
\end{align}
i.e., the inner product is preserved. Note that, $\Wf\phi_i=e_i$, $\forall i\in I$, where $\{e_i\}_{i\in I}$
is the standard basis in $\mathbb{H}(I)$, and $\dim(\Hp)=\card(I)=\dim(\mathbb{H}(I))$. Therefore,
two $p$-adic Hilbert spaces $\Hp$ and $\mathcal{K}$, over the same quadratic extension of $\mathbb{Q}_p$,
are isomorphic iff $\dim(\Hp)=\dim(\mathcal{K})$, in complete analogy w.r.t.\ separable complex Hilbert spaces.

However, the analogies between the complex and the $p$-adic Hilbert spaces cannot be pursued too far.
E.g., in a complex Hilbert space the norm stems directly from the scalar product, and
the closed subspaces --- endowed with the subset inclusion and with the orthogonal complementation ---
form an (orthomodular) orthocomplemented lattice~\cite{beltrametti81}; in particular,
a relation of the form~\eqref{eq.orthcon} below holds true, whereas, for a $p$-adic Hilbert space,
we have the following:
\begin{proposition}
Let $\Hp$ be a $p$-adic Hilbert space over $\Q$, with $\dim(\Hp)\geq \isin\in\{2,3\}$. Then, the mapping
\begin{equation}
\Hp\ni x\mapsto\sqrt{|\langle x,x\rangle|}\in\mathbb{R}^+
\end{equation}
is \emph{not} a norm. Moreover, in general, it is \emph{not} true that
\begin{equation}\label{eq.orthcon}
\emptyset\neq \cV\subset \Hp\quad\text{and}\quad \cV=\cV^{\perp\perp}\quad\implies\quad \cV+\cV^\perp=\Hp,
\end{equation}
where $\cV^\perp\defi\{x\in \Hp\sep \langle x,y\rangle=0,\, \forall y\in \cV\}$; i.e., there exists some non-empty subset
$\mathcal{V}$ of $\Hp$ such that $\mathcal{V}=\mathcal{V}^{\perp\perp}$ and violating the relation $\cV+\cV^\perp=\Hp$.
\end{proposition}
(Note: For every non-empty subset $\cV$ of $\Hp$, $\cV^\perp$ is a norm-closed linear subspace of $\Hp$.)

\begin{proof}
For every orthonormal basis $\Phi=\{\phi_1,\phi_2,\ldots\}$ in $\Hp$, the mapping
\begin{equation}
\Hp\ni x=\sum_k x_k\phi_k \mapsto \breve{x}=\{\breve{x}_1, \breve{x}_2,\ldots\}\in c_0(\mathbb{N},\Q),
\end{equation}
--- where $\breve{x}_k=x_k$, for $k\leq \dim(\Hp)$, and $\breve{x}_k=0$, otherwise --- is an isometry preserving
the inner product ($c_0(\mathbb{N},\Q)$ being endowed with the canonical inner product associated with its standard basis).
Therefore, by Proposition~\ref{prop.326}, if $\dim(\Hp)\geq \isin$, then $\Hp$ admits (nonzero) isotropic vectors:
$\exists x\in \Hp$, $x\neq 0$, such that $\sqrt{|\langle x,x\rangle|}=0$. This observation proves the first assertion.
To prove the second one, we now provide a counterexample to implication~\eqref{eq.orthcon}. Let us first show that,
for every $0\neq x\in\Hp$, $\Q\, x=(\Q\, x)^{\perp\perp}$. In fact, since $\scap$ is non-degenerate,
there is some $y\in\Hp$ such that $\langle x,y\rangle\neq 0$. Then, for every $z\in\Hp$, the vector
$\tilde{z}=z-\langle x,z\rangle\,\langle x,y\rangle^{-1}y$ is IP-orthogonal to $x$:
\begin{equation}
\langle x,\tilde{z}\rangle=\langle x,z\rangle-\langle x,z\rangle\langle x,y\rangle^{-1}\langle x,y\rangle=0.
\end{equation}
Therefore, for every $z\in\Hp$, $\tilde{z}\in (\Q\, x)^\perp$. Thus, given any $w\in (\Q\, x)^{\perp\perp}$,
we have that $\langle w,\tilde{z}\rangle=0$, $\forall z\in\Hp$; i.e.,
\begin{equation}
\langle w, z-\alpha y\rangle=0,\quad \forall z\in\Hp,
\end{equation}
where $\alpha=\langle x,z\rangle\,\langle x,y\rangle^{-1}$. This condition is equivalent to
\begin{equation}
\langle w-\overline{\beta}x, z\rangle=0,\quad \forall z\in\Hp,
\end{equation}
where $\beta=\langle x,y\rangle^{-1}\langle w,y\rangle$ ($\alpha\langle w,y\rangle=\beta\langle x,z\rangle$).
Hence, as the inner product is non-degenerate, $w-\overline{\beta}x=0$; i.e.,
$w=\overline{\beta}x\in\Q\, x$, so that $(\Q\, x)^{\perp\perp}\subset\Q\, x$.
But, for any $\emptyset \neq \cV\subset \Hp$, it is always true that $\cV\subset \cV^{\perp\perp}$;
therefore, actually, $\Q\, x=(\Q\, x)^{\perp\perp}$.
At this point, observe that, given $0\neq x\in \Hp$, we have:
\begin{align}
\langle x,x\rangle=0\;&\implies\;\Q\, x\subset (\Q\, x)^{\perp}
\nonumber\\
&\implies\; \Q\, x +(\Q\, x)^{\perp}=(\Q\, x)^\perp\neq \Hp.
\end{align}
Here, the relation $(\Q\, x)^\perp\neq \Hp$ must hold, because the assumption that $(\Q\, x)^\perp=\Hp$ would imply
\begin{equation}
\Q\, x=(\Q\, x)^{\perp\perp}=\Hp^\perp=\{0\};
\end{equation}
i.e., we would have a contradiction. In conclusion, if $\dim(\Hp)\geq \isin$, there exists a (nonzero) isotropic vector $x\in\Hp$, so that
\begin{equation}
\Q\, x =(\Q\, x)^{\perp\perp}\quad\text{and}\quad \Q\, x + (\Q\, x)^\perp\neq \Hp;
\end{equation}
i.e., implication~\eqref{eq.orthcon} is violated.
\end{proof}

\begin{remark}
In the case where $\dim(\Hp)=\infty$, the second assertion of the previous proposition can be regarded as a manifestation
of Sol\`er's celebrated theorem~\cite{Soler95}, according to which a vector space over a division ring, endowed with a
(non-degenerate) Hermitian form satisfying a relation of the form~\eqref{eq.orthcon}, and admitting an infinite orthonormal sequence,
must be real, complex or quaternionic. Note that Sol\`er calls an Hermitian space where a relation of the type~\eqref{eq.orthcon}
is satisfied an \emph{orthomodular space}. This is due to the fact that the canonical orthocomplemented lattice of (form-closed)
subspaces of a orthomodular space turns out to be an orthomodular lattice; see, e.g., Theorem~$2.8$ in~\cite{piziak91}.
\end{remark}

In the next subsection, we will also argue that a further distinguishing mark of an infinite-dimensional $p$-adic Hilbert space $\Hp$,
versus a complex Hilbert space, is that it \emph{cannot} be identified with its (topological) dual $\dH$.


\subsection{Linear operators between $p$-adic normed spaces}
\label{linops}

Let $(X,\|\cdot\|_X)$, $(Y,\|\cdot\|_Y)$ be two $p$-adic normed spaces, and let $L: X\rightarrow Y$
be a linear operator from $X$ to $Y$. In the $p$-adic setting, as in the complex case,  linear operators  are
\emph{continuous} precisely when they are \emph{bounded}~\cite{robert2013course,narici71,rooij1978non,Perez-Garcia,diagana2016non};
specifically, $L$ is bounded if
\begin{equation}\label{normonL}
\|L\|\defi \sup_{x\neq 0}\frac{\|Lx\|_Y}{\|x\|_X}<\infty.
\end{equation}
(A bounded conjugate-linear operator and its norm are defined analogously).
We denote the space of bounded --- equivalently, continuous --- linear operators $L: X\rightarrow Y$,
by $\mathcal{B}(X,Y)$, and we refer to the norm in~\eqref{normonL} as the \emph{operator norm}.

\begin{theorem}\label{th.330}
Let $X,Y$ be $p$-adic normed spaces over $\Q$ --- with $Y$ complete and normal, i.e., $\|Y\|_Y=|\Q|$ ---
and let $X_0$ be a linear subspace of $X$. Then, every linear operator $L_0\in\mathcal{B}(X_0,Y)$ admits
an extension $L\in \mathcal{B}(X,Y)$ --- a so-called \emph{Hahn-Banach extension} of $L_0$ ---
such that $\|L\|=\|L_0\|$. If $X_0$ is dense in $X$, the bounded extension $L$ of $L_0$ is unique.
\end{theorem}

\begin{proof}
By Lemma~2.4 in~\cite{rooij1978non} (or by Proposition~20.2 in~\cite{schikhof2007ultrametric}) the normed space $Y$,
being complete and normal (and the valuation group $|\mathbb{Q}_p(\sqrt{\mu})^\ast|$ discrete), is spherically complete.
Then, the first assertion of the theorem follows from the non-Archimedean Hahn-Banach theorem (i.e., Ingleton's theorem;
see Theorem~4.8 in~\cite{rooij1978non}).
The second assertion is clear.
\end{proof}

From now on, we will remove subscripts $X$ and $Y$ from the associated norms, since it will be clear from the context to which space they refer.

\begin{proposition}[\cite{robert2013course}]
The space $\mathcal{B}(X,Y)$, endowed with the operator norm, is an ultrametric normed space over $\Q$,
which is complete (hence, an ultrametric Banach space) whenever $Y$ is.
\end{proposition}

As usual, in the case where $X=Y$, we simply write $\mathcal{B}(X)$ rather than $\mathcal{B}(X,X)$.
Moreover, since $\|ST\|\leq \|S\|\,\|T\|$, for all $S,T\in\mathcal{B}(X)$, if $X$ is a $p$-adic Banach space,
then $\mathcal{B}(X)$ is a \emph{unital ultrametric Banach algebra}~\cite{narici71}.

As in the standard complex case, we define the (topological) \emph{dual} of a $p$-adic Banach space $X$
as the ultrametric Banach space $\dX\defi\mathcal{B}(X, \Q)$. Denoting by $\ddX$ the \emph{bidual} of $X$,
the linear map $\IX: X\rightarrow \ddX$ defined by $\big(\IX(x)\big)(\xi)\defi\xi(x)$, for all $x\in X$ and $\xi\in\dX$,
is continuous (see Chapt.~3 of~\cite{rooij1978non}). If $\dim(X)<\infty$, then $X$ is \emph{reflexive}
(i.e., $\IX$ is a surjective isometry); otherwise, since $\Q$ is spherically complete,
by a classical result of Fleischer --- see Theorem~$4.16$ in~\cite{rooij1978non} --- $X$ is \emph{not} reflexive.
Nevertheless, in the case where $X$ is infinite-dimensional, it may be `pseudoreflexive':

\begin{definition}
A $p$-adic Banach space $X$ is said to be \emph{pseudoreflexive} if the linear map $\IX:X\rightarrow \ddX$ is an isometry.
\end{definition}

\begin{remark}
Let $\Hp$ be a $p$-adic Hilbert space, with $\dim(\Hp)=\infty$. Then, its dual $\dH$ is \emph{not} a $p$-adic Hilbert space
isomorphic to $\Hp$ (like in the complex case), but a $p$-adic Banach space isomorphic to $\ell^{\infty}(\nat,\Q)$; see the
forthcoming definition~\eqref{defdua} and Proposition~\ref{prop.334} below.
\end{remark}

In order to describe the dual of the $p$-adic Banach space $c_0(I,X)$, where $I=\{1,2,\ldots\}$
is a countable index set, we recall that, for a $p$-adic Banach space $X$, the space $\ell^{\infty}(I,X)$
is defined as follows:
\begin{equation} \label{defdua}
\ell^{\infty}(I,X)\defi\{\xi=\{\xi_i\}_{i\in I}\sep \mbox{$\xi_i\in X$, $\xi$ bounded sequence in $X$}\}.
\end{equation}
This space, equipped with the norm
\begin{equation}
\|\xi\|_{\infty}\defi\sup_{i\in I}\|\xi_i\|,
\end{equation}
is a $p$-adic Banach space~\cite{robert2013course,rooij1978non}. In particular, for $X=\Q$, we obtain the $p$-adic Banach space
\begin{equation}
\ell^{\infty}(I,\Q)\defi \{\xi=\{\xi_i\}_{i\in I}\sep \mbox{$\xi_i\in \Q$, $\xi$ bounded sequence in $\Q$}\},
\end{equation}
endowed with the norm $\|\xi\|_{\infty}\defi \sup_{i\in I}|\xi_i|$.

\begin{remark}
The $p$-adic Banach space $\ell^{\infty}(\nat,\Q)$ --- differently from $c_0(\nat,\Q)$ --- is \emph{not} separable;
equivalently --- recall Remark~\ref{resepa} --- it is not of countable type. In fact, for every $J\subset\nat$, let
$\indf\in\ell^{\infty}(\nat,\Q)$ be defined by
\begin{equation}
(\indf)_i=
\begin{cases}
1 \quad \mbox{if $i\in J$}
\\
0  \quad \mbox{if $i\not\in J$}
\end{cases}.
\end{equation}
Clearly, $\|\indf-\indfb\|_{\infty}=1$, whenever $J\neq K\subset\nat$. Let us put
\begin{equation}
\baj\defi\big\{\xi\in\ell^{\infty}(\nat,\Q) \sep \|\xi-\indf\|_{\infty}\le p^{-1}\big\}.
\end{equation}
Thus, $\{\baj\}_{J\subset\nat}$ is a countably infinite set of balls in $\ell^{\infty}(\nat,\Q)$.
Note that these balls are mutually disjoint because, if $\xi\in\indf$ and $J\neq K\subset\nat$, then
\begin{equation}
1=\|\indf-\indfb\|_{\infty}\le\max\{\|\indf-\xi\|_{\infty},\|\xi-\indfb\|_{\infty}\}=
\max\big\{p^{-1},\|\xi-\indfb\|_{\infty}\big\},
\end{equation}
so that $\|\xi-\indfb\|_{\infty}\ge 1$ and $\xi\not\in\bak$. Now, let $\mathscr{E}$ be any
dense subset of $\ell^{\infty}(\nat,\Q)$. Each ball in $\{\baj\}_{J\subset\nat}$ must contain
at least one element of $\mathscr{E}$, and such an element is not contained in any other ball
in $\{\baj\}_{J\subset\nat}$. It follows that there is an uncountable subset of $\mathscr{E}$,
so that $\mathscr{E}$ itself is uncountable and, hence, $\ell^{\infty}(\nat,\Q)$ is not separable.
It is worth observing that, more generally, $\ell^{\infty}(\nat,\mathbb{K})$ --- where $\mathbb{K}$ is any
complete ultrametric field, with a non-trivial valuation --- is not of countable type;
see Theorem~{2.5.15} in~\cite{Perez-Garcia}.
\end{remark}

\begin{proposition}[\cite{robert2013course,natarajan2019sequence,Perez-Garcia}]\label{dualspace}
Let $X$ be a $p$-adic Banach space over $\Q$. The topological dual of the space $c_0(I,X)$ is isomorphic, as a $p$-adic Banach space,
to $\ell^{\infty}(I, \dX)$. The identification of $c_0(I,X)^\prime$ with $\ell^{\infty}(I, \dX)$ is given via the bilinear pairing
\begin{equation}\label{pairinga}
\ell^\infty(I, \dX)\times c_0(I, X)\ni(\xi,y)\mapsto\sum_{i\in I}\xi_i (y_i)\ifed\xi(y)\in\Q.
\end{equation}
\end{proposition}

\begin{remark}\label{rem.3-}
Note that the norm of any $x\in c_0(I, \Q)\subset\ell^\infty(I, \Q)$ coincides with the norm of $x$ regarded as an element of $\ell^\infty(I,\Q)$
(equivalently, of $c_0(I,\Q)^\prime$). Moreover, by suitably composing the pairing~\eqref{pairinga} with the conjugate-linear isometry
$\{\xi_i\}_{i\in I}\mapsto \{\overline{\xi_i}\}_{i\in I}$ of $\ell^\infty(I, \Q)$ onto itself, we obtain the \emph{sesquilinear pairing}
\begin{equation}\label{pairingb}
\ell^\infty(I, \Q)\times c_0(I, \Q)\ni(\xi,y)\mapsto\sum_{i\in I}\overline{\xi_i}\, y_i=\overline{\xi}(y)\equiv\langle\xi,y\rangle\in\Q.
\end{equation}
This pairing determines a conjugate-linear isometry
\begin{equation}
c_0(I, \Q)\ni x\mapsto\scapx\in c_0(I, \Q)^\prime.
\end{equation}
Also note that $c_0(I, \Q)$ is pseudoreflexive, because the mapping
\begin{equation}
c_0(I, \Q)\ni x\mapsto\scapcx\in c_0(I, \Q)^{\prime\prime}
\end{equation}
--- where $\scapcx\colon\ell^\infty(I,\Q)\ni\xi\mapsto\langle\overline{\xi},x\rangle=\xi(x)=\sum_{i\in I}\xi_i x_i$ ---
is a linear isometry ($\|\scapx\|=\|x\|_{\infty}$). Clearly, we have: $\langle \xi,x\rangle=\big(\IX(x)\big)(\overline{\xi})$,
with $X=c_0(I,\Q)$. Observe that, with a slight abuse, we are using the same symbol $\scap$ for the inner product~\eqref{eq.Hs2}
and for the sesquilinear pairing~\eqref{pairingb}.
\end{remark}

\begin{proposition}\label{prop.334}
Let $\Hp$ be a $p$-adic Hilbert space over $\Q$, and let $\Phi\equiv\{\phi_i\}_{i\in I}$ be an orthonormal basis in $\Hp$. The mapping
\begin{equation}\label{defJH}
\JH:\Hp\ni\psi\mapsto\scapsi\in\dH
\end{equation}
is a conjugate-linear isometry of $\Hp$ into its dual $\dH$, that is surjective iff $\dim(\Hp)<\infty$.
The $p$-adic Banach space $\dH$ is isomorphic to $\ell^{\infty}(I,\Q)$ --- with $\card(I)=\dim(\dH)$ ---
and this isomorphism is implemented by the surjective isometry
\begin{equation} \label{defell}
\ellphi\colon\ell^{\infty}(I,\Q)\ni\xi=\{\xi_i\}_{i\in I}\mapsto\sum_{i\in I}\xi_i\scaphii\in\dH,
\end{equation}
where, if $I=\nat$, the series converges w.r.t.\ the weak$^{\hspace{0.4mm}\ast}\hspace{-0.6mm}$-topology;
moreover,
\begin{equation} \label{relell}
\ellphi(c_0(I,\Q))=\JH(\Hp).
\end{equation}

Finally, $\Hp$ is reflexive iff $\dim(\Hp)<\infty$; in the case where $\dim(\Hp)=\infty$, $\Hp$ is pseudoreflexive,
because the mapping
\begin{equation} \label{defIH}
\IH:\Hp\ni\psi\mapsto \Big(\dH\ni\dphi\mapsto\dphi(\psi)\in\Q\Big)\in\ddH
\end{equation}
is an isometry of $\Hp$ into its bidual $\ddH$.
\end{proposition}

\begin{proof}
The mapping~\eqref{defJH} is a (conjugate-linear) isometry, because
\begin{align}
\|\JH\psi\|=\sup_{\eta\neq 0}\frac{|\langle \psi,\eta\rangle|}{\|\eta\|}&
=\sup_{\eta\neq 0}\frac{|\sum_{i\in I}\langle \psi,\phi_i\rangle\langle\phi_i,\eta\rangle|}{\|\eta\|}\nonumber\\
&=\sup_{\eta\neq 0}\frac{\max_{i\in I}|\langle \psi,\phi_i\rangle|\,
|\langle\phi_i,\eta\rangle|}{\max_{i\in I}|\langle\phi_i,\eta\rangle|}\nonumber\\
&=\max_{i\in I}|\langle\phi_i,\psi\rangle|=\|\psi\|.
\end{align}
Since $\Hp$, as a $p$-adic Banach space, is isomorphic to $c_0(I,\Q)$ $(\card(I)=\dim(\Hp))$ via the mapping
$\Hp\ni x=\sum_{i\in I}\langle\phi_i,x\rangle\phi_i\mapsto\{\langle\phi_i,x\rangle\}_{i\in I}\in c_0(I,\Q)$,
then, by suitably composing this linear isometry with the bilinear pairing~\eqref{pairinga}, we see that the
map $\ellphi\colon\ell^{\infty}(I,\Q)\rightarrow \dH$ ---
$\big(\ellphi\big)(\psi)\defi\sum_{i\in I}\xi_i\langle \phi_i,\psi\rangle$ --- is an isomorphism of $p$-adic Banach spaces.
Therefore, we can write $\ellphi=\sum_{i\in I}\xi_i\scaphii$, where, for $I=\nat$, the series converges
pointwise, namely, w.r.t.\ the weak$^\ast\hspace{-0.3mm}$-topology (see, e.g., Sect.~{7.3} of~\cite{Perez-Garcia}).
Also note that $\JH(\Hp)=\ellphi(c_0(I,\Q))$; hence, $\JH$ is surjective iff $\dim(\Hp)<\infty$. Finally, we have
already observed that, if $\dim(\Hp)=\infty$, then $\Hp$ is not reflexive. Nevertheless, $\Hp$ is pseudoreflexive, because the mapping
\begin{equation}
\IH\colon \Hp\ni \psi=\sum_{i\in I}x_i\phi_i\mapsto\Big(\dH\ni\dphi=
\sum_{j\in I}\xi_j\scaphij\mapsto\phi^\prime(\psi)=\sum_{i\in I}\xi_i x_i\in \Q\Big)\in \ddH
\end{equation}
is a linear isometry.
\end{proof}

\begin{remark} \label{renoco}
With regard to the isometry~\eqref{defell}, note that --- since, for every $\xi\equiv\{\xi_i\}_{i\in I}\in\ell^{\infty}(I,\Q)$
and every finite subset $I_0$ of $I$,
\begin{equation}
\big\|{\textstyle\sum_{i\in I_0}}\xi_i\scaphii\big\|=\max_{i\in I_0}|\xi_i|\,\|\JH\phi_i\|
=\max_{i\in I_0}|\xi_i|
\end{equation}
--- in the case where $I=\nat$, the (pointwise converging) series $\sum_{i\in I}\xi_i\scaphii$ converges
w.r.t.\ the norm topology iff $\xi\in c_0(I,X)$.
\end{remark}

As a consequence of the first assertion of Proposition~\ref{prop.334}, we have the following:
\begin{corollary}
If $A\in \Bp$ and $B\in \mathcal{B}(\Hp^\prime)$ satisfy the intertwining relation
\begin{equation}\label{interell}
\JH\circ A=B\circ \JH,
\end{equation}
then $B\hspace{0.5mm}\JH(\Hp)\subset\JH(\ran(A))\subset\JH(\Hp)$ and $\|A\|=\|B_0\|\leq\|B\|$, where $B_0$ is
the restriction of $B$ to the closed subspace $\JH(\Hp)$ of $\dH$.
\end{corollary}

\begin{proof}
If relation~\eqref{interell} holds, then, since $\JH\colon \Hp\rightarrow \dH$ is an isometry, we have
\begin{align}
\|A\|=\sup_{\psi\neq 0}\frac{\|A\psi\|}{\|\psi\|}&=
\sup_{\psi\neq 0}\frac{\|\JH(A\psi)\|}{\|\JH(\psi)\|}\nonumber\\
&=\sup_{\psi\neq 0}\frac{\big\|B\big(\JH(\psi)\big)\big\|}{\|\JH(\psi)\|}\nonumber\\
&=\sup_{0\neq\dphi\in \JH(\Hp)}\frac{\|B\dphi\|}{\|\dphi\|}=\|B_0\|\leq \|B\|,
\end{align}
where $B_0$ is the restriction of $B$ to $\JH(\Hp)$, which, by~\eqref{interell}, is a (closed) subspace of $\dH$,
stable under the action of $B$.
\end{proof}

\begin{definition}
If $A\in \Bp$ and $B\in\mathcal{B}(\Hp^\prime)$ satisfy the intertwining relation~\eqref{interell} and, moreover,
$\|A\|=\|B\|$, we say that $B$ is a \emph{dual Hahn-Banach extension} of $A$.
\end{definition}

For every bounded  linear operator $A\in\Bp$, one can define its \emph{Banach space adjoint} --- or \emph{generalized adjoint}
--- $\dA: \dH\rightarrow \dH$ by setting $\dA(\dphi)\defi \dphi\circ A$, namely,
\begin{equation}
\big(\dA\dphi\big)(\psi)\defi \dphi(A\psi),\quad \forall\psi\in\Hp,\,\forall\dphi\in\dH.
\end{equation}
In particular, $\dA(\scaphi)=\scaphi\circ A =\langle \phi,A(\cdot)\rangle$, for all $\phi\in\Hp$; namely,
\begin{equation}\label{interelb}
\big(A^\prime\circ\JH\big)(\phi)=\JH(\phi)\circ A,\quad \forall\phi\in\Hp.
\end{equation}
\begin{proposition} \label{proda}
If $A\in\Bp$, then $\dA\in\mathcal{B}(\mathcal{H}^\prime)$ and $\|\dA\|=\|A\|$.
\end{proposition}

\begin{proof}
See Section~4F of~\cite{rooij1978non}.
\end{proof}


\section{Bounded and adjointable operators in a $p$-adic Hilbert space}
\label{sec4}

In this section, we derive some useful characterizations of bounded operators in a $p$-adic Hilbert space $\Hp$.
We will focus on the case where $\dim(\Hp)=\infty$, because in the finite-dimensional setting most of the subsequent
results become trivial. Accordingly, we will identify the index set $I$ of the previous sections with $\mathbb{N}$.
For the sake of conciseness, we put $c_0\equiv c_0(\mathbb{N},\Q)$ and $\ell^{\infty}\equiv\ell^{\infty}(\mathbb{N},\Q)$.

Let $(A_{mn})$, $m,n\in\nat$, be an infinite matrix with entries in $\Q$. We denote the set of all such matrices by $\Minf$.
The matrix $(A_{mn})$ determines a linear operator $\op(A_{mn})$ in $c_0$ by putting
\begin{align}
\dom(\op(A_{mn}))\defi\big\{  &
x=\{x_n\}_{n\in \mathbb{N}}\in c_0 \sep
\mbox{the series $\sum_n A_{mn}x_n$ converges, $\forall m\in\mathbb{N}$, and}
\nonumber\\ \label{eq.1}
& \mbox{$\sum_m\sum_nA_{mn}x_n$ converges too, i.e., $\big{\{}\sum_n A_{mn}x_n\big{\}}_{m\in \mathbb{N}}\in c_0$}
\big\},
\end{align}
\begin{equation}\label{eq.2}
\textstyle
\op(A_{mn})\, x\defi\big\{\sum_n A_{mn} x_n\big\}_{m\in\mathbb{N}},\quad x=\{x_n\}_{n\in \mathbb{N}}\in\dom(\op(A_{mn})).
\end{equation}
Clearly, the \emph{matrix operator} $\op(A_{mn})$ will be --- in general --- unbounded, and here we are assuming that it is defined on its
\emph{maximal domain} $\dom(\op(A_{mn}))$.

We further introduce the following set of linear operators in $c_0$:
\begin{equation}\label{eq.3}
(c_0,c_0)\defi\{\op(A_{mn})\sep (A_{mn})\in \Minf,\;\dom(\op(A_{mn}))=c_0\}.
\end{equation}
By Theorem $65$ in~\cite{natarajan2019sequence},  $\op(A_{mn})\in (c_0,c_0)$ iff

\begin{enumerate}[label=\tt{(M\arabic*)}]

\item \label{cond.boi}  $\lim_m A_{mn}=0$, $\forall n\in\mathbb{N}$;

\item \label{cond.boii} $\sup_m(\sup_n|A_{mn}|)<\infty$.

\end{enumerate}

\begin{remark}
It is worth stressing the following points:
\begin{itemize}
\item Since $(\overline{\lim}\equiv \limsup)$
\begin{equation}
\overline{\lim_m}\Big{(}\sup_n|A_{mn}|\Big{)}\leq \sup_m\Big{(}\sup_n|A_{mn}|\Big{)}<\infty,
\end{equation}
 then the further condition that --- see Theorem~$65$ of~\cite{natarajan2019sequence} ---
\begin{equation}
\lim_{l}\frac{1}{l}\;\overline{\lim_m}\Big(\sup_n|A_{mn}|\Big)=0
\end{equation}
is redundant in the case we are considering.
\item By the Principle of the Iterated Suprema, if for $r_{mn}\geq 0$, $m,n\in\mathbb{N}$, either
$\sup_{m,n}r_{mn}<\infty$, or $\sup_m\sup_n r_{mn}<\infty$, or $\sup_n\sup_m r_{mn}<\infty$, then
\begin{equation}
\sup_{m,n}r_{mn}=\sup_m\sup_nr_{mn}=\sup_n\sup_mr_{mn}<\infty.
\end{equation}
\end{itemize}

Thus, condition~\ref{cond.boii} above can be replaced with
\begin{enumerate}[label=\tt{(M\arabic*)}$^\prime$]
\setcounter{enumi}{1}

\item \label{cond.boiii} $\sup_{m,n}|A_{mn}|<\infty$.

\end{enumerate}
\end{remark}

We next switch from the sequence space $c_0$ to an infinite-dimensional $p$-adic Hilbert space $\Hp$ over $\Q$.
We say that a linear operator $A$ in $\Hp$ is \textit{all-over} if
\begin{equation}
\dom(A)=\Hp.
\end{equation}
Since, by Theorem~\ref{th.330}, a densely defined \emph{bounded} linear operator in a $p$-adic Hilbert space
admits a \emph{unique} bounded linear extension to the whole space, with the same norm --- precisely as it happens in the standard
complex setting; see, e.g., Theorem~{4.5} of~\cite{weidmann2012linear} --- we will tacitly assume bounded operators to be all-over,
i.e., to be defined on the whole $\Hp$, unless otherwise specified. As in Section~\ref{sec3}, we denote the space of all such operators
by $\Bp$.

Keeping in mind the isomorphism of $p$-adic Hilbert spaces $\Wf\colon\Hp\rightarrow\mathbb{H}$ --- see~\eqref{isohilb} ---
where $\mathbb{H}\equiv\mathbb{H}(\mathbb{N})$ is the Banach space $c_0$, endowed with its canonical inner product~\eqref{eq.Hs2},
we can now consider \emph{matrix operators} in $\Hp$. An infinite matrix $(A_{mn})\in\Minf$ ---
together with an orthonormal basis $\Phi\equiv\{\phi_n\}_{n\in\mathbb{N}}$ in $\Hp$ --- determines a linear operator
$\opPhi(A_{mn})$ in $\Hp$ as follows:
\begin{align}
\dom(\opPhi(A_{mn}))\defi\big\{ &
\mbox{$\psi=\sum_n x_n\phi_n=\sum_n\langle \phi_n,\psi\rangle\phi_n \sep
\sum_n A_{mn}x_n$ converges, $\forall m\in\mathbb{N}$, and}
\nonumber\\ \label{eq.7}
& \mbox{$\sum_m(\sum_nA_{mn}x_n)\phi_m$ converges too, i.e.,
$\big\{\sum_n A_{mn}x_n\big\}_{m\in\mathbb{N}}\in c_0$}\big\},
\end{align}
\begin{equation}\label{eq.8}
\textstyle
\opPhi(A_{mn})\,\psi\defi\sum_m\big(\sum_n A_{mn}\langle\phi_n,\psi\rangle\big)\phi_m,\quad\psi\in\dom(\opPhi(A_{mn})).
\end{equation}
Taking into account definitions~\eqref{eq.1}, \eqref{eq.2}, \eqref{eq.7} and~\eqref{eq.8},
we see that
\begin{equation} \label{relops}
\dom(\op(A_{mn})) = \Wf\,\dom(\opPhi(A_{mn})) \quad \mbox{and} \quad
\op(A_{mn})\circ\Wf = \Wf\circ\opPhi(A_{mn}) .
\end{equation}

For every orthonormal basis $\Phi\equiv\{\phi_n\}_{n\in\mathbb{N}}$ in $\Hp$, we put
\begin{align}\label{eq.10}
(\Hp,\Hp)_{\Phi}\defi\{\opPhi(A_{mn})\sep (A_{mn})\in \Minf,\; \dom(\opPhi(A_{mn}))=\Hp\};
\end{align}
i.e., $(\Hp,\Hp)_{\Phi}$ is the set of the all-over matrix operators in $\Hp$ associated with $\Phi$.
It is clear that $(\Hp,\Hp)_{\Phi}$ is, in a natural way, a linear space over $\Q$. Actually,
we will show that its definition does not depend on the choice of $\Phi$.

To this end, first note that every bounded operator $A\in\Bp$ belongs to $(\Hp,\Hp)_\Phi$
(whatever the orthonormal basis $\Phi\equiv\{\phi_m\}_{m\in\mathbb{N}}$ is) because, putting
\begin{equation}\label{eq.11}
A_{mn}\equiv \langle \phi_m, A\phi_n\rangle, \qquad
\mbox{($\lim_m A_{mn}=0$, $\forall n\in\mathbb{N}$, $\sup_{m,n}|A_{mn}|\le\|A\|$)}
\end{equation}
we have that
\begin{equation}\label{eq.12}
A=\opPhi(A_{mn})=\sum_m\sum_n A_{mn}\langle\phi_n,\cdot\hspace{0.4mm}\rangle\phi_m
=\sum_n\sum_m A_{mn}\langle \phi_n,\cdot\hspace{0.4mm}\rangle\phi_m,
\end{equation}
where $(\langle\phi_n,\cdot\hspace{0.4mm}\rangle\phi_m)\,\psi\defi\langle\phi_n,\psi\rangle\phi_m$ (i.e.,
$\langle\phi_n,\cdot\hspace{0.4mm}\rangle\phi_m\equiv|\phi_m\rangle\langle\phi_n|$, in Dirac's notation) and
both the iterated series converge --- as can be easily checked --- w.r.t.\ the strong operator topology in $\Bp$
(i.e., the initial topology induced by the family of maps $\{\eps\colon\Bp\rightarrow\Hp\}_{\psi\in\Hp}$, where
$\eps(A)\defi A\psi$). Thus, $\Bp\subset(\Hp,\Hp)_{\Phi}$; precisely:

\begin{theorem}\label{th.1}
For every orthonormal basis $\Phi\equiv\{\phi_m\}_{m\in\mathbb{N}}$ in $\Hp$, we have that
\begin{align}\label{eq.13}
\Bp=(\Hp,\Hp)_{\Phi}=\{\opPhi(A_{mn})\sep \mbox{\rm $\lim_m A_{mn}=0$, $\forall n\in\nat$, and
$\sup_{m,n}|A_{mn}|<\infty$}\}.
\end{align}
Moreover, for every $A=\opPhi(A_{mn})\in\Bp$,
\begin{equation}\label{eq.14}
\|A\|=\sup_{m,n}|A_{mn}|=\sup_n\|A\phi_n\|.
\end{equation}
\end{theorem}

\begin{proof}
By relations~\eqref{relops}, we have that
\begin{equation}
(\Hp,\Hp)_{\Phi}=\{\opPhi(A_{mn})\sep \op(A_{mn})\in (c_0,c_0)\},
\end{equation}
where the set $(c_0,c_0)$ is completely characterized by conditions~\ref{cond.boi} and~\ref{cond.boiii}.
Hence, the second equality in~\eqref{eq.13} holds true.

We have already observed that $\Bp\subset(\Hp,\Hp)_{\Phi}$. It remains to prove the reverse inclusion $(\Hp,\Hp)_{\Phi}\subset \Bp$
--- so that, actually, $(\Hp,\Hp)_{\Phi}=\Bp$ --- and, moreover, to show that
\begin{equation}
\|A\|=\sup_{m,n}|A_{mn}|=\sup_n\|A\phi_n\|,\quad\forall A=\opPhi(A_{mn})\in\Bp.
\end{equation}
Let $A\in (\Hp,\Hp)_{\Phi}$, with $A=\opPhi(A_{mn})$. For every $\psi\in\Hp$, we have that
\begin{equation}
A\psi=\sum_m\bigg(\sum_n A_{mn}\langle\phi_n,\psi\rangle\bigg)\phi_m;
\end{equation}
hence:
\begin{equation}
\|A\psi\|=\sup_m\Big{|}\sum_n A_{mn}\langle\phi_n,\psi\rangle\Big{|}\leq\sup_{m}\sup_n|A_{mn}| \,
|\langle \phi_n,\psi\rangle|\;\leq \|\psi\| \, \sup_m\sup_n|A_{mn}|.
\end{equation}
Therefore, $A$ is bounded and
\begin{equation}\label{eq.18}
\|A\|\leq \sup_m\sup_n|A_{mn}|=\sup_{m,n}|A_{mn}|.
\end{equation}
It is then shown that $(\Hp,\Hp)_\Phi\subset\Bp$ too, and therefore, actually, the two sets coincide.

Now, let $A=\opPhi(A_{mn})\in (\Hp,\Hp)_\Phi=\Bp$. Note that
\begin{align}\label{eq.19}
\|A\phi_n\|&=\Big{\|}\sum_m\langle \phi_m,A\phi_n\rangle\phi_m\Big{\|}=
\sup_m|\langle\phi_m, A\phi_n\rangle|=\sup_m|A_{mn}|,\;\;\forall n\in\mathbb{N}.
\end{align}
From~\eqref{eq.18} and~\eqref{eq.19} it follows that
\begin{equation}
\sup_{m,n}|A_{mn}|=\sup_n\sup_m|A_{mn}|=\sup_n\|A\phi_n\|\leq\|A\|\leq \sup_{m,n}|A_{mn}|.
\end{equation}
Hence, actually, $\|A\|=\sup_{m,n}|A_{mn}|=\sup_n\|A\phi_n\|$, for every $A=\opPhi(A_{mn})\in\Bp$.
\end{proof}

Summarizing, it is proven that $\Bp$ coincides with the linear space $(\Hp,\Hp)_{\Phi}$ of the \emph{all-over matrix operators} in $\Hp$
w.r.t.\ an orthonormal basis $\Phi\equiv\{\phi_n\}_{n\in\mathbb{N}}$ (independently of the choice of $\Phi$);
moreover, one can give a complete characterization of the infinite matrices in $\Minf$ that correspond to bounded operators
(w.r.t.\ any orthonormal basis in $\Hp$). Furthermore, the norm of a bounded operator is given by the supremum, in valuation,
of its matrix elements (again, w.r.t.\ any orthonormal basis).

At this point, we move on to discuss the notion of `proper adjoint' of a bounded operator in a $p$-adic Hilbert space $\Hp$.
In fact, as argued in Section~\ref{sec3}, since, in the infinite-dimensional setting, $\Hp$ cannot be identified with its dual $\dH$,
with every bounded operator $A\in\Bp$ is associated a generalized adjoint $\dA\in\mathcal{B}(\dH)$. Nevertheless, one can single out a suitable
class of bounded operators admitting a genuine `Hilbert space adjoint' (the proper adjoint).

Given any $A\in\Bp$, we first associate with $A$ a linear operator $A^\dagger$ in $\Hp$ --- the so called \emph{pseudo-adjoint} of $A$ ---
as follows. We start with defining its domain as
\begin{align}\label{eq.21}
\dom(A^\dagger)\defi\{\phi\in\Hp\sep \langle\phi, A\psi\rangle=\langle\eta,\psi\rangle,\;\text{for some}\;\eta\in\Hp,
\;\text{and for all}\;\psi\in\Hp\}.
\end{align}
It is clear that $\dom(A^\dagger)$ is a linear subspace of $\Hp$.

\begin{remark}\label{rem.3}
The linear subspace $\dom(A^\dagger)$ can be regarded as the set of all vectors $\phi\in\Hp$ such that the bounded functional
$\langle \phi, A(\cdot)\rangle\in\dH$ can be identified, via the conjugate-linear isometry $\JH$, with an element
$\eta$ of $\Hp$; i.e.,
\begin{equation}
\phi\in\dom(A^\dagger) \ \iff \ \langle\phi, A(\cdot)\rangle=\scapeta=\JH\,\eta,
\end{equation}
for some $\eta\in\Hp$.
\end{remark}

\begin{proposition}\label{proadj}
The condition
\begin{equation}\label{eq.24}
\big(\JH\,\phi\big)(A\psi)=\langle\phi, A\psi\rangle=\langle A^\dagger\phi,\psi\rangle=\big(\JH(A^\dagger\phi)\big)(\psi),
\;\;\forall\phi\in\dom(A^\dagger),\;\forall\psi\in\Hp,
\end{equation}
uniquely determines a bounded linear operator $A^\dagger\colon\dom(A^\dagger)\rightarrow\Hp$, whose domain $\dom(A^\dagger)$
is given by~\eqref{eq.21}, and such that $\|A^\dagger\|\le\|A\|$.
Moreover, we have the following dichotomy: either $\dom(A^\dagger)=\Hp$, or $\dom(A^\dagger)$ is a closed subspace
of $\Hp$, with $\dom(A^\dagger)\subsetneq\Hp$.
\end{proposition}

\begin{proof}
For every $\phi\in\dom(A^\dagger)$, the relation
$(\JH\,\phi)\circ A=\JH\,\eta= \JH(A^\dagger\phi)$, determines a linear operator $A^\dagger\colon\dom(A^\dagger)\rightarrow\Hp$,
where the vector $\eta\ifed A^\dagger\phi\in\Hp$ is unique because $\JH$ is a (conjugate-linear) isometry.
Equivalently, one can use the fact that the sesquilinear form $\scap$ is non-degenerate; see Proposition~\ref{nondeginpro}.
(The linearity of $A^\dagger$ is evident.) Moreover, since
\begin{equation}
\|\JH(A^\dagger\phi)\|=\|A^\dagger\phi\|=\|(\JH\,\phi)\circ A\|\le\|\JH\,\phi\|\|A\|=\|\phi\|\|A\|,\quad \phi\in\dom(A^\dagger),
\end{equation}
$A^\dagger$ is bounded, with $\|A^\dagger\|\le\|A\|$.

Let us now show that $\dom(A^\dagger)$ is a closed subspace of $\Hp$. In fact, by Theorem~\ref{th.330}, there is a bounded operator
$B$ in $\Hp$ that agrees with $A^\dagger$ on $\dom(A^\dagger)$, and such that $\|B\|=\|A^\dagger\|\le\|A\|$. Therefore, for every
sequence $\{\chi_n\}_{n\in\nat}\subset\dom(A^\dagger)$, with $\chi_n\rightarrow\chi\in\overline{\dom(A^\dagger)}^{\,\|\cdot\|}$,
we have that
\begin{equation}
\langle B\chi,\psi\rangle=\lim_n\,\langle B\chi_n,\psi\rangle=\lim_n\,\langle A^\dagger\chi_n,\psi\rangle
=\lim_n\,\langle \chi_n,A\psi\rangle=\langle \chi,A\psi\rangle, \quad \forall\psi\in\Hp,
\end{equation}
where we have used the continuity of $B$ and of the inner product $\scap$; hence: $\chi\in\dom(A^\dagger)=\overline{\dom(A^\dagger)}^{\,\|\cdot\|}$
and $B=A^\dagger$. In conclusion, either $\dom(A^\dagger)=\Hp$, or $\dom(A^\dagger)$ is a closed subspace of $\Hp$, strictly
contained in $\Hp$.
\end{proof}

\begin{remark}
By the definition of $A^\dagger$ and by the dichotomy in Proposition~\ref{proadj}, it is clear that, if $\dom(A^\dagger)\subsetneq\Hp$,
then, given any bounded extension $B\in\Bp$ of $A^\dagger$, we have:
\begin{equation}
\begin{cases}
\langle B\phi,\psi\rangle=\langle\phi, A\psi\rangle,\;\;\forall\phi\in\dom(A^\dagger)=\overline{\dom(A^\dagger)}^{\,\|\cdot\|},\;\forall\psi\in\Hp, \\
\langle B\phi,\psi\rangle\neq\langle\phi, A\psi\rangle,\;\;\forall\phi\not\in\dom(A^\dagger),\;\forall\psi\not\in\sephi,
\end{cases}
\end{equation}
where the set $\sephi\defi\{\psi\in\Hp\sep \langle B\phi,\psi\rangle=\langle\phi, A\psi\rangle\}$ is a non-dense, closed linear subspace of $\Hp$
depending on $\phi\not\in\dom(A^\dagger)$.
\end{remark}

\begin{definition}
We say that $A\in\Bp$ is \emph{adjointable} if the pseudo-adjoint $A^\dagger$ of $A$ is all-over, i.e., if $\dom(A^\dagger)=\Hp$.
In such a case,  we put $A^*\equiv A^\dagger$ and we call the all-over linear operator $A^*$ the (proper) \emph{adjoint} of $A\in\Bp$.
\end{definition}

It is clear that the collection of all adjointable operators in $\Hp$ --- denoted hereafter by $\Bpa$ --- is, in a natural way,
a linear space over $\Q$ (a linear subspace of $\Bp$). We are now going to characterize it.

\begin{theorem}\label{th.2}
If $A\in\Bp$ is adjointable, then its (proper) adjoint $A^*$ is a bounded operator. Given any orthonormal basis
$\Phi\equiv\{\phi_m\}_{m\in\nat}$ in $\Hp$, if $A=\opPhi(A_{mn})\in\Bpa$, then
\begin{equation}\label{eq.25}
A^*=\opPhi(A^*_{mn})\in\Bpa,
\end{equation}
with $A^*_{mn}=\overline{A_{nm}}$. Therefore, if $A=\opPhi(A_{mn})\in\Bpa$, then
\begin{enumerate}[label=\tt{(A\arabic*)}]

\item \label{cond.bot1} $\sup_{m,n}|A_{mn}|<\infty$,

\item \label{cond.bot2} $\lim_m A_{mn}=0$, $\forall n\in\nat$,

\item \label{cond.bot3} $\lim_n A_{mn}=0$, $\forall m\in\nat$,

\end{enumerate}
and, moreover,
\begin{equation}\label{eq.26}
\|A^*\|=\sup_{m,n}|A_{mn}|=\|A\|.
\end{equation}

Conversely, if, for some orthonormal basis $\Phi\equiv\{\phi_m\}_{m\in\nat}$, $A=\opPhi(A_{mn})$ --- i.e., if $A$ is a matrix operator associated with $\Phi$ ---
where the entries of the matrix $(A_{mn})\in\Minf$ are supposed to satisfy conditions~\ref{cond.bot1}--\ref{cond.bot3} above,
then $A\in\Bpa$.
\end{theorem}

\begin{notation}\label{not.1}
Taking into account the mapping~\eqref{defell}, for every $\xi\in\ell^{\infty}$, we write symbolically
\begin{equation}\label{eq.29}
\xi_{\Phi}=\langle{\textstyle\sum_n}\xi_n\phi_n,\cdot\hspace{0.4mm}\rangle\equiv\ellphi(\overline{\xi})=
\sum_{n}\overline{\xi_n}\,\scaphin\in\dH.
\end{equation}
Here, the (pointwise converging) series $\sum_{i\in I}\overline{\xi_i}\,\langle\phi_i,\cdot\hspace{0.4mm}\rangle$
converges w.r.t.\ the norm topology too iff $\xi\in c_0$ (Remark~\ref{renoco}).
With this notation, if, in particular, $\xi\in c_0\subset\ell^\infty$, the functional $\langle\sum_n \xi_n\phi_n,\cdot\hspace{0.4mm}\rangle$
can be directly identified with the element $\sum_n\xi_n\phi_n$ of $\Hp$ (with norm $\|\sum_n\xi_n\phi_n\|=\|\xi\|_{\infty}=
\sup_n|\xi_n|=\max_n|\xi_n|$):
\begin{equation}
\xi_\Phi(\psi)=\langle{\textstyle\sum_n}\xi_n\phi_n,\psi\rangle=\sum_n\overline{\xi_n}\,\langle\phi_n,\psi\rangle.
\end{equation}
\end{notation}

\begin{proof}[Proof of Theorem~\ref{th.2}]
Given an orthonormal basis $\Phi\equiv\{\phi_n\}_{n\in\mathbb{N}}$ in $\Hp$ and a bounded operator $A\in\Bp$, $A=\opPhi(A_{mn})$,
by the second series expansion in~\eqref{eq.12} --- converging w.r.t.\ the strong operator topology --- we have that
\begin{equation}\label{eq.31}
\langle \psi, A\chi\rangle=\sum_n\sum_m A_{mn}\langle \psi, \phi_m\rangle\langle \phi_n,\chi\rangle
=\sum_n\bigg(\sum_m\overline{\overline{A_{mn}}\,\langle\phi_m,\psi\rangle}\bigg)\langle\phi_n,\chi\rangle,
\quad\forall\psi,\chi\in\Hp,
\end{equation}
where, by the arbitrariness of $\chi\in\Hp$ (and since $\ell^\infty$ is the generalized K\"othe-Toeplitz dual of $c_0$,
see Sect.~{1.6} of~\cite{natarajan2019sequence}),
\begin{equation}\label{eq.32}
\big\{\xi_{n}^{\psi}\equiv{\textstyle\sum_m}\overline{A_{mn}}\,\langle\phi_m,\psi\rangle\big\}_{n\in\nat}\in\ell^\infty.
\end{equation}
Using Notation~\ref{not.1}, we can write
\begin{equation}\label{eq.33}
\langle\psi, A\chi\rangle=\sum_n\overline{\xi_{n}^{\psi}}\,\langle\phi_m,\chi\rangle\equiv
\big\langle{\textstyle\sum_n}\xi_n^\psi\phi_m,\chi\big\rangle.
\end{equation}
Here, in general, we have a pairing between an element of $\dH$ --- i.e.,
$\langle\psi,A(\cdot)\rangle=\big\langle\hspace{-0.5mm}\sum_n\xi_{n}^{\psi}\phi_m,\cdot\hspace{0.4mm}\big\rangle$
--- and the vector $\chi\in\Hp$.

Assume now that $A \in\Bpa$. Then, for some $\eta=A^*\psi\in\Hp$,
\begin{equation}\label{eq.34}
\big\langle{\textstyle\sum_n}\xi_{n}^{\psi}\phi_m,\chi\big\rangle=\langle\psi, A\chi\rangle=\langle A^*\psi,\chi\rangle,
\quad\forall\chi\in\Hp.
\end{equation}
By the second assertion of Proposition~\ref{prop.334}, $\ellphi$ is injective and $\ellphi(c_0)=\JH(\Hp)$, so that
\begin{equation}\label{eq.35}
\big\{\xi_{n}^{\psi}\big\}_{n\in\nat}\in c_0\;\;\;\;\text{and}\;\;\;\;\sum_n\xi_{n}^{\psi}\phi_n=A^*\psi\in\Hp.
\end{equation}
We stress that here $\sum_n\xi_{n}^{\psi}\phi_n$ should be regarded as a \emph{bona fide} expansion in $\Hp$.
Recalling~\eqref{eq.32}, we see that
\begin{equation}
A^*\psi=\sum_n\xi_{n}^{\psi}\phi_n=\sum_n\sum_m\overline{A_{mn}}\,\langle\phi_m,\psi\rangle\phi_n
\implies A^*=\opPhi(A_{mn}^*),
\end{equation}
where $A_{mn}^*=\overline{A_{nm}}$.

Therefore, the adjoint $A^*$ of $A$ --- that, by definition, is defined on the whole Hilbert space $\Hp$ and, by Proposition~\ref{proadj},
is bounded --- is the (all-over) matrix operator $\opPhi(A_{mn}^*)$ and, by Theorem~\ref{th.1}, we conclude that
\begin{itemize}

\item[(i)] $\sup_{m,n}|A_{mn}^*|<\infty$ ($\Longleftrightarrow\sup_{m,n}|A_{mn}|<\infty$),

\item[(ii)] $\lim_m A_{mn}^*=0$, $\forall n\in\nat$ ($\Longleftrightarrow\lim_n A_{mn}=0$, $\forall m\in\nat$),

\end{itemize}
and, moreover, since $A\in\Bp$, also

\begin{itemize}
\item[(iii)] $\lim_n A^{*}_{mn}=0$, $\forall m\in\nat$ ($\Longleftrightarrow\lim_m A_{mn}=0$, $\forall n\in\nat$).
\end{itemize}

It is clear that $\|A^*\|=\sup_{m,n}|A^{*}_{mn}|=\sup_{m,n}|A_{mn}|=\|A\|$.

Now, let $A=\opPhi(A_{mn})$ --- with $\Phi\equiv\{\phi_m\mnat$ any orthonormal basis in $\Hp$ --- be a matrix operator.
By relation~\eqref{eq.13} in Theorem~\ref{th.1}, if $\sup_{m,n}|A_{mn}|<\infty$ and $\lim_m A_{mn}=0$, for every
$n\in\nat$, then $A$ is a bounded operator so that, for all $\psi,\chi\in\Hp$, relation~\eqref{eq.31} holds true.
Next, if, moreover, $\lim_n A_{mn}=0$, for every $m\in\nat$, then --- putting  $A^{*}_{mn}\equiv\overline{A_{nm}}$ ---
we also have that $\op(A^{*}_{mn})\in (c_0,c_0)$. Hence,
\begin{equation}
\{{\textstyle\sum_n} A^{*}_{mn}x_n\}_{m\in\nat}=\{{\textstyle\sum_n} \overline{A_{nm}}\,x_n\}_{m\in\nat}\in c_0,\quad
\forall x=\{x_n\nnat\in c_0,
\end{equation}
and, for every $\psi\in\Hp$, the series
\begin{equation}
\sum_m\bigg(\sum_n A^{*}_{mn}\langle\phi_n,\psi\rangle\bigg)\phi_m,
\end{equation}
must converge to some vector $\eta=\opPhi(A^{*}_{mn})\,\psi\in\Hp$, where $\opPhi(A^{*}_{mn})\in(\Hp,\Hp)_{\Phi}=\Bp$.
In conclusion, for every $\psi\in\Hp$, there is some $\eta=\opPhi(A^{*}_{mn})\,\psi\in\Hp$ such that, given any $\chi\in\Hp$,
by~\eqref{eq.31} (with the indices $m,n$ merely re-named) we have:
\begin{equation}
\langle\psi, A\chi\rangle =\sum_m\bigg(\sum_n\overline{\overline{A_{nm}}\,\langle\phi_n,\psi\rangle}\bigg)\langle\phi_m,\chi\rangle
=\sum_m\bigg(\sum_n \overline{A^{*}_{mn}\,\langle\phi_n,\psi\rangle}\bigg)\langle\phi_m,\chi\rangle=\langle\eta,\chi\rangle.
\end{equation}
Otherwise stated, $A=\opPhi(A_{mn})$ --- with the matrix $(A_{mn})$ that satisfies conditions~\ref{cond.bot1}--\ref{cond.bot3}
in the statement of the theorem --- is adjointable, and $A^*=\opPhi(A_{mn}^*)$, $A^{*}_{mn}\equiv\overline{A_{nm}}$, because the relation
$\langle\psi, A\chi\rangle=\langle A^\ast\psi,\chi\rangle$, for all $\psi,\chi\in\Hp$ determines $A^\ast$ uniquely.
\end{proof}

We will now derive some direct consequences of Theorem~\ref{th.2}.

\begin{corollary} \label{corad}
If $A\in\Bpa$, then $A^*\in\Bpa$ too and
\begin{equation}
(A^*)^*=A.
\end{equation}
Moreover, for all $A,B\in\Bpa$ and all $\alpha\in\Q$, we have that ($\alpha A,A+B\in\Bpa$ and) $AB,\id\in\Bpa$,
together with
\begin{equation}
(\alpha A)^*=\overline{\alpha} A, \quad (A+B)^*=A^*+B^*, \quad (AB)^*=B^*A^* .
\end{equation}
\end{corollary}

\begin{proof}
If $A=\opPhi(A_{mn})\in\Bp$ is adjointable, then the matrix elements of $(A_{mn})$ satisfy conditions~\ref{cond.bot1}--\ref{cond.bot3}
in Theorem~\ref{th.2}, and clearly, the matrix elements of $A^*$ w.r.t. the orthonormal basis $\Phi\equiv\{\phi_m\}_{m\in\mathbb{N}}$
--- recall that $A^*=\opPhi(A^*_{mn})$, with $A^*_{mn}=\overline{A_{nm}}$ --- satisfy these conditions too. Hence, $A^*$ is adjointable
too and $(A^*)^*=\opPhi(A_{mn})=A$. We have already observed that $\Bpa$ is a linear subspace of $\Bp$. The remaining facts are clear
from the definition of the adjoint of a bounded operator in a $p$-adic Hilbert space.
\end{proof}

\begin{corollary} \label{corhbe}
If $A\in\Bpa$, then its generalized adjoint $\dA\in\mathcal{B}(\dH)$ is a dual Hahn-Banach extension of its adjoint $A^*\in\Bpa$.
Moreover, if $A\in\Bpa$ and, for some all-over linear operator $B$ in $\Hp$, the intertwining relation $A^\prime\circ\JH=\JH\circ B$
holds, then $B\in\Bpa$, $\dA$ is a dual Hahn-Banach extension of $B$ and $B=A^\ast$.
\end{corollary}

\begin{proof}
We have already observed that the generalized adjoint $\dA$ of a bounded operator $A$ satisfies the intertwining
relation~\eqref{interelb}; i.e., for every $\phi\in\Hp$, $\big(A^\prime\circ\JH\big)(\phi)=\JH(\phi)\circ A$.
If $A\in\Bpa$, by Proposition~\ref{proadj}, the proper adjoint $A^*$ is (uniquely) determined by the condition that
$\JH(\phi)\circ A=\big(\JH\circ A^*\big)(\phi)$, for all $\phi\in\Hp$. Therefore, we have that
\begin{equation}
\big(A^\prime\circ\JH\big)(\phi)=\JH(\phi)\circ A=\big(\JH\circ A^*\big)(\phi), \quad \forall\phi\in\Hp;
\end{equation}
i.e., the conjugate-linear isometry $\JH$ intertwines $A^*$ with $\dA$. Thus, $\dA$ is a dual Hahn-Banach extension
of $A^*$, because, by relation~\eqref{eq.26} in Theorem~\ref{th.2} and by Proposition~\ref{proda}, we also have that
$\|A^*\|=\|A\|=\|\dA\|$.

Moreover, if $A\in\Bpa$ and, for some all-over operator $B$ in $\Hp$, $A^\prime\circ\JH=\JH\circ B$, then
\begin{equation}
\big(\JH\circ A^*\big)(\phi)=
\JH(\phi)\circ A=\big(A^\prime\circ\JH\big)(\phi)=\big(\JH\circ B\big)(\phi), \quad \forall\phi\in\Hp,
\end{equation}
so that $B=A^*\in\Bpa$ (Corollary~\ref{corad}), because the relation $\JH(\phi)\circ A=\big(\JH\circ A^*\big)(\phi)$, satisfied for all
$\phi\in\Hp$, uniquely determines the operator $A^*$ (Proposition~\ref{proadj}).
\end{proof}

\begin{notation} \label{notalim}
Given an infinite matrix $(A_{mn})\in\Minf$, by writing $\lim_{m+n}A_{mn}=\alpha$, for some $\alpha\in\Q$, we mean that
\begin{equation}\label{meana}
\mbox{$\forall\epsilon>0$, $\card(\{(m,n)\in\mathbb{N}\times\mathbb{N}\sep |A_{mn}-\alpha|\geq\epsilon\})<\infty$}.
\end{equation}
Equivalently, we mean that
\begin{equation}\label{meanb}
\mbox{$\forall\epsilon>0$, $\exists \ttN\in\mathbb{N}$, such that, if $\max\{m,n\}>\ttN$, then $|A_{mn}-\alpha|<\epsilon$,}
\end{equation}
or, also, that
\begin{equation}\label{meanc}
\mbox{$\forall\epsilon>0$, $\exists\ttN\in\mathbb{N}$, such that, if $m+n>\ttN$, then $|A_{mn}-\alpha|<\epsilon$.}
\end{equation}
\end{notation}

\begin{corollary}\label{cor.2}
Let $\Phi=\{\phi_n\}_{n\in\mathbb{N}}$ be any orthonormal basis in $\Hp$, and let $A=\opPhi(A_{mn})$ be a matrix operator.
If the elements of the matrix $(A_{mn})$ satisfy the condition that
\begin{equation}\label{eq.40}
\lim_{m+n}A_{mn}=0,
\end{equation}
then $A\in\Bpa$.
\end{corollary}

\begin{proof}
Condition~\eqref{eq.40} implies:
\begin{itemize}

\item[(i)] $\sup_{m,n}|A_{mn}|<\infty$, because for every $\epsilon>0$ we have that
the set $\{m,n\in\mathbb{N}\times\mathbb{N}\sep |A_{mn}|\geq \epsilon\}$ is finite.

\item[(ii)] $\lim_{m} A_{mn}=0$, $\forall n\in\nat$, and $\lim_n A_{mn}=0$, $\forall m\in\nat$.

\end{itemize}

By the final assertion on Theorem~\ref{th.2}, it follows that $A\in\Bpa$.
\end{proof}

\begin{definition}
We say that a bounded operator $A\in\Bp$ is \emph{self-adjoint} if
\begin{equation}
A\in\Bpa\;\;\;\;\text{and}\;\;\;\;A^*=A.
\end{equation}
\end{definition}

From Theorem~\ref{th.2} we also immediately derive the following:

\begin{corollary}
Let $\Phi=\{\phi_n\}_{n\in\mathbb{N}}$ be any orthonormal basis in $\Hp$. An adjointable bounded operator $A\in\Bpa$
--- $A=\opPhi(A_{mn})$ --- is self-adjoint iff
\begin{equation}
A_{mn}=\overline{A_{nm}}, \quad\forall m,n\in\mathbb{N}.
\end{equation}

Therefore, a matrix operator $\opPhi(A_{mn})$ is self-adjoint iff
\begin{enumerate}[label=\tt{(S\arabic*)}]

\item $A_{mn}=\overline{A_{nm}}$, $\forall m,n\in\nat$,

\item $\sup_{m,n}|A_{mn}|<\infty$,

\item $\lim_{m}A_{mn}=0$, $\forall n\in\mathbb{N}$.

\end{enumerate}
\end{corollary}

It is clear that the set of all self-adjoint bounded operators in $\Hp$ --- denoted hereafter by $\Bps$ ---
is a $\mathbb{Q}_p$-linear subspace of $\Bpa$ (by field restriction).

We conclude this section by observing that $\Bpa$ is a Banach $*$-algebra.

\begin{proposition}\label{algestruc}
The linear space $\Bpa$ is a $p$-adic Banach space and a (unital) Banach subalgebra of $\Bp$. Therefore, $\Bpa$,
endowed with the adjoining operation $A\mapsto A^*$, is a $p$-adic Banach $*$-algebra.
\end{proposition}

\begin{proof}
By Corollary~\ref{corad}, the linear subspace $\Bpa$ of $\Bp$ is a actually a subalgebra of $\Bp$,
containing the identity $\id$, and the mapping $\Bpa\ni A\mapsto A^*\in\Bpa$ is an involution.
Therefore, the only thing to be shown is that $\Bpa$ is closed in $\Bp$.
In fact, let $\{A_n\}_{n\in\nat}$ be a sequence in $\Bpa$, converging in $\Bp$:
$\lim_n A_n=A\in\Bp$ (in the norm topology). Since the adjoining operation is an \emph{isometric} involution,
then $\lim_n A_n^*=B$, for some bounded operator $B\in\Bp$ ($\{A_n^*\}_{n\in\nat}$ being a Cauchy sequence in $\Bp$).
It follows that
\begin{equation}
\langle\phi, A\psi\rangle=\lim_n\,\langle\phi,A_n\psi\rangle=\lim_n\,\langle A_n^*\phi,\psi\rangle=
\langle B\phi,\psi\rangle, \quad \forall\phi,\psi\in\Hp.
\end{equation}
Thus, $A$ is adjointable (and $B=A^*$); i.e., $\Bpa$ is a $p$-adic Banach space.
\end{proof}


\section{Unitary operators in a $p$-adic Hilbert space}
\label{sec5}

The definition of a unitary operator in a $p$-adic Hilbert space is not as simple as in the complex case.
Clearly, this is due to the fact that the relation between the norm and the inner product is not the `standard one'.
As in the complex setting, one can actually consider various (equivalent) definitions. Since, orthonormal bases and
matrix operators turn out to play a central role in the $p$-adic setting, we will introduce unitary operators as matrix
operators relating any pair of orthonormal bases; see Definition~\ref{def.u4} below. Eventually, it will be shown that
a unitary operator is nothing but an automorphism of a $p$-adic Hilbert space (Definition~\ref{ishil}).

In order to prove the main result of this section, we first need to collect a few preliminary facts.
We will work in a $p$-adic Hilbert space $\Hp$ over $\Q$, with $\dim(\Hp)=\enne$, where $\enne\in\nat$ or $\enne=\infty$;
accordingly, we will put $\senne\defi\{n\in\nat \sep n\le\enne\}$ (i.e., $\senne\equiv\nat$, for $\enne=\infty$).
As in Section~\ref{sec4}, we will assume bounded operators to be all-over, and as in the previous sections,
the term `isometry' will stand for `norm-isometry'. However, since the inner product will play a major role here,
for the sake of clarity it is worth starting with the following:

\begin{definition}\label{def.u1}
A \emph{linear operator} $A$ in $\Hp$ is called an \emph{isometry} if
\begin{equation}\label{eq.u1}
\dom(A)=\Hp\quad\text{and}\quad\|A\phi\|=\|\phi\|,\quad\forall\phi\in\Hp;
\end{equation}
i.e., if it is all-over and norm-preserving (N-preserving).
\end{definition}

\begin{lemma}\label{lem.u1}
Let $A$ be a linear operator in $\Hp$. The following facts are equivalent:

\begin{enumerate}[label=\rm{(\roman*)}]

\item \label{cond.le52i}  $A$ is a surjective isometry;

\item \label{cond.le52ii} $A$ is bounded, admits a bounded inverse $A^{-1}$ (with $\dom(A^{-1})=\Hp$) and
\begin{equation}\label{eq.u2}
\|A\|=1=\|A^{-1}\| .
\end{equation}

\end{enumerate}
\end{lemma}

\begin{proof}
Suppose that~\ref{cond.le52i} holds. Then, by~\eqref{eq.u1}, $A$ is bounded and $\|A\|=1$. Moreover, $A$ is bijective
and the linear operator $A^{-1}$ is an isometry too, because, for every $\psi\in\Hp$, there is some $\phi\in\Hp$ such
that $\psi=A\phi$ and $\|\psi\|=\|A\phi\|=\|\phi\|$. Therefore, $\|A^{-1}\psi\|=\|A^{-1}A\phi\|=\|\phi\|=\|\psi\|$;
whence, $A^{-1}$ is bounded and $\|A^{-1}\|=1=\|A\|$.

Conversely, suppose that~\ref{cond.le52ii} holds (in particular, $\ker(A)=\{0\}$). Assume that $A$ is \emph{not} an isometry.
Then, there exists some $\phi\in\Hp$, $\phi\neq 0$, such that
\begin{equation}\label{eq.u3}
0\neq \|\phi\|\neq\|A\phi\|\neq 0.
\end{equation}
Hence, we have:
\begin{equation}\label{eq.u4}
0\neq \frac{\|A\phi\|}{\|\phi\|}\neq 1.
\end{equation}

Now, if $\|A\phi\|/\|\phi\|>1$, then we would have $\|A\|>1$, which would contradict one of the hypotheses in~\ref{cond.le52ii}.
Instead, if $0<\|A\phi\|/\|\phi\|<1$, then we would have $\|A^{-1}\|>1$, because, in such a case, we should conclude that
\begin{equation}\label{eq.u5}
1<\frac{\|\phi\|}{\|A\phi\|}=\frac{\|A^{-1}\psi\|}{\|AA^{-1}\psi\|}=\frac{\|A^{-1}\psi\|}{\|\psi\|},\quad
\mbox{for some $\psi\in\Hp\setminus\{0\}$}.
\end{equation}
This, as well, would contradict one of the hypotheses of~\ref{cond.le52ii}. Therefore, the bijection $A$ must be an isometry.
\end{proof}

\begin{remark}\label{rem.u1}
By the Bounded Inverse Theorem (see, e.g., Corollary~{3.6} in~\cite{rooij1978non}, or Subsect.~{2.8} of~\cite{bosch1984non}),
if $A$ is a \emph{bijective} bounded operator in $\Hp$ --- in particular, $\dom(A)=\ran(A)=\Hp$ --- then $A^{-1}$ is bounded too.
Thus, point~\ref{cond.le52ii} in Lemma~\ref{lem.u1} can be reformulated as follows:
\begin{itemize}

\item[(ii)$^\prime$] $A$ is bounded and bijective, and $\|A\|=1=\|A^{-1}\|$.

\end{itemize}
\end{remark}

\begin{definition}\label{def.u2}
We say that a linear operator $A$ in $\Hp$ is \emph{inner-product-preserving} (in short, IP-preserving) if
\begin{equation}\label{eq.u6}
\langle A\phi, A\psi\rangle=\langle\phi, \psi\rangle,\quad\forall\,\phi,\psi\in\dom(A).
\end{equation}
\end{definition}

\begin{lemma}\label{lem.u2}
If $A$ is an IP-preserving, bounded operator in $\Hp$, admitting a bounded inverse $A^{-1}$
(equivalently, an IP-preserving, bijective bounded operator), then $A^{-1}$ is IP-preserving too:
\begin{equation}\label{eq.u7}
\langle A^{-1}\phi, A^{-1}\psi\rangle=\langle\phi,\psi\rangle,\quad\forall\phi,\psi\in\Hp.
\end{equation}
\end{lemma}

\begin{proof}
For every pair of vectors $\phi,\psi\in\Hp$, we have that $\phi=A\eta$, $\psi=A\chi$, for some $\eta,\chi\in\Hp$,
because $A$ is surjective, and
\begin{align*}
\langle A^{-1}\phi, A^{-1}\psi\rangle&=\langle A^{-1}A\eta, A^{-1}A\chi\rangle
=\langle \eta, \chi\rangle=\langle A\eta, A\chi\rangle=\langle\phi,\psi\rangle,
\end{align*}
where we have used the fact that $A$ is IP-preserving.
\end{proof}

\begin{remark}\label{rem.u2}
From the previous proof it is clear that, if $A$ is any IP-preserving, injective operator in $\Hp$, then
\begin{equation}
\langle A^{-1}\phi, A^{-1}\psi\rangle=\langle \phi,\psi\rangle,\qquad\forall \phi,\psi\in\dom(A^{-1})=\ran(A).
\end{equation}
\end{remark}

\begin{remark}\label{rem.u3}
A bounded operator admitting a bounded inverse is often called --- and we will indeed call it --- a \emph{top-linear isomorphism}.
Thus, Lemma~\ref{lem.u2} can be rephrased as follows: If $A$ is an IP-preserving top-linear isomorphism, then $A^{-1}$ is
an IP-preserving top-linear isomorphism too.
\end{remark}

\begin{proposition}\label{prop.u1}
A densely defined, IP-preserving operator is injective. If a bounded operator $A\in\Bp$ is IP-preserving, then $\|A\|\geq 1$;
in particular, if $A$ is an IP-preserving top-linear isomorphism, then
\begin{equation}\label{eq.u8}
\|A\|,\|A^{-1}\|\geq 1 .
\end{equation}
\end{proposition}

\begin{proof}
If a linear operator $A$ in $\Hp$ is IP-preserving, given any $\psi\in\ker(A)$, we have that
\begin{equation}\label{eq.u9}
0=\langle A\psi, A\phi\rangle=\langle\psi, \phi\rangle,\qquad\forall\phi\in\dom(A).
\end{equation}
Thus, if $A$ is densely defined --- i.e., if $\overline{\dom(A)}^{\hspace{0.6mm}\|\cdot\|}=\Hp$ --- then by the continuity
of the inner product (w.r.t.\ each of its arguments), we conclude that, actually, $\langle\psi,\phi\rangle=0$, $\forall\phi\in\Hp$;
hence, the Hermitian sesquilinear form $\scap$ being non-degenerate, $\psi=0$. Therefore, $\ker(A)=\{0\}$ and $A$ is injective.

Now, suppose that $A$ is bounded and IP-preserving. Applying the latter property and the Cauchy-Schwarz inequality, we find that
\begin{equation}\label{eq.u10}
|\langle\phi,\psi\rangle|=|\langle A\phi,A\psi\rangle|\leq \|A\phi\|\;\|A\psi\|\leq \|A\|^2\;\|\phi\|\;\|\psi\|,
\end{equation}
for all $\phi,\psi\in\Hp$. Setting $\psi=\phi$ in relation~\eqref{eq.u10}, and choosing this vector $\phi$ in such a way that
$|\langle\phi,\phi\rangle|=1=\|\phi\|$ (e.g. an element of an orthonormal basis in $\Hp$), we conclude that
$\|A\|^2\geq 1$; hence $\|A\|\geq 1$. Finally, if $A$ is an IP-preserving top-linear isomorphism, then,
by Lemma~\ref{lem.u2}, $A^{-1}$ enjoys the same property, so that both inequalities in~\eqref{eq.u8} hold true.
\end{proof}

We will now prove that, under mild conditions, an IP-preserving operator is a top-linear isomorphism.

\begin{theorem}\label{th.u1}
A surjective, IP-preserving, all-over operator $A$ in $\Hp$ is an adjointable top-linear isomorphism and $A^*=A^{-1}$; moreover,
\begin{equation}\label{eq.u11}
\|A\|=\|A^*\|=\|A^{-1}\|\geq 1.
\end{equation}
\end{theorem}

\begin{proof}
Since $A$ is IP-preserving and $\dom(A)=\Hp$, by Proposition~\ref{prop.u1} it is injective.
Let us prove that $A$ is bounded (equivalently continuous). By the Closed Graph Theorem
(see, e.g., Theorem~{3.5} in~\cite{rooij1978non}, or Subsect.~{2.8} of~\cite{bosch1984non}),
it is sufficient to show that $A$ is a closed operator. Let $\{\chi_n\}_{n\in\nat}$ be a sequence
in $\Hp$ such that
\begin{equation}\label{eq.u12}
\mbox{$\displaystyle \lim_n\chi_n=0$ and $\displaystyle \lim_nA\chi_n=\phi$, for some $\phi\in\Hp$}.
\end{equation}
In order to conclude that $A$ is closed, we need to show that $\phi=0$. Indeed, by the continuity of the scalar product
(w.r.t.\ each of its arguments) and by the fact that $A$ is IP-preserving, we have:
\begin{equation}\label{eq.u13}
\langle \phi, A\psi\rangle=\langle{\textstyle\lim_n} A\chi_n, A\psi\rangle=\lim_n\,\langle A\chi_n, A\psi\rangle
=\lim_n\,\langle \chi_n, \psi\rangle=\langle{\textstyle\lim_n}\chi_n,\psi\rangle=0,\quad\forall\psi\in\Hp.
\end{equation}
Since $A$ is surjective, we conclude that $\langle\phi,\eta\rangle=0$, $\forall\eta\in\Hp$; hence,
the Hermitian sesquilinear form $\scap$ being non-degenerate, $\phi=0$, so that $A$ is closed.

Summarizing, a surjective, IP-preserving, all-over operator $A$ is bijective and bounded;
hence, by the Bounded Inverse Theorem, a top-linear isomorphism.

Let us now show that $A\in\Bpa$ and $A^*=A^{-1}$. In fact, by Lemma~\ref{lem.u2},
the bounded operator $A^{-1}$ is IP-preserving too; hence:
\begin{equation}\label{eq.u14}
\langle\phi, A\psi\rangle=\langle A^{-1}\phi, A^{-1}A\psi\rangle=\langle A^{-1}\phi, \psi\rangle,\qquad \forall\phi,\psi\in\Hp.
\end{equation}
Therefore, $A$ is adjointable and $A^*=A^{-1}$, so that $\|A\|=\|A^*\|=\|A^{-1}\|$ and, by~\eqref{eq.u8} in Proposition~\ref{prop.u1},
relation~\eqref{eq.u11} holds true.
\end{proof}

\begin{notation}
Given two vectors $\phi,\psi\in\Hp$, by writing
\begin{equation}\label{eq.u15}
\phi\nperp\psi,
\end{equation}
we mean that $\phi$ and $\psi$ are \emph{norm-orthogonal} each other (recall from Subsection~\ref{ipbs} that $\phi\perp\psi$ means
that $\phi$ and $\psi$ are IP-orthogonal, instead); i.e., that $\|\alpha \phi+\beta \psi\|=\max\{\|\alpha \phi\|,\|\beta \psi\|\}$,
for all $\alpha,\beta\in\Q$.
\end{notation}

\begin{definition}\label{def.u3}
A linear operator $A$ in $\Hp$ is said to be \emph{norm-orthogonality-preserving} (in short, NO-preserving) if
\begin{equation}\label{eq.u16}
\phi,\psi\in\dom(A),\ \phi\nperp\psi\quad\implies\quad A\phi\nperp A\psi.
\end{equation}
\end{definition}

\begin{theorem}\label{th.u2}
Every all-over, NO-preserving operator in $\Hp$ is bounded. Specifically, every all-over, NO-preserving operator
in $\Hp$ is a nonzero scalar multiple of an isometry and, conversely, a nonzero scalar multiple of an isometry is NO-preserving.
In particular, a linear operator $A$ in $\Hp$ is an isometry if and only if $A$ is an all-over, NO-preserving (hence, bounded) operator
such that $\|A\|=1$.
\end{theorem}

\begin{proof}
Since $\|\Hp\|\defi\{\|\phi\|\sep \phi\in\Hp\}=|\Q|$, the `ramification index' of $\Hp$ is equal to $1$,
so that we can apply Corollary $1.3$ in \cite{shilkret1972orthogonal} (actually, the first assertion of the
theorem follows from Corollary~{1.1} \emph{ibidem}, and does not require the mentioned property of $\Hp$).
\end{proof}

We can now introduce the unitary operators in the $p$-adic setting and provide a suitable characterization of
this class of operators.

\begin{definition}\label{def.u4}
A matrix operator in $\Hp$ of the form
\begin{equation}
U=\opPhi(\langle\phi_m,\psi_n\rangle)
\end{equation}
--- where $\Phi\equiv\{\phi_m\menne$, $\Psi\equiv\{\psi_n\nenne$ are orthonormal bases in $\Hp$ --- is called
a \emph{unitary operator}. We will denote the set of all such operators in $\Hp$ by $\mathcal{U}(\Hp)$.
\end{definition}

The set $\mathcal{U}(\Hp)$ is characterized by the following result:

\begin{theorem}\label{th.u3}
Given a linear operator $U$ in $\Hp$, the following facts are equivalent:

\begin{enumerate}[label=\tt{(U\arabic*)}]

\item $U$ \label{cond.th5u1}is a unitary operator --- i.e., $U=\opPhi(\langle\phi_m,\psi_n\rangle)$ --- for some pair
of orthonormal bases $\Phi\equiv\{\phi_m\menne$ and $\Psi\equiv\{\psi_n\nenne$ in $\Hp$;

\item \label{cond.th5u2} $U\in\Bp$ and, for some pair of orthonormal bases $\Phi\equiv\{\phi_m\menne$ and
$\Psi\equiv\{\psi_n\nenne$, $U\phi_k=\psi_k$, $\forall k\in\senne$;

\item \label{cond.th5u3} $U\in\Bpa$, $\|U\|=1$ and $UU^*=\mathrm{Id}=U^*U$;

\item \label{cond.th5u4} $U$ is a surjective IP-preserving, all-over (hence, bounded) operator and $\|U\|=1$;

\item \label{cond.th5u5} $U$ is an IP-preserving top-linear isomorphism and $\|U\|=1=\|U^{-1}\|$;

\item \label{cond.th5u6} $U$ is an automorphism of the $p$-adic Hilbert space $\Hp$, namely, an IP-preserving surjective isometry;

\item \label{cond.th5u7} $U$ is a surjective, IP-preserving, NO-preserving, all-over operator;

\item \label{cond.th5u8} $U$ is bounded and transforms orthonormal bases into orthonormal bases.
\end{enumerate}
\end{theorem}

\begin{proof}
We will first show that~\ref{cond.th5u1} $\iff$~\ref{cond.th5u2} $\implies$~\ref{cond.th5u3}.

Note that
\begin{equation}\label{eq.u18}
\mbox{$|\langle\phi_m,\psi_n\rangle|\leq \|\phi_m\|\;\|\psi_n\|=1$, $\forall m,n\in\senne$, and,
for $\enne=\infty$, $\displaystyle\lim_m\,\langle\phi_m,\psi_n\rangle=0$, $\forall n\in\nat$.}
\end{equation}
Hence, by Theorem~\ref{th.1}, we have that $\opPhi(\langle\phi_m,\psi_n\rangle)\in\Bp$ and, moreover,
\begin{equation}\label{eq.u19}
\opPhi(\langle\phi_m,\psi_n\rangle)\phi_k=\sum_m\langle\phi_m,\psi_k\rangle\phi_m=\psi_k,\quad\forall k\in\senne.
\end{equation}
Therefore, \ref{cond.th5u1} $\implies$~\ref{cond.th5u2}.

Conversely, if $U$ satisfies~\ref{cond.th5u2}, then
\begin{equation}\label{eq.u20}
U=\opPhi(\langle\phi_m, U\phi_n\rangle)=\opPhi(\langle\phi_m,\psi_n\rangle),
\end{equation}
where in the first equality, we have used the expression of a bounded matrix operator (w.r.t. any orthonormal basis).
Thus,~\ref{cond.th5u1} holds true.

Now, given a unitary operator $U=\opPhi(\langle\phi_m,\psi_n\rangle)\in\mathcal{U}(\Hp)\subset\Bp$, since, for $\enne=\infty$,
$\lim_n\langle\phi_m,\psi_n\rangle=0$, for all $m\in\nat$, then, by the last assertion of Theorem~\ref{th.2},
we conclude that $U$ is adjointable; i.e., $\mathcal{U}(\Hp)\subset\Bpa$ as well. Moreover, as previously shown,
the unitary operator $U=\opPhi(\langle\phi_m,\psi_n\rangle)$ is completely determined by condition~\ref{cond.th5u2};
therefore:
\begin{equation}\label{eq.u21}
\langle\phi_m, U^*\psi_k\rangle=\langle U\phi_m,\psi_k\rangle=\langle\psi_m,\psi_k\rangle=\delta_{mk}.
\end{equation}
Thus, $U^*\psi_k=\phi_k$, $\forall k\in\senne$, and hence --- noting that: $A\in\Bp$, $A\phi_m=\phi_m$,
$\forall m\in\senne$ (where $\{\phi_m\menne$ is any orthonormal basis) $\implies$ $A=\mathrm{Id}$ --- we have:
\begin{equation}\label{eq.u22}
\mbox{$U^*U=\mathrm{Id}=UU^*$; i.e., $U^*=U^{-1}$}.
\end{equation}
Also note that $\sup_{m}\langle\phi_m,\psi_n\rangle=\|\psi_n\|=1$, $\forall n\in\senne$; hence:
\begin{equation}\label{eq.u23}
1=\sup_{m,n}\,\langle\phi_m,\psi_n\rangle=\|U\|=\|U^*\|=\|U^{-1}\|.
\end{equation}
Thus, if $U\in\mathcal{U}(\Hp)$, then $U$ satisfies condition~\ref{cond.th5u3}.

Next, it is clear that~\ref{cond.th5u3} $\implies$~\ref{cond.th5u4}, because, if the conditions in~\ref{cond.th5u3}
are satisfied, then $U$ is a surjective (adjointable) bounded operator and
\begin{equation}
\langle U\chi, U\eta\rangle=\langle U^*U\chi, \eta\rangle=\langle\chi, \eta\rangle,\qquad\forall\chi, \eta\in\Hp;
\end{equation}
i.e., $U$ is IP-preserving. Moreover, by Theorem~\ref{th.u1},~\ref{cond.th5u4} $\implies$~\ref{cond.th5u5}.
Also, if $U$ satisfies~\ref{cond.th5u5}, then, by Lemma~\ref{lem.u1}, $U$ is a surjective isometry (and IP-preserving);
i.e., \ref{cond.th5u5} $\implies$~\ref{cond.th5u6}.

Let us now prove that~\ref{cond.th5u6} $\iff$~\ref{cond.th5u7}. In fact, by the second assertion
of Theorem~\ref{th.u2}, if $U$ is an isometry, then it is a NO-preserving (all-over) operator;
hence: \ref{cond.th5u6} $\implies$~\ref{cond.th5u7}. Conversely, if $U$ is a NO-preserving, all-over operator,
then (again by the second assertion of Theorem~\ref{th.u2}) $U$ is a non-zero scalar multiple of an isometry:
$U=zJ$, with $z\in\Q\setminus\{0\}$. Now, if, moreover, $U$ is IP-preserving and surjective, then by Theorem~\ref{th.u1},
$\|U\|=\|U^{-1}\|$. Thus, $J$ is a surjective isometry and
\begin{equation}
|z|=\|zJ\|=\|U\|=\|U^{-1}\|=\|z^{-1}J^{-1}\|=|z|^{-1}\;\;\implies\;\; |z|=1.
\end{equation}
Therefore, $U=zJ$ is an IP-preserving, surjective isometry; i.e.,~\ref{cond.th5u7} $\implies$~\ref{cond.th5u6}, as well.

At this point, let us observe that~\ref{cond.th5u6} implies~\ref{cond.th5u8}. Indeed, if $U$ is an IP-preserving, surjective isometry,
then, given any orthonormal basis $\{\phi_m\menne$ in $\Hp$, and, putting $\psi_n=U\phi_n$, $\forall n\in\senne$, we obtain another
orthonormal basis $\{\psi_n\nenne$, because
\begin{equation}\label{eq.u25}
\langle\psi_j,\psi_k\rangle=\langle U\phi_j, U\phi_k\rangle=\langle\phi_j,\phi_k\rangle=\delta_{jk};
\end{equation}
in addition, for every set $\{z_n\nenne\subset\Q$ --- converging to $0$, if $\enne=\infty$ ---
\begin{equation}\label{eq.u26}
\big\|{\textstyle\sum_n} z_n\psi_n\big\|=\big\|U^{-1}{\textstyle\sum_n} z_n\psi_n\big\|=\big\|{\textstyle\sum_n} z_n\phi_n\big\|=\max_n|z_n|,
\end{equation}
where we have used the fact that $U^{-1}$ is an isometry, and, for every $\chi\in\Hp$,
\begin{equation}\label{eq.u27}
\chi=U(U^{-1}\chi)=\sum_n\langle\phi_n, U^{-1}\chi\rangle\, U\phi_n=\sum_n
\langle U\phi_n, UU^{-1}\chi\rangle\, U\phi_n=\sum_n\langle\psi_n,\chi\rangle\psi_n.
\end{equation}
Thus, by~\eqref{eq.u26} and~\eqref{eq.u27}, $\{\psi_n\nenne$ is a normal basis, and specifically, by~\eqref{eq.u25}, it is
orthonormal.

Finally, it is obvious that~\ref{cond.th5u8} $\implies$~\ref{cond.th5u2}, and this observation completes the proof,
since overall we have shown that: \ref{cond.th5u1} $\iff$~\ref{cond.th5u2} $\implies$~\ref{cond.th5u3} $\implies$~\ref{cond.th5u4}
$\implies$~\ref{cond.th5u5} $\implies$~\ref{cond.th5u6} $\implies$~\ref{cond.th5u8} $\implies$~\ref{cond.th5u2}, and,
moreover, \ref{cond.th5u6} $\iff$~\ref{cond.th5u7}.
\end{proof}

\begin{remark}\label{rem.4u}
One can easily check that
\begin{equation}\label{eq.u28}
U\defi\opPhi(\langle\phi_m,\psi_n\rangle)=\sum_k|\psi_k\rangle\langle\phi_k|=\opPsi(\langle\phi_m,\psi_n\rangle),
\end{equation}
--- where, if $\enne=\infty$, the series converges w.r.t.\ the strong operator topology --- and
\begin{equation}\label{eq.u29}
U^*=\opPhi(\langle\psi_m,\phi_n\rangle)=\sum_k|\phi_k\rangle\langle\psi_k|
=\opPsi(\langle\psi_m,\phi_n\rangle)=U^{-1}\in\mathcal{U}(\Hp).
\end{equation}
\end{remark}

\begin{remark}\label{rem.u5}
By the characterization~\ref{cond.th5u6} of $\mathcal{U}(\Hp)$,  it is clear that $\mathcal{U}(\Hp)$ is, in a natural way,
a group. In fact, the product (composition) of two unitary operators is unitary and $\mathrm{Id}\in\mathcal{U}(\Hp)$.
Moreover, by~\eqref{eq.u29} --- or, say, by~\ref{cond.th5u6} and Lemma~\ref{lem.u2} (if $U$ is an IP-preserving surjective isometry,
then $U^{-1}$ shares the same property) --- if $U\in\mathcal{U}(\Hp)$, then $U^{-1}=U^*\in\mathcal{U}(\Hp)$ too. Let us also observe
that the \emph{unitary group} of the $p$-adic Hilbert space $\Hp$ is the intersection of two other remarkable groups, i.e.,
\begin{equation}\label{eq.u30}
\mathcal{U}(\Hp)=\ipp(\Hp)\cap\nop(\Hp),
\end{equation}
where:
\begin{itemize}

\item $\ipp(\Hp)\subset\Bpa$ is the group of all surjective, IP-preserving, all-over operators (note that, by Theorem~\ref{th.u1}
and Lemma~\ref{lem.u2}, if $A\in\ipp(\Hp)$, then $A^{-1}=A^*\in\ipp(\Hp)$ too);

\item $\nop(\Hp)$ is the group of all surjective, NO-preserving, all-over operators --- equivalently, the group of all non-zero
scalar multiples of surjective isometries (Theorem~\ref{th.u2}).

\end{itemize}
\end{remark}

\begin{remark}\label{rem.u6}
In the case where $\Hp$ is finite-dimensional --- $\dim(\Hp)=\enne\in\nat$; hence, $\Bp=\Bpa$ is just the set $\linop$ of
all linear operators in $\Hp$, and $UU^*=\id$ iff $U^*U=\id$ --- the characterization~\ref{cond.th5u3} of $\mathcal{U}(\Hp)$
provides a simple description of the unitary group of $\Hp$ as a matrix group, i.e.,
\begin{equation} \label{carunops}
\textstyle
\mathcal{U}(\Hp)=\big\{\opPhi(U_{mn})\sep \sum_{n=1}^\enne U_{ln}\overline{U_{mn}}=\delta_{lm},\;
\max_{m,n}|U_{mn}|=1\big\},
\end{equation}
where $\Phi=\{\phi_m\menne$ is any orthonormal basis in $\Hp$. It is worth observing that here the condition
\begin{equation}\label{eq.u31}
\|\opPhi(U_{mn})\|=\max_{m,n}|U_{mn}|=1
\end{equation}
\emph{cannot be dispensed with} (unlike the complex case). We illustrate this point by means of an explicit example.

Assume that $\Hp$ is a $p$-adic Hilbert space, with $p\neq 2$ and $\dim(\Hp)=4$. As shown in the proof of Proposition~{5.3}
in~\cite{albeverio1999non}, there exists a solution $x_1,\dots,x_4$ of the equation
\begin{equation}\label{eq.u32}
x_1^2+x_2^2+x_3^2+x_4^2=p^{2\ka},\qquad x_1,\dots,x_4\in\mathbb{Z},
\end{equation}
--- for any $\ka\in\mathbb{N}$ --- satisfying the condition that
\begin{equation}\label{eq.u33}
\max_i|x_i|=1.
\end{equation}
Consider, then, the matrix (with rational coefficients)
\begin{equation}
(A_{mn})=\frac{1}{p^\ka}
\begin{pmatrix}
x_1 & x_2 & x_3 & x_4\\
-x_2 & x_1 & -x_4 & x_3\\
-x_4 & -x_3 & x_2 & x_1\\
-x_3 & x_4 & x_1 & -x_2
\end{pmatrix},
\end{equation}
where $x_1,\dots,x_4$ is the aforementioned  solution of~\eqref{eq.u32}--\eqref{eq.u33}. Clearly, we have:
\begin{equation}
\mbox{$\sum_{n=1}^4 A_{ln}\overline{A_{mn}}=\sum_{n=1}^4 A_{ln}A_{mn}=\delta_{lm}$; but $\max_{m,n}|A_{mn}|=p^\ka\neq 1$.}
\end{equation}
Thus, $A=\opPhi(A_{mn})\in\ipp(\Hp)$ (because $A^*=A^{-1}$, hence, $A$ is IP-preserving), \emph{but} $A\notin\nop(\Hp)$,
because
\begin{equation}
\|A\|=\|A^*\|=\|A^{-1}\|=p^\ka >1,
\end{equation}
so that $A$ cannot be a non-zero scalar multiple of an isometry (in such a case, we should have that $\|A^{-1}\|=\|A\|^{-1}$).
Therefore, $A\notin\mathcal{U}(\Hp)$. Otherwise stated, $A$ cannot be unitary, since it preserves the inner product,
but not the norm-orthogonality.
\end{remark}

\begin{remark}\label{rem.u7}
Let us observe explicitly that the group $\ipp(\Hp)$ admits a further characterization; namely,
\begin{equation}\label{eq.u36}
\ipp(\Hp)=\big\{A\in\Bpa\sep \mbox{$A$ bijective and $A^*=A^{-1}$}\big\}.
\end{equation}
In fact, by Theorem~\ref{th.u1}, $\ipp(\Hp)$ is contained in the set defined on the right hand side of~\eqref{eq.u36}.
Conversely, it is clear that every bijective operator $A\in\Bpa$, such that $A^*=A^{-1}$, is IP-preserving:
\begin{equation}\label{eq.u37}
\langle A\eta, A\chi\rangle=\langle A^*A\eta, \chi\rangle=\langle\eta,\chi\rangle,\qquad\forall\eta,\chi\in\Hp.
\end{equation}
Therefore, relation~\eqref{eq.u36} holds true. As a consequence, we obtain a simple description of the
unitary group $\mathcal{U}(\Hp)$. Indeed, note that, by~\eqref{eq.u36} and by the characterization~\ref{cond.th5u3}
of a unitary operator, we have:
\begin{equation}\label{eq.u38}
\mathcal{U}(\Hp)=\ipp(\Hp)\cap\usp=\ipp(\Hp)\cap\uba,
\end{equation}
where $\usp$ and $\uba$ are, respectively, the \emph{unit sphere} and that \emph{unit ball} in $\Bp$; i.e.,
\begin{equation}\label{eq.u39}
\usp\defi\{A\in\Bp\sep \|A\|=1\},\quad \uba\defi\{A\in\Bp\sep \|A\|\le1\}.
\end{equation}
The first equality in~\eqref{eq.u38} corresponds to the characterization~\ref{cond.th5u4} of a $p$-adic unitary operator,
and the second equality follows from~\eqref{eq.u11} in Theorem~\ref{th.u1}, according to which $\|\ipp(\Hp)\|\subset[1,\infty)$.
\end{remark}

\begin{remark}\label{rem.u8}
Considering again the case where $\Hp$ is finite-dimensional --- $\dim(\Hp)=\enne\in\nat$ and $\Bp=\linop=\Bpa$ ---
using elementary methods of matrix analysis one can prove that, in this case, the group $\nop(\Hp)$ admits
the following further characterization:
\begin{equation}
\nop(\Hp)=\big\{A\in\linop\sep \|A\|=1=|\det(A)|\big\},
\end{equation}
where $\det(A)$ is the determinant of the representative matrix of $A$ w.r.t.\ any basis in the finite-dimensional
vector space $\Hp$. Clearly, if $A\in\mathcal{U}(\Hp)$, then $|\det(A)|=1$ automatically, because
$|\det(AA^\ast)|=\big|\det(A)\,\overline{\det(A)}\,\big|=|\det(A)|^2$.
\end{remark}

\begin{example}
Let us consider the case where $p=2$ and $\mu=14$; i.e., $\Hp$ is a $p$-adic Hilbert space over $\mathbb{Q}_2(\sqrt{14})$.
Let us assume that $\dim(\Hp)=2$, and, given an orthonormal basis $\Phi=\{\phi_1,\phi_2\}$ in $\Hp$, let us consider
a linear operator $U=\opPhi(U_{mn})$. By~\eqref{carunops}, $U$ is unitary iff
\end{example}
\begin{equation}
(U_{mn})=
\begin{pmatrix}
a & b \\
c & d
\end{pmatrix}, \
\mbox{where: $a\hspace{0.3mm}\overline{a}+b\hspace{0.3mm}\overline{b}=1=c\hspace{0.3mm}\overline{c}+d\hspace{0.3mm}\overline{d}$,
$a\hspace{0.3mm}\overline{c}+b\hspace{0.3mm}\overline{d}=0$, $\max\{|a|,|b|,|c|,|d|\}=1$.}
\end{equation}
To satisfy this condition, we can put, e.g., $a=\sqrt{-7}^2$, $b=\frac{2}{a}\sqrt{14}$, $c=b$ and $d=a$, where $\sqrt{-7}^2$
is any of the two $2$-adic square roots of $-7=1+0\cdot 2+0\cdot 2^2+1\cdot 2^3+1\cdot 2^4+\cdots\in(\mathbb{Q}_2^\ast)^2$.
Therefore, $\Psi\equiv\{\psi_1=a\phi_1+b\phi_2,\psi_2=b\phi_1+a\phi_2\}$ is another orthonormal basis in $\Hp$.


\section{The trace class of a $p$-adic Hilbert space}
\label{sec6}

In this section, we will introduce a suitable notion of trace class operator in a $p$-adic Hilbert space $\Hp$.
As in Section~\ref{sec4}, we will assume that $\dim(\Hp)=\infty$ (and we will use the notations adopted
therein), because in the finite-dimensional case the notion of trace introduced here becomes completely
analogous to the notion of trace of a linear operator in a finite-dimensional complex Hilbert space
and the results of this section hold true with obvious modifications.


\subsection{Traceable operators}
\label{preliminaries}

We start with the following:

\begin{definition}\label{def.3}
Let $\Phi=\{\phi_m\mnat$ be an orthonormal basis in $\Hp$, and let $T$ be a (densely defined) linear operator in $\Hp$
such that $\Phi\subset\dom(T)$. We say that the operator $T$ is \emph{traceable} w.r.t.\ $\Phi$ if the series
\begin{equation}\label{eq.43}
\sum_m\langle\phi_m, T\phi_m\rangle
\end{equation}
is convergent. Namely, if $\lim_m\langle\phi_m, T\phi_m\rangle=0$ (see Proposition~\ref{sumlemma}).
\end{definition}

\begin{proposition}\label{prop.1}
A matrix operator $T=\opPhi(T_{mn})$ is (such that $\Phi\subset\dom(T)$ and) traceable w.r.t.\ $\Phi$ iff
\begin{equation}\label{eq.44}
\lim_m T_{mn}=0,\;\;\forall n\in\mathbb{N},\;\;\text{and}\;\;\lim_m T_{mm}=0.
\end{equation}
\end{proposition}

\begin{proof}
It is easy to see that $\Phi\subset\dom(T)$, with $T=\opPhi(T_{mn})$, iff $\lim_m T_{mn}=0$, $\forall n\in\nat$.
Indeed, recalling~\eqref{eq.7} and~\eqref{eq.8}, if $\phi_n\in\dom(T)$, then $T\phi_n=\sum_m T_{mn}\phi_m$
(hence, $\{T_{mn}\mnat\in c_0$) and, conversely, if $\lim_{m}T_{mn}=0$, then $\phi_n\in\dom(T)$.

Moreover, if $\phi_n\in\dom(T)$, then ($T\phi_n=\sum_m T_{mn}\phi_m$ and)
\begin{equation}\label{eq.45}
\langle \phi_n, T\phi_n\rangle=\sum_m T_{mn}\langle \phi_n,\phi_m\rangle=T_{nn}.
\end{equation}

Therefore, if $\Phi\subset\dom(T)$, then $\lim_m T_{mn}=0$, $\forall n$, and, if, moreover, the series~\eqref{eq.43}
is convergent, then
\begin{equation}\label{eq.46}
\lim_m T_{mm}=\lim_m\,\langle\phi_m, T\phi_m\rangle=0.
\end{equation}
Conversely, if both conditions in~\eqref{eq.44} hold true, then $\Phi\subset\dom(T)$, and
$\lim_m\langle\phi_m, T\phi_m\rangle=\lim_m T_{mm}=0$, so that the series~\eqref{eq.43}
converges; i.e., $T=\opPhi(T_{mn})$ is traceable w.r.t.\ $\Phi$.
\end{proof}

\begin{remark}\label{rem.5}
By Proposition~\ref{prop.1} and by the characterization of matrix elements of a bounded operator
(see~\eqref{eq.13} in Theorem~\ref{th.1}), it is clear that one can construct matrix operators in
$\Hp$ ($\mathrm{dim}(\Hp)=\infty$) that are traceable w.r.t.\ a given orthonormal basis in $\Hp$,
but \emph{not} bounded (and, thus, not all-over). Precisely, one has to take any matrix operator
$T=\opPhi(T_{mn})$ satisfying both conditions in~\eqref{eq.44} and such that $\sup_{m,n}|T_{mn}|=\infty$.
\end{remark}

The previous remark motivates us to consider a smaller class of matrix operators for the definition
of the trace class of $\Hp$.

\begin{definition}\label{def.4}
Let $\Phi\equiv\{\phi_m\mnat$ be an orthonormal basis in $\Hp$. We introduce the following set of matrix operators:
\begin{equation}
\mathcal{T}_\Phi(\Hp)\defi\big\{\opPhi(T_{mn})\sep \mbox{$T_{mn}\in\Minf$ s.t.\
$\lim_{m+n}T_{mn}=0$}\big\}.
\end{equation}
\end{definition}

\begin{remark}\label{rem.6}
Recalling Notation~\ref{notalim}, the limit $\lim_{m+n}T_{mn}=0$ means that
\begin{equation}\label{eq.47}
\mbox{$\forall\epsilon>0$, $\card(\{(m,n)\in\mathbb{N}\times\mathbb{N}\sep |T_{mn}|\geq\epsilon\})<\infty$}.
\end{equation}
Equivalently, $\lim_{m+n}T_{mn}=0$ means that
\begin{equation}\label{eq.47b}
\mbox{$\forall\epsilon>0$, $\exists\ttN\in\mathbb{N}$, such that, if $\max\{m,n\}>\ttN$, then $|T_{mn}|<\epsilon$,}
\end{equation}
or, also, that
\begin{equation}\label{eq.48}
\mbox{$\forall\epsilon>0$, $\exists\ttN\in\mathbb{N}$, such that, if $m+n>\ttN$, then $|T_{mn}|<\epsilon$.}
\end{equation}
Moreover, conditions~\eqref{eq.47}--\eqref{eq.48} are equivalent to assuming that the double series
$\sum_{m,n}T_{mn}$ is convergent, where
\begin{equation}\label{eq.49}
\sum_{m,n}T_{mn}=\lim_{\ttN\rightarrow\infty}\bigg(\sum_{m=1}^{\ttN}\sum_{n=1}^{\ttN}T_{mn}\bigg).
\end{equation}
It is a remarkable fact that, given  a double sequence $\{x_{mn}\}_{m,n\in\nat}$ in $\Q$,
if $\lim_{m+n}x_{mn}=0$, then both the iterated series
\begin{equation}
\sum_m\sum_n x_{mn}\;\;\;\;\text{and}\;\;\;\;\sum_n\sum_m x_{mn}
\end{equation}
converge and
\begin{equation}\label{eq.50}
\sum_m\sum_n x_{mn}=\sum_n\sum_m x_{mn}=\sum_{m,n}x_{mn}.
\end{equation}
Therefore, if $\lim_{m+n}T_{mn}=0$, then the convergent double series~\eqref{eq.49} can be expressed as
an iterated series.

Another useful fact is that, given a double sequence $\{x_{mn}\}_{m,n\in\nat}\subset\Q$,
\begin{align} \label{limiff}
\lim_{m+n}x_{mn}=0\quad & \iff \quad
\begin{cases}
\lim_m x_{mn}=0,\ \forall n\in\nat,\quad\lim_n x_{mn}=0,\ \forall m\in\nat,\\
\text{and}\ \lim_{m,n}x_{mn}=0\quad \mbox{(Pringsheim limit)}\footnotemark
\end{cases}
\\ \label{limiffbis}
& \iff \quad \mbox{$\lim_m x_{mn}=0$, $\forall n\in\nat$, and $\lim_n x_{mn}=0$, \emph{uniformly} in $m\in\nat$.}
\end{align}
In relation~\eqref{limiffbis}, the expression ``$\lim_n x_{mn}=0$, \emph{uniformly} in $m\in\nat$'' means:
for every $\epsilon>0$, $\exists\ttN\in\mathbb{N}$ such that, for $n>\ttN$ and all $m\in\nat$, $|x_{mn}|<\epsilon$.
\footnotetext{Namely, $\forall\epsilon>0$, $\exists\ttN\in\mathbb{N}$ such that, if $m,n>\ttN$,
then $|x_{mn}|<\epsilon$.}

For the previous claims, see p.~{62} of~\cite{schikhof2007ultrametric}, Exercise~{23.B}, and Chapt.~{8} of~\cite{natarajan2014introduction}.
\end{remark}

If a linear operator $T$ in $\Hp$ is traceable w.r.t.\ an orthonormal basis $\Phi=\{\phi_m\mnat$,
we denote the sum of the series~\eqref{eq.43} by the symbol $\mathrm{tr}_{\Phi}(T)$.

\begin{proposition}\label{prop.t2}
If $T=\opPhi(T_{mn})\in\mathcal{T}_\Phi(\Hp)$, then it is traceable w.r.t.\ $\Phi$ --- in particular,
$\Phi\subset\dom(T)$ --- and
\begin{equation}\label{eq.51}
\tr_\Phi(T)\defi\sum_m\langle\phi_m, T\phi_m\rangle=\sum_m T_{mm}.
\end{equation}
\end{proposition}

\begin{proof}
Observe that
\begin{equation}\label{eq.52}
\card(\{m\in\mathbb{N}\sep |T_{mm}|\ge\epsilon\})\leq
\card(\{(m,n)\in\mathbb{N}\times\mathbb{N}\sep |T_{mn}|\ge\epsilon\}).
\end{equation}
Thus, recalling Remark~\ref{rem.6}, we argue that
\begin{equation}
T=\opPhi(T_{mn})\in\mathcal{T}_\Phi(\Hp)\;\;\iffdef\;\;\lim_{m+n}T_{mn}=0\;\;\implies\lim_m T_{mm}=0,
\end{equation}
and
\begin{equation}
\lim_{m+n}T_{mn}=0\;\;\implies\;\;\lim_m T_{mn}=0,\;\;\forall n\in\mathbb{N}.
\end{equation}
Thus, if $T=\opPhi(T_{mn})$ belongs to $\mathcal{T}_\Phi(\Hp)$, then $\Phi\subset\dom(T)$
(see the proof of Proposition~\ref{prop.1}) and $T$ is traceable w.r.t.\ $\Phi$. Moreover,
$\langle\phi_m, T\phi_m\rangle=T_{mm}$, for every $m\in\mathbb{N}$ (see~\eqref{eq.45});
hence, relation~\eqref{eq.51} holds true.
\end{proof}

We now provide a more precise characterization of the set of matrix operators $\mathcal{T}_\Phi(\Hp)$.

\begin{proposition}\label{prop.t3}
Let $\Phi\equiv\{\phi_n\nnat$ be any orthonormal basis in $\Hp$. Then, the following facts are equivalent:

\begin{enumerate}[label=\tt{(T\arabic*)}]

\item \label{cond.p6t1}  $T\in\Bpa$ and $\lim_{m+n}\langle \phi_m, T\phi_n\rangle=0$;

\item \label{cond.p6t2}  $T\in\Bp$ and $\lim_{m+n}\langle \phi_m, T\phi_n\rangle=0$;

\item \label{cond.p6t3}  $T\in\mathcal{T}_\Phi(\Hp)$.

\end{enumerate}
\end{proposition}

\begin{proof}
It is obvious that \ref{cond.p6t1} $\implies$~\ref{cond.p6t2}. Also, if $T\in\Bp$, then
$T=\opPhi(T_{mn})$, where $T_{mn}=\langle\phi_m,T\phi_n\rangle$; hence \ref{cond.p6t2} $\implies$~\ref{cond.p6t3}.
Next, if $T\in\mathcal{T}_\Phi(\Hp)$, then by Corollary~\ref{cor.2}, we have that $T\in\Bpa$.
Moreover, since $T$ is bounded, $T=\opPhi(\langle \phi_m, T\phi_n\rangle)$.
Hence, \ref{cond.p6t3} $\implies$~\ref{cond.p6t1}, and the proof is complete.
\end{proof}

\begin{corollary}\label{cor.4}
$\mathcal{T}_\Phi(\Hp)$ is a linear subspace of $\Bpa$ and
\begin{equation}\label{eq.53}
T\in\mathcal{T}_\Phi(\Hp)\implies T^*\in\mathcal{T}_\Phi(\Hp).
\end{equation}
\end{corollary}

\begin{proof}
Since~\ref{cond.p6t1} $\iff$~\ref{cond.p6t3} in Proposition~\ref{prop.t3} --- therefore, $\mathcal{T}_\Phi(\Hp)\subset\Bpa\subset\Bp$
--- it is sufficient to note that, for $S,T\in\Bp$, $\phi,\psi\in\Hp$ and scalars $a,b\in\Q$, the following estimate holds:
\begin{equation}
|\langle\phi, (aS+b\,T)\psi\rangle|\leq \max\{|a|\,|\langle\phi,S\psi\rangle|, |b|\,|\langle\phi,T\psi\rangle|\}.
\end{equation}
It follows that
\begin{equation}
\lim_{m+n}\langle \phi_m, S\phi_n\rangle=0=\lim_{m+n}\langle\phi_m, T\phi_n\rangle\implies\lim_{m+n}\langle\phi_m, (aS+bT)\phi_n\rangle=0,
\end{equation}
namely, $S,T\in\mathcal{T}_\Phi(\Hp)\implies aS+b\, T\in\mathcal{T}_\Phi(\Hp)$. Moreover, if $T\in\mathcal{T}_\Phi(\Hp)\subset\Bpa$,
then $\langle\phi_m,T\phi_n\rangle=\overline{\langle\phi_n, T^*\phi_m\rangle}$. Hence,
$\lim_{m+n}\langle\phi_m,T\phi_n\rangle=0\implies\lim_{m+n}\langle\phi_m,T^*\phi_n\rangle=0$, i.e., the implication~\eqref{eq.53} holds true.
\end{proof}


\subsection{The trace class}
\label{main}

Our next task is to show that, actually, the definition of the linear subspace  $\mathcal{T}_\Phi(\Hp)$ of $\Bpa$
does not depend on the choice of $\Phi$; i.e., given any pair of orthonormal bases $\Phi\equiv\{\phi_m\mnat$ and
$\Psi\equiv\{\psi_n\nnat$ in $\Hp$, we have that
\begin{equation}\label{eq.54}
\mathcal{T}_\Phi(\Hp)=\mathcal{T}_\Psi(\Hp).
\end{equation}

To prove this important fact, we need to establish a further relevant property of $\mathcal{T}_\Phi(\Hp)$.
To this aim, we will use the following technical result:

\begin{lemma}\label{lem.t1}
Let $(A_{lm}),(T_{mn})\in \Minf$ be any pair of infinite matrices satisfying the following conditions:
\begin{enumerate}[label=\rm{(\alph*)}]

\item \label{cond.l6a} $\alpha\equiv\sup_{l,m}|A_{lm}|<\infty$,

\item \label{cond.l6b} $\lim_l A_{lm}=0$, $\forall m\in\mathbb{N}$, and

\item \label{cond.l6c} $\lim_{m+n}T_{mn}=0$.

\end{enumerate}
Then, for every $(l,n)\in\mathbb{N}\times\mathbb{N}$, the series $\sum_m A_{lm}T_{mn}$ converges to some
$S_{ln}\in\Q$ and $\lim_{l+n}S_{ln}=0$.
\end{lemma}

\begin{proof}
By condition~\ref{cond.l6a}, the sequence $\{A_{lm}\}_{m\in\nat}$ belongs to $\ell^{\infty}$, for every $l\in\nat$,
and, by condition~\ref{cond.l6c}, the sequence $\{T_{mn}\}_{m\in\nat}$ belongs to $c_0$, for every $n\in\nat$
(see relation~\eqref{limiff} in Remark~\ref{rem.6}). Hence, the series $\sum_m A_{lm}T_{mn}$ is convergent,
for all $l,n\in\nat$, and we can put
\begin{equation}\label{eq.55}
S_{ln}=\sum_m A_{lm}T_{mn}\in\Q.
\end{equation}
Moreover, condition~\ref{cond.l6c} also entails that
\begin{enumerate}[label=\rm{(\alph*)}]

\setcounter{enumi}{3}

\item \label{cond.l6d} $\tau\equiv\sup_{m,n}|T_{mn}|<\infty$.

\end{enumerate}

Let us assume that $\alpha,\tau>0$ (otherwise there is nothing to prove), and let us take any $\epsilon>0$.
Now, by~\ref{cond.l6c}, there is some $\ttN\in\nat$ such that, if $\max\{m,n\}>\ttN$, then $|T_{mn}|<\epsilon/\alpha$.
Next, by~\ref{cond.l6b}, there exists some $\ttL\in\nat$ such that, if $m\leq \ttN$ and $l>\ttL$, then $|A_{lm}|<\epsilon/\tau$
(because we are considering a \emph{finite} set $\{A_{l1}\}_{l\in\nat},\dots,\{A_{l\ttN}\}_{l\in\nat}$ of sequences
converging to $0$).

Summarizing, we have found that, for every $\epsilon>0$, there are $\ttL,\ttN\in\nat$ such that the following additional
conditions hold:
\begin{enumerate}[label=\rm{(\alph*)}]

\setcounter{enumi}{4}

\item \label{cond.l6e} $l>\ttL,\ m\leq\ttN\implies|A_{lm}|<\epsilon/\tau$,

\item \label{cond.l6f} $\max\{m,n\}>\ttN\implies|T_{mn}|<\epsilon/\alpha$.

\end{enumerate}

Therefore, eventually we obtain the following estimates:
\begin{enumerate}[label=\sf{(E\arabic*)}]

\item \label{este1} By~\ref{cond.l6a} and~\ref{cond.l6f} --- for all $l\in\nat$ and all $n>\ttN$ ---
we have that
\begin{equation*}
|S_{ln}|=|{\textstyle\sum_m} A_{lm}T_{mn}|\leq \sup_m |A_{lm}|\,|T_{mn}|<\alpha\,\frac{\epsilon}{\alpha}=\epsilon.
\end{equation*}

\item \label{este2} By~\ref{cond.l6d} and~\ref{cond.l6e}, and by~\ref{cond.l6a} and~\ref{cond.l6f} --- for all $l>\ttL$
and all $n\in\nat$ --- we have:
\begin{equation*}
\hspace{-2mm}|S_{ln}|\leq \sup_m |A_{lm}|\,|T_{mn}|=\max\hspace{-0.4mm}\bigg\{\max_{m\leq\ttN}\{|A_{lm}|\,|T_{mn}|\},\sup_{m>\ttN}|A_{lm}|\,|T_{mn}|\bigg\}
\hspace{-0.5mm}<\max\hspace{-0.4mm}\bigg\{\frac{\epsilon}{\tau}\,\tau,\,\alpha\,\frac{\epsilon}{\alpha}\bigg\}\hspace{-0.5mm}=\epsilon.
\end{equation*}

\end{enumerate}

In conclusion, by the estimates~\ref{este1} and~\ref{este2}, for every $\epsilon>0$, there are $\ttL,\ttN\in\nat$ such that
\begin{equation}\label{eq.56}
\mbox{$l>\ttL$ and/or $n>\ttN$}\ \implies\ |S_{ln}|<\epsilon
\end{equation}
(and, \emph{a fortiori}, if $\max\{l,n\}>\ttM\equiv\max\{\ttL,\ttN\}$, then $|S_{ln}|<\epsilon$).

Eventually, we have shown that, for every $\epsilon>0$, the set
\begin{equation}
\{(l,n)\in\mathbb{N}\times\mathbb{N}\sep |S_{ln}|>\epsilon\}
\end{equation}
is finite. Equivalently, for every $\epsilon>0$, there exists some $\ttM\in\nat$ such that, if $\max\{l,n\}>\ttM$, then
$|S_{ln}|<\epsilon$; namely, $\lim_{l+n}S_{ln}=0$.
\end{proof}

\begin{theorem}\label{th.t3}
Given any orthonormal basis $\Phi$, the linear subspace $\mathcal{T}_\Phi(\Hp)$ of $\Bpa\subset\Bp$
is a \emph{left ideal} in $\Bp$, i.e.,
\begin{equation}\label{eq.57}
AT\in\mathcal{T}_\Phi(\Hp),\;\;\;\;\forall A\in\Bp,\forall\, T\in\mathcal{T}_\Phi(\Hp).
\end{equation}
Moreover, $\mathcal{T}_\Phi(\Hp)$ is a \emph{two sided $*$-ideal} in $\Bpa$, i.e.,
\begin{equation}\label{eq.58}
T^*\in\mathcal{T}_\Phi(\Hp),\;\;\;\;AT,\,TA\in\mathcal{T}_\Phi(\Hp),\;\;\forall\, T\in\mathcal{T}_\Phi(\Hp),\;\;\forall A\in\Bpa.
\end{equation}
\end{theorem}

\begin{proof}
Let us prove property~\eqref{eq.57} of $\mathcal{T}_\Phi(\Hp)$. Since $A\in\Bp$ and $T\in\mathcal{T}_\Phi(\Hp)\subset\Bp$, we have
that $S=AT\in\Bp$ and
\begin{equation}\label{eq.59}
A=\opPhi(A_{lm}),\,T=\opPhi(T_{mn})\;\;\implies\;\;S=\opPhi(S_{ln}),
\end{equation}
where $S_{ln}=\sum_m A_{lm}T_{mn}$.
Here, the infinite matrix $(A_{lm})$ satisfies conditions~\ref{cond.l6a} and~\ref{cond.l6b} in Lemma~\ref{lem.t1},
because $A$ is bounded (Theorem~\ref{th.1}). Moreover, $(T_{mn})$ satisfies condition~\ref{cond.l6c} (by the definition
of $\mathcal{T}_\Phi(\Hp)$). Hence, by the same lemma, $\lim_{l+n}S_{ln}=0$; i.e., $S=AT\in\mathcal{T}_\Phi(\Hp)$.

Let us now prove that $\mathcal{T}_\Phi(\Hp)$ satisfies properties~\eqref{eq.58}, as well. We have already shown that, if $T\in\mathcal{T}_\Phi(\Hp)$,
then $T^*\in\mathcal{T}_\Phi(\Hp)$ too (Corollary~\ref{cor.4}).

Next, if, additionally, $A\in\Bpa$, then $AT,A^*T^*\in\mathcal{T}_\Phi(\Hp)$, so that
\begin{equation}\label{eq.60}
TA=(A^*T^*)^*\in\mathcal{T}_\Phi(\Hp).
\end{equation}
The proof is complete.
\end{proof}

We will now derive two remarkable consequences of Theorem~\ref{th.t3}; the most important one is
the following:

\begin{corollary}\label{cor.t5}
For every pair of orthonormal bases $\Phi\equiv\{\phi_m\mnat$ and $\Psi\equiv\{\psi_n\nnat$, we have that
\begin{equation}\label{eq.70}
\mathcal{T}_\Phi(\Hp)=\mathcal{T}_\Psi(\Hp)\equiv\Tp.
\end{equation}
\end{corollary}

\begin{proof}
Let $U$ be the unitary operator determined by condition~\ref{cond.th5u2} in Theorem~\ref{th.u3}, i.e.,
\begin{equation}
U=\opPhi(\langle\phi_m,U\phi_n\rangle)=\opPhi(\langle\phi_m,\psi_n\rangle).
\end{equation}
Recalling that \ref{cond.p6t2}~$\hspace{-1.5mm}\iff\hspace{-1.5mm}$~\ref{cond.p6t3} (Proposition~\ref{prop.t3}),
we have:
\begin{align}
T\in\mathcal{T}_\Phi(\Hp)&\iff U^*TU\in\mathcal{T}_\Phi(\Hp)
\quad (\mbox{\eqref{eq.58} in Theorem~\ref{th.t3}, $U^*=U^{-1}\in\Bpa$})
\nonumber\\
&\iff T\in\Bp,\ 0=\lim_{m+n}\langle\phi_m, U^*TU\phi_n\rangle=\lim_{m+n}\langle\psi_m,T\psi_n\rangle
\nonumber\\
&\iff T\in\mathcal{T}_\Psi(\Hp).
\end{align}
Therefore, $\mathcal{T}_\Psi(\Hp)=\mathcal{T}_\Phi(\Hp)$, for any pair of orthonormal bases $\Phi,\Psi$ in $\Hp$.
\end{proof}

\begin{definition}
We call an operator belonging to the two-sided $*$-ideal $\Tp$ of $\Bpa$ --- whose definition does not depend on
the choice of an orthonormal basis in $\Hp$ (by Corollary~\ref{cor.t5}) --- a \emph{trace class} operator.
The linear space $\Tp$ itself will be called the \emph{the trace class of} $\Hp$.
\end{definition}

We next obtain a second remarkable consequence of Theorem~\ref{th.t3}:

\begin{corollary}\label{cor.t6}
Given a linear operator $T$ in $\Hp$, the following facts are equivalent:
\begin{enumerate}[label=\rm{(\roman*)}]

\item \label{cond.c6i} $T\in\Tp$;

\item \label{cond.c6ii} $T\in\Bp$ and, for \emph{some} pair $\Phi\equiv\{\phi_m\mnat, \Psi\equiv\{\psi_n\nnat$ of
orthonormal bases, satisfies the condition that
\begin{equation}\label{eq.72}
\lim_{m+n}\langle\phi_m,T\psi_n\rangle=0;
\end{equation}

\item \label{cond.c6iii} $T\in\Bp$ and, for \emph{every} pair $\Phi\equiv\{\phi_m\mnat,\Psi\equiv\{\psi_n\nnat$ of
orthonormal bases in $\Hp$, satisfies condition~\eqref{eq.72}.

\end{enumerate}
\end{corollary}

\begin{proof}
Clearly, \ref{cond.c6iii} $\implies$~\ref{cond.c6ii}. Let us prove that \ref{cond.c6ii}  $\implies$~\ref{cond.c6i}.

Assume that~\ref{cond.c6ii} holds, and let $U$ be the unitary operator determined by
\begin{equation}\label{eq.73}
\mbox{$U\phi_k=\psi_k$, $\forall k\in\mathbb{N}$; i.e., $U=\opPhi(\langle\phi_m,\psi_n\rangle)$}.
\end{equation}
By~\eqref{eq.72} we have:
\begin{equation}\label{eq.74}
0=\lim_{m+n}\langle\phi_m, T\psi_n\rangle=\lim_{m+n}\langle\phi_m, TU\phi_n\rangle.
\end{equation}
Therefore, $TU\in\Tp$ and, by Theorem~\ref{th.t3}, $T=(TU)U^*\in\Tp$ too; i.e., \ref{cond.c6ii} $\implies$~\ref{cond.c6i}.

It is then sufficient to show that~\ref{cond.c6i} $\implies$~\ref{cond.c6iii}, as well. Let $T\in\Tp$ and let
$\Phi\equiv\{\phi_m\mnat$, $\Psi\equiv\{\psi_n\nnat$ be \emph{any} pair of orthonormal bases in $\Hp$. We have:
\begin{equation}
\lim_{m+n}\langle\phi_m, T\psi_n\rangle=\lim_{m+n}\langle\phi_m, TU\phi_n\rangle=0.
\end{equation}
Here, $U$ is the unitary operator determined by~\eqref{eq.73}, and we have used the fact that $TU\in\Tp$
(Theorem~\ref{th.t3}).
\end{proof}

\begin{remark}\label{rem.8}
Recalling Remark~\ref{rem.6}, the condition that $T\in\Bp$ satisfies~\eqref{eq.72} is equivalent to
the condition that the series
\begin{equation}\label{eq.75}
\sum_{m,n}\langle\phi_m, T\psi_n\rangle
\end{equation}
be convergent. This is reminiscent of the fact that, in a (infinite-dimensional, separable)
\emph{complex Hilbert space} $\mathcal{K}$,
\begin{equation}\label{eq.76}
T\in\mathcal{T}(\mathcal{K})\;\iff\;\sum_{m}\sum_n|\langle \eta_m, T\chi_n\rangle|<\infty,
\end{equation}
for some --- equivalently, for any --- pair $\{\eta_m\mnat,\{\chi_n\nnat$ of orthonormal bases in $\mathcal{K}$.
This is probably the tightest connection that one can establish between the $p$-adic and the complex
trace class. Recall indeed that, for a complex Hilbert space $\mathcal{K}$, one usually first defines the trace of
a \emph{positive} bounded operator (that may be finite or infinite). Then, the trace class $\mathcal{T}(\mathcal{K})$
is introduced as the set of all bounded operators $T$ such that their absolute value $|T|$ (the unique positive square root
of $T^*T$) has a \emph{finite} trace; see, e.g.,~\cite{Reed}. But this route cannot be pursued in the $p$-adic
setting, because there is no natural notion of positivity for a bounded operator.
\end{remark}

Having shown that the definition of the trace class $\Tp$ does not depend on the choice of an
orthonormal basis in $\Hp$, we now want to prove that, for every $T=\opPhi(T_{mn})\in\Tp$, the
\emph{trace} itself of $T$ --- i.e., the the quantity (recall Proposition~\ref{prop.t2})
\begin{equation}\label{eq.77}
\mathrm{tr}_\Phi(T)\defi\sum_m\langle\phi_m,T\phi_m\rangle=\sum_mT_{mm}\in\Q
\end{equation}
--- \emph{does not} depend on the orthonormal basis $\Phi\equiv\{\phi_m\mnat$. Thus, we can call
$\tr(T)\equiv\tr_\Phi(T)$ the \emph{trace of the operator} $T\in\Tp$.

We first need to establish a technical fact.

\begin{lemma}\label{lem.t2}
Given double sequences $\{x_{mn}\}_{m,n\in\nat}$, $\{y_{mn}\}_{m,n\in\nat}$ in $\Q$, the following facts
hold true:
\begin{enumerate}[label=\sf{(\roman*)}]

\item \label{cond.le6i} $\lim_{m+n}x_{mn}=0$ and $|y_{mn}|\leq\alpha\in\mathbb{R}^+$,
$\forall m,n\in\nat\;\implies\;\lim_{m+n}x_{mn}y_{mn}=0$.

\item \label{cond.le6ii} If $\{x_{mn}\}_{m,n\in\nat}$ is of the form $x_{mn}=y_mz_n$, where $\lim_m y_m=0=\lim_n z_n$, then
\begin{equation}
\lim_{m+n}x_{mn}=0.
\end{equation}

\end{enumerate}
\end{lemma}

\begin{proof}
Claim~\ref{cond.le6i} is obvious. Let us prove~\ref{cond.le6ii}.

Both the sequences $\{y_m\mnat,\{z_n\nnat$ converge to zero; hence:
\begin{equation}\label{eq.788}
\mbox{$\displaystyle\lim_{m}x_{mn}=\lim_{m}y_m z_n= 0$, for all $n\in\nat$, and $\displaystyle\lim_{n}y_m z_n= 0$,
for all $m\in\nat$.}
\end{equation}
Moreover,
\begin{equation}\label{eq.79}
\lim_{m,n}x_{mn}=0, \qquad \mbox{(Pringsheim limit)}
\end{equation}
because, $\forall\epsilon>0$, $\exists\ttN\in\mathbb{N}$ such that, if $m,n>\ttN$, then
\begin{equation}\label{eq.80}
|y_m|,|z_n|<\sqrt{\epsilon},
\end{equation}
so that  $|x_{mn}|=|y_m|\,|z_n|<\epsilon$. As recalled in Remark~\ref{rem.6}, conditions~\eqref{eq.788}
and~\eqref{eq.79} together entail that
\begin{equation}
\lim_{m+n}x_{mn}=0,
\end{equation}
which proves~\ref{cond.le6ii}.
\end{proof}

\begin{theorem}\label{th.t5}
If $T\in\Tp$, then, for any pair of orthonormal bases $\Phi\equiv\{\phi_m\mnat$ and $\Psi\equiv\{\psi_n\nnat$ in $\Hp$,
we have that
\begin{equation}\label{eq.84}
\tr_\Phi(T)=\tr_\Psi(T)\equiv\tr(T).
\end{equation}
\end{theorem}

\begin{proof}
Indeed, first note that
\begin{equation}\label{eq.85}
\tr_\Phi(T)\defi\sum_m \langle\phi_m, T\phi_m\rangle =
\sum_{m}\sum_n \langle\phi_m,\psi_n\rangle\langle\psi_n,T\phi_m\rangle,
\end{equation}
where we have used the expansion $\phi_m=\sum_n\langle\psi_n,\phi_m\rangle\psi_n$ and the continuity of the inner product.
Since $|\langle \phi_m,\psi_n\rangle|\leq 1$, for all $m,n\in\mathbb{N}$, then, by the implication~\ref{cond.c6i}
$\implies$~\ref{cond.c6ii} in Corollary~\ref{cor.t6}, and by claim~\ref{cond.le6i} of Lemma~\ref{lem.t2}, we can argue that
\begin{align*}
T\in\Tp\;\;&\implies\;\;\lim_{m+n}\langle\psi_n, T\phi_m\rangle=0\qquad (\text{Corollary~\ref{cor.t6}})\\
&\implies\;\;\lim_{m+n}\langle \phi_m, \psi_n\rangle\langle \psi_n, T\phi_m\rangle=0\qquad (\text{Lemma~\ref{lem.t2}}).
\end{align*}
Thus, as recalled in Remark~\ref{rem.6} (see~\eqref{eq.50}), we can exchange the sums on the r.h.s.\ of~\eqref{eq.85},
so obtaining
\begin{align}\label{eq.86}
\tr_\Phi(T)&=\sum_n\sum_m\langle\phi_m,\psi_n\rangle\langle\psi_n,T\phi_m\rangle
\nonumber\\
&=\sum_n\sum_m\langle\phi_m,\psi_n\rangle
{\textstyle \langle\psi_n,T (\sum_k\langle\psi_k,\phi_m\rangle\psi_k)\rangle}
\nonumber\\
&=\sum_n\sum_m\sum_k\langle\psi_k,\phi_m\rangle\langle\phi_m,\psi_n\rangle\langle\psi_n,T\psi_k\rangle,
\end{align}
where, for the second equality, we have used the fact that
$T\sum_k\langle\psi_k,\phi_m\rangle\psi_k=\sum_k\langle\psi_k,\phi_m\rangle(T\psi_k)$
($T$ being bounded) and, once again, the continuity of the inner product.

Next, since $|\langle\psi_k,\phi_m\rangle|\leq 1$, for all $k,m\in\nat$, $\lim_{m}\langle\phi_m,\psi_n\rangle=0$,
for all $n\in\nat$, and
\begin{equation}\label{eq.88}
\lim_k\langle \psi_n, T\psi_k\rangle=\lim_k \langle T^*\psi_n,\psi_k\rangle=0,\quad \forall n\in\nat,
\end{equation}
$T$ being (of trace class, hence) adjointable, we have:
\begin{align*}
&\lim_{m+k}\langle\phi_m,\psi_n\rangle\langle\psi_n,T\psi_k\rangle=0,\quad \forall n\in\nat,
\qquad\mbox{(by point~\ref{cond.le6ii} of Lemma~\ref{lem.t2})}
\\
\implies &\lim_{m+k}\langle\psi_k,\phi_m\rangle\langle\phi_m,\psi_n\rangle\langle\psi_n,T\psi_k\rangle=0,
\quad\forall n\in\nat.\qquad \mbox{(by point~\ref{cond.le6i} of Lemma~\ref{lem.t2})}
\end{align*}
Therefore, it is further possible to exchange the sums over $m$ and $k$ in the last line of~\eqref{eq.86}.

Eventually, we obtain that
\begin{align}
\tr_\Phi(T)&=\sum_n\sum_k\bigg(\sum_m\langle\psi_k,\phi_m\rangle\langle\phi_m,\psi_n\rangle\bigg)\langle\psi_n, T\psi_k\rangle
\nonumber\\
&=\sum_n\sum_k\langle\psi_k,\psi_n\rangle\langle\psi_n,T\psi_k\rangle
\nonumber\\
&=\sum_n\sum_k\delta_{nk}\langle\psi_n, T\psi_k\rangle
\nonumber\\
&=\sum_n\langle\psi_n, T\psi_n\rangle\ifed\tr_\Psi(T).
\end{align}
Here, for obtaining the second equality we have exploited the (very familiar, in a complex Hilbert space)
relation $\sum_m\langle\psi_k,\phi_m\rangle\langle\phi_m,\psi_n\rangle=\langle\psi_k,\psi_n\rangle$, which,
in the $p$-adic setting, is provided by the last line of~\eqref{reside}.
\end{proof}

The trace enjoys the following remarkable properties:

\begin{proposition}\label{prop.tnew}
Given trace class operators $S,T\in\Tp$, we have:
\begin{enumerate}[label=\tt{(P\arabic*)}]

\item $\tr(S+T)=\tr(S)+\tr(T)$, and $\tr(\alpha\, T)=\alpha\,\tr(T)$, for all $\alpha\in\Q$, i.e.,
the trace $\tr\colon\Tp\rightarrow\Q$ is a linear functional;

\item \label{conjtrace} $\tr(T^*)=\overline{\tr(T)}$;

\item for every unitary operator $U\in\mathcal{U}(\Hp)$, $\tr(UTU^*)=\tr(T)$.

\end{enumerate}
\end{proposition}

\begin{proof}
The linearity of the trace is clear. Let us prove property~\ref{conjtrace}.

Let $\Phi\equiv\{\phi_n\nnat$ be any orthonormal basis in $\Hp$. Then, we have:
\begin{equation}
\tr(T^\ast)= \sum_n\langle \phi_n, T^\ast\phi_n\rangle=\sum_n\langle T\phi_n,\phi_n\rangle
=\sum_n\overline{\langle\phi_n, T\phi_n\rangle}=\overline{\sum_n\langle\phi_n, T\phi_n\rangle}
=\overline{\tr(T)}.
\end{equation}
Moreover, for every unitary operator $U\in\mathcal{U}(\Hp)$,
\begin{equation}
\tr(UTU^*)= \sum_n\langle \phi_n, UTU^* \phi_n\rangle=\sum_n\langle \psi_n,T \psi_n\rangle
=\tr(T),
\end{equation}
where we have used the fact that $U^*$ is a unitary operator too, so that, by point~\ref{cond.th5u8}
of Theorem~\ref{th.u3}, $\Psi\equiv\{\psi_n\nnat=\{U^*\phi_n\nnat$ is an orthonormal basis.
\end{proof}


\subsection{The cyclic property}
\label{cyclic}

The reader will have noticed that in Proposition~\ref{prop.tnew} are listed all the main properties
of the trace --- say, in a complex Hilbert space --- \emph{except} the `cyclic property'. We are now
going to show that the $p$-adic trace possesses this important property too, \emph{provided that}
the domain of the map $\tr\argo$ be suitably extended (for the sake of simplicity, we will denote
the extended map by the same symbol).

In fact, recalling the first assertion of Theorem~\ref{th.t3}, $\Tp$ is a \emph{left} --- but not a right ---
ideal in $\Bp$ ($\dim(\Hp)=\infty$). Let us better clarify this point by means of an explicit example.

\begin{example} \label{exnontc}
Let $B\in\Bp$ a bounded operator that is \emph{not} adjointable, and let $\chi\in\Hp$ a (nonzero) vector
such that $\chi\not\in\dom(B^\dagger)$, where $B^\dagger$ is the pseudo-adjoint of $B$. For every $\phi\in\Hp$,
such that $\langle\phi,\phi\rangle=1$ (e.g., an element of an orthonormal basis), we can consider the trace class
operator $T=|\phi\rangle\langle\chi|\in\Tp$. Let us show that the bounded operator $TB\in\Bp$ is \emph{not} a
trace class operator. Indeed, for every $\psi\in\Hp$, we have: $\langle\phi,TB\psi\rangle=\langle\chi,B\psi\rangle$.
Now, since $\chi\not\in\dom(B^\dagger)$, \emph{there is no vector} $\eta\in\Hp$ such that
$\langle\eta,\psi\rangle=\langle\chi,B\psi\rangle=\langle\phi,TB\psi\rangle$, for all $\psi\in\Hp$; otherwise stated,
$\phi\not\in\dom((TB)^\dagger)$. Therefore, (whereas $BT=|B\phi\rangle\langle\chi|\in\Tp$) $TB\not\in\Bpa\supset\Tp$,
and this fact entails that $\Tp$ is \emph{not} a right ideal in $\Bp$.
\end{example}

By the previous discussion, in order to derive, in the $p$-adic setting, the cyclic property
of the trace, we need to introduce a new class of operators that we will call the \emph{weak trace class}.

\begin{definition}\label{def.7}
We say that a bounded operator $A\in\Bp$ is \emph{uniformly traceable} if it is traceable with respect to
every orthonormal basis in $\Hp$ and, moreover,
\begin{equation}\label{eq.90}
\tr_\Phi(A)=\tr_\Psi(A)\equiv\tr(A)
\end{equation}
for every pair of orthonormal bases $\Phi\equiv\{\phi_m\mnat$ and $\Psi\equiv\{\psi_n\nnat$ in $\Hp$.
\end{definition}

It is clear that the set of all uniformly traceable operators in $\Hp$ form a linear subspace $\Tw$
of $\Bp$, which is precisely the weak trace class of $\Hp$.

\begin{remark}\label{re.9}
Let $\mathcal{K}$ be  a (separable) complex Hilbert space, with $\dim(\mathcal{K})=\infty$.
It is well known that a bonded operator $A\in\mathcal{B}(\mathcal{K})$ is of trace class iff
it is uniformly traceable; i.e., iff the series
\begin{equation}\label{eq.91}
\sum_m\langle\chi_m, A\chi_m\rangle
\end{equation}
converges to a \emph{unique} limit for every orthonormal basis $\{\chi_m\mnat$ in $\mathcal{K}$
(see, e.g., Proposition~{4.42} of~\cite{moretti2018spectral}). It turns out that --- see Remark~\ref{reincls}
below --- this property does not hold true for a $p$-adic Hilbert space; namely,
$\Tw\supsetneq\Tp$ ($\dim(\Hp)=\infty$).
\end{remark}

\begin{proposition}[Cyclic property of the trace]\label{prop.5}
For every bounded operator $B\in\Bp$ and for every trace class operator $T\in\Tp$, we have that
\begin{equation}\label{eq.92}
BT\in\Tp\subset\Tw\;\;\;\;\text{and}\;\;\;\; TB\in\Tw.
\end{equation}
Moreover, we have:
\begin{equation}\label{eq.93}
\tr(BT)=\tr(TB).
\end{equation}
\end{proposition}

\begin{proof}
We have already shown that $BT\in\Tp$, because $\Tp$ is a left ideal in $\Bp$ (see Theorem~\ref{th.t3},
where $\mathcal{T}_\Phi(\Hp)\equiv\Tp$). Then, we have:
\begin{align}\label{eq.94}
\tr(BT)&=\sum_m\langle\phi_m, BT\phi_m\rangle
\nonumber\\
&=\sum_m\langle\phi_m, B ({\textstyle\sum_n}\langle\psi_n, T\phi_m\rangle\,\psi_n)\rangle
\qquad \mbox{(because $T\phi_m=\sum_n\langle\psi_n, T\phi_m\rangle\psi_n$)}
\nonumber\\
&=\sum_m\sum_n\langle\phi_m, B\psi_n\rangle\langle\psi_n,T\phi_m\rangle,
\qquad \mbox{(continuity of $B$ and of the inner product)}
\end{align}
where $\Phi\equiv\{\phi_m\mnat$, $\Psi\equiv\{\psi_n\nnat$ is any pair of orthonormal bases in $\Hp$.

Note that
\begin{equation}
\mbox{$|\langle\phi_m, B\psi_n\rangle|\leq \|B\|$, $\forall m,n\in\mathbb{N}$,
and $\displaystyle\lim_{m+n}\langle\psi_n, T\phi_m\rangle=0$}.
\end{equation}
Hence, by point~\ref{cond.le6i} of Lemma~\ref{lem.t2}, we have:
\begin{equation}\label{eq.95}
\lim_{m+n}\langle\phi_m, B\psi_n\rangle\langle\psi_n, T\phi_m\rangle=0.
\end{equation}
By~\eqref{eq.95}, we can exchange the sums in the last line of~\eqref{eq.94}, so obtaining
\begin{align}\label{eq.96}
\tr(BT)&=\sum_n\sum_m\langle\psi_n, T\phi_m\rangle\langle\phi_m, B\psi_n\rangle
\nonumber\\
&=\sum_n\langle\psi_n, T({\textstyle\sum_m}\langle\phi_m, B\psi_n\rangle\phi_m)\rangle
\nonumber\\
&=\sum_n\langle\psi_n, TB\psi_n\rangle=\tr_\Psi(TB).
\end{align}
By the arbitrariness of the orthonormal basis $\Psi$ in $\Hp$, we conclude that $TB\in\Tw\subset\Bp$
and $\tr(TB)=\tr_\Psi(TB)=\tr(BT)$.
\end{proof}

\begin{remark}
The inclusion relation
\begin{equation}
TB\subset\Tw, \quad \forall\, T\in\Tp,\ \forall B\in\Bp,
\end{equation}
is a manifestation of the fact that $TB$ is a \emph{compact operator}; see Corollary~\ref{corut} below.
In particular, the (non-adjointable, bounded) operator $TB$ constructed in Example~\ref{exnontc} is compact.
\end{remark}

\begin{proposition} \label{trineq}
For every bounded operator $B\in\Bp$ and for every trace class operator $T\in\Tp$, we have that
\begin{equation}\label{trinequa}
|\tr(BT)|=|\tr(TB)|\le \|B\|\,\|T\| \quad \mbox{and} \quad |\tr(T)|\le\|T\|.
\end{equation}
\end{proposition}

\begin{proof}
In fact,  given any orthonormal basis $\{\phi_m\mnat$ in $\Hp$,
$|\tr(BT)|=|\sum_m\langle\phi_m, BT\phi_m\rangle|\le\max_{m\in\nat}|\langle\phi_m, BT\phi_m\rangle|\le
\max_{m\in\nat}\|BT\phi_m\|\le \|B\|\,\|T\|$. In particular, putting $B=\id$, we obtain also the second
inequality in~\eqref{trinequa}.
\end{proof}

\begin{corollary} \label{corbotra}
The linear functional $\tr\argo\colon\Tp\ni T\mapsto\tr(T)\in\Q$ is bounded and $\|\tr\argo\|=1$.
\end{corollary}

\begin{proof}
By the second inequality in~\eqref{trinequa} the functional $\tr\argo$ is bounded and $\|\tr\argo\|\le1$.
If $\phi$ is an element of an orthonormal basis in $\Hp$, then $|\tr(|\phi\rangle\langle\phi|)|=1=\|\,|\phi\rangle\langle\phi|\,\|$,
so that the previous inequality is saturated.
\end{proof}


\subsection{Trace class operators as compact operators}
\label{compact}

As is well known, the trace class operators in a (infinite-dimensional, separable) \emph{complex} Hilbert space
$\mathcal{K}$ form a Banach space, when endowed with the trace norm. This space is embedded in the Hilbert
space of all Hilbert-Schmidt operators in $\mathcal{K}$ (endowed with the Hilbert-Schmidt product). The
closure --- w.r.t.\ the operator norm --- of these spaces coincides with the closure of the linear space of all
finite rank operators in $\mathcal{K}$; namely, with the Banach space of compact operators, which is the
only proper closed two-sided ideal in the Banach algebra of bounded operators. In particular,
the trace class of $\mathcal{K}$ is \emph{not} closed w.r.t.\ the operator norm ($\dim(\mathcal{K})=\infty$).
See, e.g., the standard references~\cite{Reed,Simon,weidmann2012linear,moretti2018spectral}.

As the reader may expect, this familiar picture keeps some of its main features --- but also requires some
essential modification --- when switching to a (infinite-dimensional) $p$-adic Hilbert space $\Hp$.

As above, for the sake of simplicity, we will assume that $\dim(\Hp)=\infty$, \emph{but} all subsequent results
(and their proofs) remain valid --- with obvious adaptations, and even if possibly getting trivial --- in the
finite-dimensional setting. E.g., the `canonical decomposition' of an adjointable compact operator ---
see Corollary~\ref{decocpa} below --- holds true in the case where $\dim(\Hp)<\infty$ and $\Tp=\Cp=\Bpa=\Bp$
is just the space $\linop$ of all linear operators in $\Hp$.

\begin{definition} \label{deficom}
An all-over linear operator $C$ in $\Hp$ is said to be \emph{compact} if $C\,\Hba$ --- where $\Hba$ is
the unit ball in $\Hp$: $\Hba\defi\{\psi\in\Hp\sep \|\psi\|\le1\}$ --- is a precompact subset of $\Hp$
(namely, if $C\,\Hba$ has a compact closure).
\end{definition}

\begin{remark}
In formulating the previous definition, we have taken into account the fact that $\Q$ is locally compact,
because, in this case, the \emph{compactoid} subsets of $\Hp$ coincide with the precompact subsets.
See Chapt.~{4} of~\cite{rooij1978non}; in particular, Sect.~{4.S} and the subsequent definition of
a compact operator in the non-Archimedean setting (also see the seminal paper~\cite{Serre1962}).
\end{remark}

\begin{remark}
From Definition~\ref{deficom} it is clear that the linear space $\Fp$ of all \emph{finite-rank operators}
--- the linear operators in $\Hp$ having finite-dimensional range spaces --- consists of compact
operators.
\end{remark}

\begin{theorem} \label{theocpops}
The set $\Cp$ of all compact operators in $\Hp$ is a closed linear subspace of $\Bp$. Specifically, $\Cp$
is the closure of the linear subspace $\Fp$ of all finite rank operators in $\Hp$. Moreover, $\Cp$ is the
only proper closed two-sided ideal in $\Bp$.

Given any orthonormal basis $\Phi\equiv\{\phi_m\}_{m\in\nat}$ in $\Hp$, a bounded operator $A=\opPhi(A_{mn})\in\Bp$
--- $A_{mn}=\langle\phi_m,A\phi_n\rangle$ --- is compact iff
\begin{equation} \label{conmatcom}
\lim_m\big({\textstyle\sup_{n\in\nat}}|A_{mn}|\big)=0.
\end{equation}

Every compact operator  $C\in\Cp$ can be expressed as
\begin{equation}\label{decp}
C=\sum_{j\in J}\lambda_j\, \osj\pro\dsj,
\end{equation}
where $J=\{1,2,\ldots\}$ is a countable index set and
\begin{itemize}

\item $\{\lambda_j\}_{j\in J}\subset\Q$ --- for $C\neq 0$, we assume that
$\{\lambda_j\}_{j\in J}\subset\Qa\equiv\Q\setminus\{0\}$ --- and, if $J=\nat$, $\lim_j \lambda_j=0$;

\item $\{\osj\}_{j\in J}$ and $\{\dsj\}_{j\in J}$ are contained in $\Hp$ and $\dH$, respectively,
with $\|\osj\|=\|\dsj\|=1$;

\item $\{\osj\}_{j\in J}$ is a (normalized) norm-orthogonal system in $\Hp$;

\item $\osj\pro\dsj\colon\Hp\rightarrow\Hp$ is the bounded operator defined by $(\osj\pro\dsj)\,\psi\defi\dsj(\psi)\,\osj$,
and the sum in~\eqref{decp} --- whenever $J$ is not finite --- converges w.r.t.\ the norm topology.

\end{itemize}

In particular, the norm-orthogonal system $\{\osj\}_{j\in J}$ can be chosen to be contained
in any orthonormal basis in $\Hp$.

Conversely, every operator $C$ of the previous form --- i.e., such that $C\psi=\sum_{j\in J}\lambda_j\, \dsj(\psi)\,\osj$,
for all $\psi\in\Hp$, with $\{\lambda_j\}_{j\in J}$, $\{\osj\}_{j\in J}$ and $\{\dsj\}_{j\in J}$ as specified above --- is compact.
\end{theorem}

\begin{proof}
For the first two assertions, see Chapt.~{4} of~\cite{rooij1978non}; in particular, Theorem~{4.39} and
the subsequent discussion. For the third assertion, since the valuation group $|\Q^\ast|$ is discrete,
we can apply Theorem~{5} and Corollary~{6} of~\cite{vanderPut}.

Let us prove the fourth assertion. Given any orthonormal basis $\Phi\equiv\{\phi_m\}_{m\in\nat}$ in $\Hp$,
let $A=\opPhi(A_{mn})$ be a bounded operator. By the first series expansion in~\eqref{eq.12},
for every vector $\psi\in\Hp$, we have:
\begin{equation}
\textstyle
A\psi=\sum_m\big(\sum_n A_{mn}\langle \phi_n,\psi\rangle\big)\phi_m, \quad
\mbox{where $\xm\equiv\{A_{mn}\nnat\in\ell^\infty$, for all $n\in\nat$}.
\end{equation}
Therefore, putting
\begin{equation}
\tfm=\ellphi(\xm)=\sum_{n\in\nat}A_{mn}\scaphin
\end{equation}
--- where $\ellphi\colon\ell^{\infty}\rightarrow\dH$ is the surjective isometry defined by~\eqref{defell}, and
convergence of the series w.r.t.\ the weak$^{\ast}$-topology is understood ---
we obtain that, for every $\psi\in\Hp$, $A\psi=\sum_m\tfm(\psi)\phi_m$, with $\|\tfm\|=\|\xm\|_{\infty}=\sup_{n\in\nat}|A_{mn}|$.
Moreover, by the corollary after Proposition~{4} of~\cite{Serre1962}, we conclude that $A$ is compact iff
$0=\lim_m\|\tfm\|=\lim_m\big({\textstyle\sup_{n\in\nat}}|A_{mn}|\big)$.

Decomposition~\eqref{decp} of a compact operator is essentially the equivalence of points~$(\alpha)$ and~$(\varepsilon)$
in Theorem~{4.40} of~\cite{rooij1978non}, but with some improvement that requires a suitable modification of
the proof therein. We outline the modified proof.

Let $C$ be a compact operator in $\Hp$, and let us assume that $C\neq 0$ (otherwise, there is nothing to prove).
By the preceding part of the proof, we argue that $C$ can be expressed in the form
\begin{equation} \label{predecp}
C\psi=\sum_{j\in J}\tfj(\psi)\,\osj , \quad \forall\psi\in\Hp,
\end{equation}
where $J=\{1,2,\ldots\}$ is a countable index set, $\tfj\colon\Hp\rightarrow\Q$, $j\in J$, is a \emph{nonzero}
bounded linear functional --- if $J=\nat$, such that $\lim_j \|\tfj\|=0$ --- and $\{\osj\}_{j\in J}$ is a
normalized norm-orthogonal system (in particular, it can be chosen to be a contained in any orthonormal basis).
Taking into account that $\|\Hp\|=|\Q|$, there is a subset $\{\lambda_j\}_{j\in J}$ of $\Q$ such that
$0<|\lambda_j|=\|\tfj\|$. It is then sufficient to put
\begin{equation}
\dsj\defi\frac{1}{\lambda_j}\tfj
\end{equation}
--- where $\{\dsj\}_{j\in J}$ is a set of normalized functionals in $\dH$ and, if $J=\nat$, $\lim_j \lambda_j=0$
--- to obtain decomposition~\eqref{decp} from~\eqref{predecp}.

We stress that, for $J=\nat$, since $\|\osj\pro\dsj\|\le 1$ --- and, hence, $\lim_j|\lambda_j|\,\|\osj\pro\dsj\|=0$ ---
the series in~\eqref{decp} converges not only w.r.t.\ the strong operator topology, but also w.r.t.\ the
norm topology.

Conversely, every linear operator $C$ of the form $C=\sum_{j\in J}\lambda_j\, \osj\pro\dsj$ --- with $\{\lambda_j\}_{j\in J}$,
$\{\osj\}_{j\in J}$ and $\{\dsj\}_{j\in J}$ as above (in particular, $\|\osj\|=\|\dsj\|=1$ and, if $J=\nat$, $\lim_j |\lambda_j|=0$)
--- is compact, because it is the norm-limit of a sequence of finite rank operators.
\end{proof}

From Theorem~\ref{theocpops} we derive three important consequences.

\begin{corollary}\label{matcarcom}
Given any orthonormal basis $\Phi\equiv\{\phi_m\}_{m\in\nat}$, a matrix operator $A=\opPhi(A_{mn})$
in $\Hp$ is compact iff
\begin{enumerate}[label=\tt{(C\arabic*)}]

\item \label{cond.c1} $\sup_{m,n}|A_{mn}|<\infty$,

\item \label{cond.c2} $\lim_m A_{mn}=0$, $\forall n\in\nat$,

\item \label{cond.c3} $\lim_{m,n}A_{mn}=0$ (Pringsheim limit).

\end{enumerate}
\end{corollary}

\begin{proof}
If $A=\opPhi(A_{mn})$ is compact, then it is bounded, so that, by Theorem~\ref{th.1}, it satisfies conditions~\ref{cond.c1}
and~\ref{cond.c2}. Moreover, it must satisfy condition~\eqref{conmatcom} in Theorem~\ref{theocpops}, as well.
The latter condition is easily shown to be equivalent to the pair of conditions formed by~\ref{cond.c2} (once again) and~\ref{cond.c3}.
Conversely, if $A=\opPhi(A_{mn})$ satisfies conditions~\ref{cond.c1}--\ref{cond.c3}, then it is bounded and verifies
condition~\eqref{conmatcom}, as well; hence, by the fourth assertion of Theorem~\ref{theocpops}, it is compact.
\end{proof}

\begin{corollary}\label{corut}
Every compact operator $C$ in $\Hp$ is uniformly traceable --- i.e., $\Cp\subset\Tw$ --- and, with
$\{\lambda_j\}_{j\in J}$, $\{\osj\}_{j\in J}$ and $\{\dsj\}_{j\in J}$ as in Theorem~\ref{theocpops},
\begin{equation} \label{comptra}
C=\sum_{j\in J}\lambda_j\, \osj\pro\dsj \implies
\tr(C)=\sum_{j\in J}\lambda_j\, \dsj(\osj) .
\end{equation}
\end{corollary}

\begin{proof}
For any orthonormal basis $\Phi\equiv\{\phi_m\mnat$ in $\Hp$, we have
\begin{equation}
\lim_{m+j}\lambda_j\,\langle\phi_m,\osj\rangle\,\dsj(\phi_m)=0,
\end{equation}
because  ($\lim_m\langle\phi_m,\osj\rangle=0\Longrightarrow$) $\lim_m\lambda_j\,\langle\phi_m,\osj\rangle\,\dsj(\phi_m)=0$,
for all $j\in\nat$, and, moreover, $\lim_j\lambda_j\,\langle\phi_m,\osj\rangle\,\dsj(\phi_m)=0$ \emph{uniformly} in $m\in\nat$
($|\lambda_j\,\langle\phi_m,\osj\rangle\,\dsj(\phi_m)|\le|\lambda_j|$); see relation~\eqref{limiffbis} in Remark~\ref{rem.6}.
It follows that, if $J=\nat$, the double series $\sum_{j,m\in\nat}\lambda_j\,\langle\phi_m,\osj\rangle\,\dsj(\phi_m)$ is
convergent and its sum coincides with the sum of both the iterated series (Remark~\ref{rem.6}); hence:
\begin{align}
\sum_{j\in\nat}\lambda_j\, \dsj(\osj)
& =
\sum_{j\in\nat}\lambda_j\,\dsj\big({\textstyle\sum_{m\in\nat}}\langle\phi_m,\osj\rangle\,\phi_m\big)
\nonumber\\
& =
\sum_{j\in\nat}\sum_{m\in\nat}\lambda_j\,\langle\phi_m,\osj\rangle\,\dsj(\phi_m)
\nonumber\\
& =
\sum_{j,m\in\nat}\lambda_j\,\langle\phi_m,\osj\rangle\,\dsj(\phi_m)=
\sum_{m\in\nat}\sum_{j\in\nat}\lambda_j\,\langle\phi_m,\osj\rangle\,\dsj(\phi_m)
=\tr_\Phi(C).
\end{align}
Here, the second equality follows from the continuity of the functional $\dsj$, whereas for the last equality we have used
the decomposition of $C$ (converging w.r.t.\ the norm topology) and the continuity of the inner product. Clearly, in the case
where $J=\{1,2,\ldots\}$ is a finite subset of $\nat$, one can freely exchange the finite sum with the series, so obtaining
the same result.

Now, since the quantity $\tr_\Phi(C)=\sum_{j\in J}\lambda_j\, \dsj(\osj)$ does not depend on the choice of
the orthonormal basis $\Phi\equiv\{\phi_m\mnat$, it turns out that every compact operator $C$ in $\Hp$
is uniformly traceable and relation~\eqref{comptra} holds true.
\end{proof}

\begin{corollary}\label{decocpa}
Every compact operator $C$ in $\Hp$, belonging to the closed subspace
\begin{equation}
\Cpa\defi\Cp\cap\Bpa
\end{equation}
of $\Cp$, can be expressed in the form
\begin{equation}\label{decpa}
C=\sum_{j\in J}\lambda_j\, |\osj\rangle\langle\fj|,
\end{equation}
where $J=\{1,2,\ldots\}$ is a countable index set and
\begin{itemize}

\item $\{\lambda_j\}_{j\in J}\subset\Q$ --- for $C\neq 0$, we assume that
$\{\lambda_j\}_{j\in J}\subset\Qa$ --- and, if $J=\nat$, $\lim_j \lambda_j=0$;

\item both the sets $\{\osj\}_{j\in J}$ and $\{\fj\}_{j\in J}$ are contained in $\Hp$, with $\|\osj\|=\|\fj\|=1$;

\item $\{\osj\}_{j\in J}$ --- or $\{\fj\}_{j\in J}$ --- is a (normalized) norm-orthogonal system in $\Hp$;

\item $|\osj\rangle\langle\fj|\colon\Hp\rightarrow\Hp$ is the bounded operator defined by
$(|\osj\rangle\langle\fj|)\,\psi\defi\langle\fj,\psi\rangle\,\osj$, and the sum in~\eqref{decp}
--- whenever $J$ is not finite --- converges w.r.t.\ the norm topology.

\end{itemize}

In particular, the norm-orthogonal set $\{\osj\}_{j\in J}$ --- alternatively, the norm-orthogonal set $\{\fj\}_{j\in J}$
--- can be chosen to be contained in any orthonormal basis in $\Hp$.

Conversely, every operator $C$ of the previous form --- i.e., such that $C\psi=\sum_{j\in J}\lambda_j\,\langle\fj,\psi\rangle\,\osj$,
for all $\psi\in\Hp$, with $\{\lambda_j\}_{j\in J}$, $\{\osj\}_{j\in J}$ and $\{\fj\}_{j\in J}$ as specified above --- belongs
to $\Cpa$.

Finally, $\Cpa$ is a two-sided $\ast$-ideal in $\Bpa$.
\end{corollary}

\begin{proof}
Let $C$ be an adjointable compact operator in $\Hp$. Then, by Theorem~\ref{theocpops},
we have that $C=\sum_{j\in J}\lambda_j\, \osj\pro\dsj$, with $\{\lambda_j\}_{j\in J}$,
$\{\osj\}_{j\in J}$ and $\{\dsj\}_{j\in J}$ as specified therein.

It is easy to check
that the generalized adjoint $C^\prime\in\mathcal{B}(\Hp^\prime)$ is given by
\begin{equation}\label{decpb}
C^\prime=\sum_{j\in J}\lambda_j\, \dsj\pro(\IH\osj), \qquad (\IH\osj\in\ddH).
\end{equation}
Here, $\IH\colon\Hp\rightarrow\ddH$ is the isometry defined by~\eqref{defIH}
and the operator $\dsj\pro(\IH\osj)\colon\dH\rightarrow\dH$ is of the form
$(\dsj\pro(\IH\osj))\,\phi^\prime=((\IH\osj)(\phi^\prime))\,\dsj=\phi^\prime(\osj)\,\dsj$,
for all $\phi^\prime\in\dH$; moreover, for $J=\nat$, the series converges in $\mathcal{B}(\Hp^\prime)$
w.r.t.\ the norm topology (because $\|\dsj\pro(\IH\osj)\|= 1$ and, hence,
$\lim_j|\lambda_j|\,\|\dsj\pro(\IH\osj)\|=0$). Since $C$ is adjointable, by Corollary~\ref{corhbe},
$C^\prime$ is a dual Hahn-Banach extension of the proper adjoint $C^\ast$ of $C$.

Let us assume that, in particular, the set $\{\osj\}_{j\in J}$ is contained in any --- arbitrarily chosen,
by one of the assertions of Theorem~\ref{theocpops} --- orthonormal basis.
It follows that, for every $k\in J$,
\begin{equation}
\big(C^\prime\circ\JH\big)(\osk)=\sum_{j\in J}\lambda_j\, \langle\osk,\osj\rangle\,\dsj
=\lambda_k\,\dsk\in\JH(\ran(C^\ast))\subset\JH(\Hp).
\end{equation}
(Recall that $\JH\colon\Hp\rightarrow\dH$ is the conjugate-linear isometry defined by~\eqref{defJH} and
the intertwining relation $A^\prime\circ\JH=\JH\circ C^\ast$ holds.) Therefore, $\{\dsj\}_{j\in J}\subset\ran(\JH)$.

Observe now that, defining $\fj\in\Hp$ by
\begin{equation}
\JH\fj=\dsj\in\ran(\JH), \quad \forall j\in J
\end{equation}
--- where the vectors $\{\fj\}_{j\in J}\subset\Hp$ are uniquely determined by the functionals $\{\dsj\}_{j\in J}$ and
$\|\fj\|=\|\dsj\|=1$, because $\JH$ is a (conjugate-linear) isometry --- we can introduce the bounded operator
\begin{equation}
D=\sum_{j\in J}\overline{\lambda_j}\, |\fj\rangle\langle\osj|.
\end{equation}
Here, if $J=\nat$, the series converges w.r.t.\ the norm topology, and $D$ is both compact and adjointable
(being the norm-limit of a sequence of \emph{adjointable} finite rank operators). Then, we have:
\begin{equation}
\big(C^\prime\circ\JH\big)(\psi)=\sum_{j\in J}\lambda_j\, \langle\psi,\osj\rangle\,\dsj
=\sum_{j\in J}\lambda_j\,\overline{\langle\osj,\psi\rangle}\,\JH\fj
=\big(\JH\circ D\big)(\psi), \quad \forall\psi\in\Hp.
\end{equation}
Hence, by the second assertion of Corollary~\ref{corhbe}, $\Cpa\ni D=C^\ast$ and, using the continuity
of the adjoining operation in $\Bpa$, we find that
\begin{equation}
C=D^\ast=\sum_{j\in J}\big(\overline{\lambda_j}\, |\fj\rangle\langle\osj|\big)^\ast=
\sum_{j\in J}\lambda_j\, |\osj\rangle\langle\fj|,
\end{equation}
where $\{\lambda_j\}_{j\in J}$, $\{\osj\}_{j\in J}$ are as specified in Theorem~\ref{theocpops} --- in particular,
the norm-orthogonal set $\{\osj\}_{j\in J}$ can be chosen to be contained in an orthonormal basis ---
and $\{\fj\}_{j\in J}$ is a set of normalized vectors.

Observe also that, by the first part of the proof, if $C$ is compact and adjointable,
then its adjoint is (both adjointable and) compact; i.e., $C\in\Cpa\implies C^\ast\in\Cpa$
(and, of course, $C^{\ast\ast}=C$), so that
\begin{equation}
\Cpa^\ast=\{C^\ast\sep C\in\Cpa\}\subset\Cpa
\end{equation}
and
\begin{equation}
\Cpa=\Cpa^{\ast\ast}\subset\Cpa^{\ast\ast\ast}=\Cpa^\ast.
\end{equation}
Hence, actually, $\Cpa=\Cpa^\ast$.

Therefore, every operator $D\in\Cpa$ is of the form $D=C^\ast$, for some $C\in\Cpa$, and,
again by the first part of the proof, we know that it can be written as
$\sum_{j\in J}\gamma_j\, |\fj\rangle\langle\osj|$, where:
$\{\gamma_j\equiv\overline{\lambda_j}\}_{j\in J}\subset\Q$ and, if $J=\nat$, $\lim_j \gamma_j=0$;
$\|\osj\|=\|\fj\|=1$; the norm-orthogonal system $\{\osj\}_{j\in J}$ can be chosen to be contained in any orthonormal basis in $\Hp$.
Thus, we have an alternative option for the choice of the norm-orthogonal system in decomposition~\eqref{decpa}
(where the vectors $\{\osj\}_{j\in J}$ now play the role of functionals, i.e., $\{\langle\osj,\cdot\hspace{0.4mm}\rangle\}_{j\in J}$).

The third assertion of the corollary is clear, since, if $C$ is of the form~\eqref{decpa}, then it is is both compact
and adjointable (being the norm-limit of a sequence of adjointable finite rank operators).
Finally, $\Cpa=\Cp\cap\Bpa$ is a two-sided $\ast$-ideal in $\Bpa$, because $\Cp$ is a two-sided ideal in $\Bp$, $\Bpa$ is an algebra
(w.r.t.\ composition of operators) and, as previously argued, $\Cpa=\Cpa^\ast$.
\end{proof}

\begin{remark}
The expressions~\eqref{decp} and~\eqref{decpa} may be regarded as a $p$-adic counterpart of the
\emph{singular value decomposition}~\cite{Reed,Simon,weidmann2012linear,moretti2018spectral}
of a compact operator in a separable complex Hilbert space. However, we stress that decompositions~\eqref{decp}
and~\eqref{decpa} --- in particular, the set of coefficients $\{\lambda_j\}_{j\in J}$ --- are not unique.
\end{remark}

\begin{definition}
Any (non-unique) decomposition of the form~\eqref{decp} (or of the form~\eqref{decpa}) will be called
\emph{a canonical decomposition} of the compact operator $C\in\Cp$ (respectively, of the adjointable compact operator
$C\in\Cpa$); in particular, we will assume that $\{\osj\}_{j\in J}$ --- or $\{\fj\}_{j\in J}$, in the case where $C\in\Cpa$
--- is a (normalized) norm-orthogonal system in $\Hp$. In the case where $\{\osj\}_{j\in J}$ --- alternatively,
$\{\fj\}_{j\in J}$, for $C\in\Cpa$ --- is chosen to be contained in an orthonormal basis, we will call \emph{orthonormal}
the associated canonical decomposition.
\end{definition}

We will now show that the trace class operators in $\Hp$ form a suitable class of compact operators.

\begin{definition}
We say that a  linear operator $A$ in $\Hp$ is \emph{block-finite} w.r.t.\ an orthonormal basis $\Phi\equiv\{\phi_m\}_{m\in\nat}$
if it is of the form $A=\opPhi(A_{mn})$, where, for some $k\in\nat$,
\begin{equation}
\max\{m,n\}>k \implies A_{mn}=0 .
\end{equation}
We say that $A$ is \emph{block-finite} (\emph{tout court}) if it is block-finite w.r.t.\ some orthonormal basis in $\Hp$.
\end{definition}

Clearly, all block-finite operators are of \emph{finite rank} (hence, compact), and the set of all
block-finite operators w.r.t.\ an orthonormal basis $\Phi\equiv\{\phi_m\}_{m\in\nat}$ is a linear subspace
$\blfi$ of $\Cp$.

\begin{theorem} \label{theoclo}
$\Tp$ is a closed linear subspace of $\Cp$. Specifically, $\Tp$ is the closure of the linear subspace $\blfi$ of $\Cp$,
for every orthonormal basis $\Phi\equiv\{\phi_m\}_{m\in\nat}$ in $\Hp$.
\end{theorem}

\begin{proof}
Let us first show that, for every orthonormal basis $\Phi\equiv\{\phi_m\}_{m\in\nat}$ in $\Hp$,
$\Tp\subset\overline{\blfi}^{\,\|\cdot\|}$. In fact, given any trace class operator $T=\opPhi(T_{mn})\in\Tp$,
we can define the matrix operator $\kT\defi\opPhi\big(\kTmn\big)$, $k\in\nat$, where
\begin{equation}
\kTmn=
\begin{cases}
\mbox{$T_{mn}$, if $\max\{m,n\}\le k$},
\\
\mbox{$0$, otherwise}.
\end{cases}
\end{equation}
Clearly, $\kT\in\blfi$ and, by construction,
\begin{equation}
\big\|T-\kT\big\|=\sup_{m,n}\big|T_{mn}-\kTmn\hspace{-0.6mm}\big|=\sup\{|T_{mn}|\sep \max\{m,n\}> k\}.
\end{equation}
Now, since $T$ is of trace class, for every $\epsilon>0$, there is some $j_\epsilon\in\nat$ such that
\begin{equation} \label{arga}
\max\{m,n\}>j_\epsilon \implies |T_{mn}|<\epsilon;
\end{equation}
hence, for every $k\ge j_\epsilon$,
\begin{equation} \label{argb}
\big\|T-\kT\big\|=\sup\{|T_{mn}|\sep \max\{m,n\}> k\}\le\sup\{|T_{mn}|\sep \max\{m,n\}> j_\epsilon\}\le\epsilon.
\end{equation}
Therefore, $\lim_k\big\|T-\kT\big\|=0$, hence, $\Tp$ is contained in the norm-closure of $\blfi$.

Let us next prove that $\Tp\supset\overline{\blfi}^{\,\|\cdot\|}$, as well, so that, actually,
$\Tp=\overline{\blfi}^{\,\|\cdot\|}$. Indeed, let us now suppose that $\big\{\kC=\opPhi(\kCmn)\big\}_{k\in\nat}$
is a sequence in $\blfi$ converging, in norm, to some (necessarily compact) operator $C=\opPhi(C_{mn})$.
Then, for every $\epsilon>0$, there is some $j_\epsilon\in\nat$ such that
\begin{equation}
k> j_\epsilon\implies\big\|C-\kC\big\|<\epsilon.
\end{equation}
Moreover, for every $k\in\nat$, since $\kC$ is block-finite w.r.t.\ $\Phi\equiv\{\phi_m\}_{m\in\nat}$, there
is some $l_k\in\nat$ such that
\begin{equation}
\max\{m,n\}>l_k \implies \kCmn=0.
\end{equation}
Therefore, for every $\epsilon>0$, there is some
$j_\epsilon\in\nat$ such that
\begin{align}
k> j_\epsilon\implies\epsilon>\big\|C-\kC\big\| & =\sup_{m,n}\big|C_{mn}-\kCmn\hspace{-0.6mm}\big|
\nonumber\\
& =\max\bigg\{\sup_{\max\{m,n\}\le l_k}\big|C_{mn}-\kCmn\hspace{-0.6mm}\big|, \sup_{\max\{m,n\}> l_k}|C_{mn}|\bigg\}.
\end{align}
In conclusion, for every $\epsilon>0$, there is some $l\in\nat$ --- say, $l\equiv l_k$, for any $k>j_\epsilon$
--- such that
\begin{equation}
\max\{m,n\}> l\implies |C_{mn}|<\epsilon.
\end{equation}
Otherwise stated, $\lim_{m+n}|C_{mn}|=0$, so that $C=\opPhi(C_{mn})=\lim_k\kC$ is a trace class operator
and --- by the arbitrariness of the converging sequence $\big\{\kC\big\}_{k\in\nat}\subset\blfi$ ---
$\Tp\supset\overline{\blfi}^{\,\|\cdot\|}$.

Eventually, it is shown that $\Tp=\overline{\blfi}^{\,\|\cdot\|}$ and the proof is complete.
\end{proof}

We will next prove a remarkable characterization of the trace class of $\Hp$.

\begin{theorem} \label{theochar}
The following characterization of the trace class of $\Hp$ holds true:
\begin{equation}
\Tp=\Cpa\defi\Cp\cap\Bpa.
\end{equation}
Moreover, given any canonical decomposition of a trace class operator $T\in\Tp=\Cpa$
--- i.e., $T=\sum_{j\in J}\lambda_j\, |\osj\rangle\langle\fj|$, where $\{\lambda_j\}_{j\in J}$,
$\{\osj\}_{j\in J}$ and $\{\fj\}_{j\in J}$ are as specified in Corollary~\ref{decocpa}; in particular,
$\{\osj\}_{j\in J}$, or $\{\fj\}_{j\in J}$, is a normlized norm-orthogonal system in $\Hp$ ---
we have that
\begin{equation} \label{tctra}
\tr(T)=\sum_{j\in J}\lambda_j\, \langle\fj,\osj\rangle ,
\end{equation}
and the following estimate holds:
\begin{equation} \label{tcnor}
|\tr(T)|\le\max_{j\in J} |\lambda_j|\,|\langle\fj,\osj\rangle|\le\|T\|=\max_{j\in J} |\lambda_j|.
\end{equation}
{\rm (Compare with the second inequality in~\eqref{trinequa}.)}
\end{theorem}

\begin{proof}
We already know that $\Tp\subset\Cp\cap\Bpa$. Let us show that this inclusion is, actually, an equality.
In fact, given any orthonormal basis $\Phi\equiv\{\phi_m\}_{m\in\nat}$ in $\Hp$, by Corollary~\ref{matcarcom}
a compact operator $A=\opPhi(A_{mn})\in\Cp$ must verify conditions~\ref{cond.c1}--\ref{cond.c3} therein.
If, in addition, $A$ is adjointable, then, by Theorem~\ref{th.2},  we also have that
\begin{enumerate}[label=\tt{(C\arabic*)}]
\setcounter{enumi}{3}

\item \label{cond.c4} $\lim_n A_{mn}=0$, $\forall m\in\nat$.

\end{enumerate}
By relation~\eqref{limiff} in Remark~\ref{rem.6}, conditions~\ref{cond.c2}--\ref{cond.c4} are equivalent to
$\lim_{m+n}A_{mn}=0$ (and condition~\ref{cond.c1} becomes redundant). Hence, $A$ is a trace class operator.

The proof of relation~\eqref{tctra} is similar to the proof of Corollary~\ref{corut}: for any orthonormal basis
$\Phi\equiv\{\phi_m\mnat$ in $\Hp$,
\begin{align}
\tr(T)=\tr\big({\textstyle\sum_{j\in J}}\lambda_j\, |\osj\rangle\langle\fj|\big) & =
\sum_{m\in\nat}\sum_{j\in J}\lambda_j\,\langle\phi_m,\osj\rangle\langle\fj,\phi_m\rangle
\nonumber\\
& =\sum_{j\in J}\sum_{m\in\nat}\lambda_j\,\langle\fj,\phi_m\rangle\langle\phi_m,\osj\rangle
=\sum_{j\in J}\lambda_j\, \langle\fj,\osj\rangle .
\end{align}
Here, if $J=\nat$, exchanging the sums is justified by the fact that
$\lim_{m+j}\lambda_j\,\langle\phi_m,\osj\rangle\langle\fj,\phi_m\rangle=0$,
because $\lim_m\lambda_j\langle\phi_m,\osj\rangle\langle\fj,\phi_m\rangle=0$, for all $j\in\nat$,
and $\lim_j\lambda_j\,\langle\phi_m,\osj\rangle\langle\fj,\phi_m\rangle=0$ \emph{uniformly}
in $m\in\nat$ ($|\lambda_j\,\langle\phi_m,\osj\rangle\langle\fj,\phi_m\rangle|\le|\lambda_j|$).

Let us now prove that $\|T\|=\max_{j\in J} |\lambda_j|$. Since $\|T\|=\|T^\ast\|$, in the following we can assume,
without loss of generality, that $\{\osj\}_{j\in J}$ (rather than $\{\fj\}_{j\in J}$) is a normalized norm-orthogonal
system in $\Hp$. Hence, for every vector $\psi\in\Hp$, we have that
$\|T\psi\|=\|\sum_{j\in J}\lambda_j\,\langle\fj,\psi\rangle\,\osj\|=\sup_{j\in J}|\lambda_j|\,|\langle\fj,\psi\rangle|$,
and
\begin{equation}
\|T\|=\sup_{\psi\neq 0}\frac{\|T\psi\|}{\|\psi\|}=
\sup_{j\in J}\bigg(|\lambda_j|\,\sup_{\psi\neq 0}\frac{|\langle\fj,\psi\rangle|}{\|\psi\|}\bigg)
=\sup_{j\in J}|\lambda_j|\,\|\JH\fj\|=\max_{j\in J}|\lambda_j|,
\end{equation}
where we have used the fact that $\JH$ is a (conjugate-linear) isometry. Then, since
$|\langle\fj,\osj\rangle|\le 1$, the estimate~\eqref{tcnor} holds true.
\end{proof}

\begin{remark} \label{reincls}
By Theorem~\ref{theocpops}, Corollary~\ref{matcarcom} and Corollary~\ref{decocpa}, it is clear that \emph{not}
every compact operator is adjointable and then $\Tp=\Cpa\subsetneq\Cp\subset\Tw$ ($\dim(\Hp)=\infty$). For instance,
the bounded operator $TB$ in Example~\ref{exnontc} is compact but not adjointable.
\end{remark}

In a infinite-dimensional separable \emph{complex} Hilbert space, the product of two trace class operators
is of trace class too, \emph{but} not every trace class operator is the product of two trace class
operators (instead, it can expressed as the product of two Hilbert-Schmidt operators);
see~\cite{Reed,Simon,weidmann2012linear,moretti2018spectral}. In a $p$-adic Hilbert space $\Hp$,
putting
\begin{equation}
\Tp^2\defi\{S\,T \sep S,T\in\Tp\},
\end{equation}
we have that $\Tp^2\subset\Tp$, because $\Tp$ is a two-sided ideal in $\Bpa$; actually, from Theorem~\ref{theochar}
we derive the following:

\begin{corollary}
$\Tp^2=\Tp$. In particular, every trace class operator $R\in\Tp$ can be expressed in the form
$R=S\,T$, for some $S,T\in\Tp$.
\end{corollary}

\begin{proof} We only need to prove that $\Tp^2\supset\Tp$; i.e., that every $R\in\Tp$ is
of the form $R=S\,T$, for suitable $S,T\in\Tp$. Since $\Tp=\Cpa$, we can write
$R=\sum_{j\in J}\lambda_j\, |\osj\rangle\langle \fj|$, with $J=\{1,2,\ldots\}\subset\nat$, and $\{\lambda_j\}_{j\in J}$,
$\{\osj\}_{j\in J}$, $\{\fj\}_{j\in J}$ as specified in Corollary~\ref{decocpa}.
Now, given any orthonormal basis $\Phi\equiv\{\phi_m\mnat$ in $\Hp$, let us set
$S=\sum_{j\in J}\kappa_j\, |\osj\rangle\langle \phi_j|$, $T=\sum_{j\in J}\nu_j\, |\phi_j\rangle\langle \fj|$, where
$\kappa_j\nu_j=\lambda_j$ and, if $J=\nat$, $\lim_j\kappa_j=0=\lim_j\nu_j$. The existence of suitable sets $\{\kappa_j\}_{j\in J}$,
$\{\nu_j\}_{j\in J}$ in $\Q$ satisfying the previous conditions is guaranteed by Lemma~{8.1.5} in~\cite{Perez-Garcia}.
Therefore, $\{\kappa_j\}_{j\in J}$, $\{\osj\}_{j\in J}$ and $\{\phi_j\}_{j\in J}$ --- and, analogously,
$\{\nu_j\}_{j\in J}$, $\{\phi_j\}_{j\in J}$ and $\{\fj\}_{j\in J}$ --- are as prescribed in Corollary~\ref{decocpa}.
Hence, $S,T\in\Cpa=\Tp$ (note that Theorem~\ref{theochar} is essential here), and, by construction,
$S\,T=\sum_{j\in J}\kappa_j\nu_j\, |\osj\rangle\langle \fj|=\sum_{j\in J}\lambda_j\, |\osj\rangle\langle \fj|=R$.
\end{proof}


\subsection{The $p$-adic Hilbert-Schmidt space}
\label{Hilbert-Schmidt}

In the light of the results of the previous section, it should not be surprising that, in the $p$-adic setting,
$\Tp$ actually plays a two-fold role: the trace class and the Hilbert-Schmidt space.

In fact, let us introduce the sesquilinear form
\begin{equation}
\Tp\times\Tp\ni(S,T)\mapsto\tr(S^*T)\ifed\langle S,T\hsp \in\Q ,
\end{equation}
which is Hermitian, because
\begin{equation}
\langle S,T\hsp\defi\tr(S^*T)=\overline{\tr(T^*S)}=\overline{\langle T,S\rangle}_{\hspace{-0.5mm}\mbox{\tiny $\Tp$}} .
\end{equation}
Notice that, here, for obtaining the second equality, we have used property~\ref{conjtrace} of the trace (see
Proposition~\ref{prop.tnew}).

We will call the Hermitian sesquilinear form $\inpr$ in $\Tp$ the ($p$-adic) \emph{Hilbert-Schmidt product}.

Given any orthonormal basis $\Phi\equiv\{\phi_m\}_{m\in\nat}$ in $\Hp$, we can also consider the family
of matrix operators $\big\{\jkE\big\}_{j,k\in\nat}\subset\blfi$ defined by
\begin{equation}
\jkE\defi\opPhi\big(\jkEmn\big), \ \mbox{where $\jkEmn=\delta_{jm}\delta_{kn}$;}
\end{equation}
namely, in the usual  Dirac notation, $\jkE=|\phi_j\rangle\langle\phi_k|$ (i.e., $\jkEmn\psi=\langle\phi_k,\psi\rangle\phi_j$).
Note that, for every trace class operator $T\defi\opPhi(T_{mn})$, we have:
\begin{equation}
\big\langle \jkE, T\bhsp=\tr(|\phi_k\rangle\langle\phi_j|\, T)= T_{jk}.
\end{equation}
It follows that the Hermitian sesquilinear form $\inpr$ is non-degenerate, because
\begin{equation}
\big\langle T,\jkE\bhsp=0,\ \forall j,k\in\nat  \implies T=0.
\end{equation}

\begin{theorem} \label{theohs}
The $p$-adic Banach space $\Tp$ --- endowed with the $p$-adic Hilbert-Schmidt product $\inpr$ --- becomes an inner product
$p$-adic Banach space. Moreover, for every orthonormal basis $\Phi\equiv\{\phi_m\}_{m\in\nat}$ in $\Hp$,
$\{\jkE\}_{j,k\in\nat}$ is an orthonormal basis in $\Tp$. Therefore, the triple $(\Tp,\|\cdot\|,\inpr)$ is, actually,
a $p$-adic Hilbert space.
\end{theorem}

\begin{proof}
We have already shown that $\Tp$, endowed with the operator norm, is a $p$-adic Banach space, and that
the sesquilinear form $\inpr$ is both Hermitian and non-degenerate.

Observe now that, for all $S,T\in\Tp$, we have:
\begin{equation}
|\langle S,T\hsp|=|\tr(S^*T)|=|{\textstyle \sum_m}\langle S\phi_m, T\phi_m\rangle|\le
\max_m |\langle S\phi_m, T\phi_m\rangle|\le \|S\|\, \|T\|;
\end{equation}
i.e., $\inpr$ satisfies the Cauchy-Schwarz inequality, as well. Therefore, $\inpr$ is an inner product, and
$\Tp$, endowed with this sesquilinear form, is an inner product $p$-adic Banach space.

It remains to show that $\big\{\jkE\big\}_{j,k\in\nat}$ is an orthonormal basis in $\Tp$. Since it is clear that
\begin{equation}
\big\langle\jkE,\rsE\bhsp=\tr(|\phi_k\rangle\langle\phi_j|\, |\phi_r\rangle\langle\phi_s| )=\langle\phi_j,\phi_r\rangle
\langle\phi_s,\phi_k\rangle=\delta_{jr}\delta_{ks},
\end{equation}
we only need to prove that $\big\{\jkE\big\}_{j,k\in\nat}$ is a \emph{normal} basis. In fact, for every finite subset $I$
of $\nat\times\nat$ and every finite subset $\{\alpha_{jk}\}_{j,k\in I}$ of $\Q$, we have:
\begin{equation}
\left\|{\textstyle\sum_{j,k\in I}}\,\alpha_{jk}\, \jkE\right\|=\max_{j,k\in I}|\alpha_{jk}|.
\end{equation}
Moreover, for every trace class operator $T\defi\opPhi(T_{mn})$, we have that
\begin{equation}
T=\lim_l \lT,\ \mbox{where $\lT\defi\sum_{\max\{j,k\}\le l} T_{jk}\,\jkE$}.
\end{equation}
This fact is a consequence of the estimate
\begin{equation}
\big\|T-\lT\hspace{-0.6mm}\big\|=\sup_{m,n}\big|T_{mn}-\lTmn\hspace{-0.6mm}\big|=\sup\{|T_{mn}|\sep \max\{m,n\}> l\},
\end{equation}
together with the same argument used in the first part of the proof of Theorem~\ref{theoclo},
which shows that --- $T$ being of trace class --- $\lim_l\big\|T-\lT\hspace{-0.6mm}\big\|=0$.

In conclusion, $\big\{\jkE\big\}_{j,k\in\nat}$ is an orthonormal basis in the inner product $p$-adic
Banach space $\Tp$, which is then a $p$-adic Hilbert space.
\end{proof}

The $p$-adic Hilbert space $(\Tp,\|\cdot\|,\inpr)$ will be called the \emph{$p$-adic Hilbert-Schmidt space}.


\subsection{Selfadjoint trace class operators}
\label{selfadjoint}

Let us now consider the $\mathbb{Q}_p$-linear space $\Tps\defi\Tp\cap\Bps=\Cpa\cap\Bps$ of all selfadjoint
trace class operators in the $p$-adic Hilbert space $\Hp$, that is closed in $\Tp$, because the mapping
$\Tp\ni T\mapsto T^\ast\in\Tp$ is a (conjugate-linear) isometry and, hence, continuous (thus, $\Tps$, endowed
with the operator norm, is an ultrametric Banach space over $\Qp$).

\begin{proposition} \label{prosydec}
Every  selfadjoint trace class operator $T\in\Tps$ can be expressed in the form
\begin{equation}\label{sydec}
T=\sum_{j\in J}(\sj\, |\osj\rangle\langle\fj|+\overline{\sj}\, |\fj\rangle\langle\osj|),
\end{equation}
where $J=\{1,2,\ldots\}$ is a countable index set and
\begin{itemize}

\item $\{\sj\}_{j\in J}\subset\Q$ --- for $T\neq 0$, we assume that
$\{\sj\}_{j\in J}\subset\Qa$ --- and, if $J=\nat$, $\lim_j\sj=0$;

\item $\{\osj\}_{j\in J}$ is a normalized norm-orthogonal system in $\Hp$,
and $\|\fj\|=1$, for all $j\in J$;

\item the sum in~\eqref{sydec} --- whenever $J$ is not finite --- converges w.r.t.\ the norm topology.

\end{itemize}

In particular, the norm-orthogonal system $\{\osj\}_{j\in J}$ can be chosen to be contained
in any orthonormal basis in $\Hp$.

Conversely, every linear operator $T$ of the previous form belongs to $\Tps$, and
\begin{equation} \label{trasa}
\tr(T)=2\sum_{j\in J}\scj(\sj\,\langle\fj,\osj\rangle)=\sum_{j\in J}(\sj\,\langle\fj,\osj\rangle
+\overline{\sj}\,\langle\osj,\fj\rangle)\in\mathbb{Q}_p;
\end{equation}
moreover, $|\tr(T)|\le\|T\|\le\max_{j\in J} |\sj|$.
\end{proposition}

\begin{proof}
Clearly, a trace class operator $T\in\Tp$ is selfadjoint iff it is of the form $T=A+A^\ast$,
for some $A\in\Tp=\Cpa$. Then, by Corollary~\ref{decocpa}, $A=\sum_{j\in J}\sj\, |\osj\rangle\langle\fj|$
--- with $\{\sj\}_{j\in J}$, $\{\osj\}_{j\in J}$ and $\{\fj\}_{j\in J}$ as above --- so that $T\in\Tps$
iff it is of the form~\eqref{sydec}, and then formula~\eqref{trasa} follows immediately from~\eqref{tctra}.
Moreover, by the estimate~\eqref{tcnor}, we have that
$|\tr(T)|\le\|T\|=\|A+A^\ast\|\le\max\{\|A\|,\|A^\ast\|\}=\|A\|=\max_{j\in J} |\sj|$.
\end{proof}

\begin{definition}
Given a selfadjoint trace class operator $T\in\Tps$, an expression of the form~\eqref{sydec} will be called
\emph{a symmetric decomposition} of $T$. In the case where $\{\osj\}_{j\in J}$ is chosen to be contained in
an orthonormal basis, we will call \emph{orthonormal} the associated symmetric decomposition.
\end{definition}


\section{A $p$-adic model for quantum states}
\label{sec7}

Building on the foundations laid down in the preceding sections, we will now attempt at achieving a general definition
of a quantum state in the $p$-adic setting. As usual, the standard complex case will provide us with a useful road map,
\emph{but}, when dealing with a $p$-adic Hilbert space, the emergence of non-trivial peculiarities should be expected.


\subsection{The complex setting in a nutshell}
\label{complex}

Since we do expect that the general lines, rather than the peculiar features, of the theory will be preserved
when switching from the complex to the $p$-adic case, it may be sensible to consider, as a starting point,
the most abstract formulation of standard quantum mechanics, i.e., the so-called
\emph{algebraic formulation}~\cite{strocchi2008introduction,moretti2018spectral,giachetta2005geometric,bratteli2012operator}.
This formulation relies on the following set of fundamental assumptions:
\begin{itemize}

\item A quantum system can be described by means of two main classes of objects --- \emph{states} and \emph{observables} ---
mutually related by means of a natural \emph{pairing} map. By suitably exploiting these two kinds of objects, one can then
construct all other parts of the theory: measurements, symmetry transformations, dynamics etc.

\item The (bounded) observables of the system are supposed to form the selfadjoint part $\calgs$ of an abstract
\emph{non-commutative unital $C^\ast$-algebra} $\calg$.

\item A generic \emph{state} $\omega$ (of $\calg$) is defined as a \emph{normalized positive functional} on $\calg$; i.e., as a functional
$\omega\colon\calg\rightarrow\mathbb{C}$ satisfying the conditions
\begin{equation}
\omega(A^*A)\geq 0, \quad\forall A\in\calg, \qquad\omega(\mathrm{Id})=1.
\end{equation}
Here, the \emph{positive elements} $A^*A$ of $\mathfrak{A}$ form a convex cone.

\item Denoting by $\stalg$ the set of all states of the $C^\ast$-algebra $\calg$, the pairing between observables
and states is provided by the evaluation map $\calgs\times\stalg\ni(A,\omega)\mapsto\omega(A)$.
\end{itemize}

From these assumptions, one can then derive the following main facts:
\begin{enumerate}

\item Every state $\omega\colon\calg\rightarrow\mathbb{C}$ is automatically continuous (i.e., bounded, as
a linear functional); specifically, it turns out that $\|\omega\|=\omega(\id)=1$.

\item $\stalg$ is a convex subset of the (complex) Banach space of bounded functionals on $\mathfrak{A}$.

\item For every $A\in\calg$ and every state $\omega\in\stalg$, $\omega(A^\ast)=\overline{\omega(A)}$.

\item In particular, for every observable $A\in\calgs$ and every state $\omega\in\stalg$, the \emph{real} quantity $\omega(A)$
--- i.e., the pairing of $A$ with $\omega$ --- can be interpreted as the \emph{expectation value} of the observable $A$ when
the physical system is in the state $\omega$.

\item By the celebrated \emph{Gelfand-Naimark theorem}~\cite{strocchi2008introduction,moretti2018spectral, giachetta2005geometric},
$\calg$ can be realized as --- i.e., is isometrically $\ast$-isomorphic to --- a $C^\ast$-subalgebra $\salge$ of the $C^\ast$-algebra of
all bounded operators $\mathcal{B}(\mathcal{K})$ in a complex Hilbert space $\mathcal{K}$. For the sake of simplicity, we will suppose henceforth
that $\mathcal{K}$ is separable and $\mathfrak{C}=\mathcal{B}(\mathcal{K})$ (this is the case of `ordinary' quantum mechanics).

\item By the previous identification of $\calg$ with $\mathcal{B}(\mathcal{K})$, we can single out a distinguished class of states
--- the so-called \emph{trace-induced states} $\trast$ --- that can be defined by
\begin{equation}
\omega\in\trast\iffdef\mbox{$\omega=\tr(\argo\rho_{\omega})\colon\calg\rightarrow\mathbb{C}$, for some $\rho_{\omega}\in\densK$},
\end{equation}
where $\densK\subset\mathcal{T}(\mathcal{K})$ is the convex set of all unit-trace positive trace class operators in $\mathcal{K}$,
the so-called \emph{density} or \emph{statistical} operators.

\item It is worth stressing that, in the case where $\dim(\mathcal{K})=\infty$, $\trast\subsetneq\stalg$.
There is a remarkable characterization of trace-induced states as those states that are \emph{$\sigma$-additive}~\cite{Emch,moretti2018spectral}.
Identifying the abstract algebra $\calg$ with $\mathcal{B}(\mathcal{K})$, a state $\omega\in\cst$ is
$\sigma$-additive if
\begin{equation}
\omega\big({\textstyle\sum_{j\in J}}P_j\big)=\sum_{j\in J}\omega(P_j),
\end{equation}
for every (countable) family $\{P_j\}_{j\in J}$ of pairwise orthogonal projections in $\mathcal{K}$, where
the possibly infinite sum $\sum_{j\in J}P_j$ is supposed to converge w.r.t.\ the weak operator topology
(it actually converges w.r.t.\ the strong operator topology, as well). Therefore, $\omega\in\cst$ is
$\sigma$-additive iff $\omega=\tr(\argo\rho_{\omega})$, for some density operator $\rho_{\omega}\in\densK$.
The role, the meaning and the relevance of those states that are \emph{not} $\sigma$-additive is controversial~\cite{Emch},
and one often restricts to the trace-induced ones; equivalently, to density operators. This is analogous to restricting
to $\sigma$-additive probability measures in classical statistical mechanics.

\item The \emph{spectral decomposition} $A=\int_\mathbb{R}\lambda\,\de\spm(\lambda)$ of a selfadjoint operator in
$\mathcal{K}$ --- where  $\spm$ is the spectral measure uniquely associated with $A$ --- allows one to complete
the probabilistic interpretation of the theory. In particular, it shows that every (bounded or unbounded) observable
can be expressed in terms of the \emph{lattice of projections} $\proK\subset\bsa$, whose elements are then regarded as
the \emph{elementary propositions} of the theory~\cite{Emch}.

\item Eventually, one is led in a natural way to describe the observables of a quantum system in terms of PVMs (projection-valued
measures) or, more generally, of POVMs (positive-operator-valued measures, also called ``semispectral measures'')~\cite{Holevo,Busch,Heino}.
The --- both conceptually and mathematically transparent --- generalization of PVMs into POVMs has a remarkable physical interpretation
related to the theory of open quantum systems (Naimark's dilation theorem~\cite{Holevo}).

\end{enumerate}


\subsection{Convexity and probability in the $p$-adic setting}
\label{convexity}

Quantum probability theory is tailored on classical probability theory, of which it can be regarded as a non-commutative
counterpart. This is not surprising because the outcome of a quantum measurement process must be, ultimately, a classical
probability distribution. In particular, both theories share essentially the same notion of convexity.

Clearly, the basic rules of the game must be re-written when switching to the $p$-adic setting. We start with briefly
introducing the $p$-adic (or, more generally, non-Archimedean) notion of convexity. Our treatment will be rather sketchy;
for further details, the reader may refer to Sect.~{2.5} of~\cite{narici71} and Sect.~{3.1} of~\cite{Perez-Garcia}.
Moreover, we will adapt the main definitions and results to the special case that will be considered in the next subsection.

Let $(X,\|\cdot\|)$ be a normed space over $\Q$. By field restriction, we can regard it as a vector space over
$\Qp$ and consider a notion of $\Qp$-convexity (rather than $\Q$-convexity). We will keep trace of this choice ---
essentially motivated by our objectives --- in the notation that will be adopted.

\begin{definition}
A subset $\acoset$ of $X$ is said to be \emph{absolutely $\Qp$-convex} if ($0\in\acoset$ and) $\lambda x + \mu y\in\acoset$,
for all $x,y\in \acoset$ and all $\lambda,\mu\in\pint$, where $\pint=\{\lambda\in\Qp \sep |\lambda|\le 1\}$ is the ring
of $p$-adic integers. Given any subset $\suse$ of $X$, its \emph{absolutely $\Qp$-convex hull} $\acop(\suse)$ is defined as the intersection
of all absolutely $\Qp$-convex sets containing $\suse$.
\end{definition}

We have that $\acop(\emptyset)=\emptyset$ and, if $\suse\neq\emptyset$,
\begin{equation}
\acop(\suse)=\{\lambda_1 x_1 + \cdots +\lambda_n x_n \sep
n\in\nat,\ x_1, \ldots ,x_n\in\suse,\ \lambda_1,\ldots ,\lambda_n\in\pint\}.
\end{equation}
We will denote by $\cacop(\suse)$ the (norm-)closure of the set $\acop(\suse)$.

\begin{definition}
A subset $\coset$ of $X$ is said to be \emph{$\Qp$-convex} if it is either empty or of the form
$x+\acoset$, for some $x\in X$ and some (nonempty) absolutely $\Qp$-convex subset $\acoset$ of $X$.
Given any subset $\suse$ of $X$, its \emph{$\Qp$-convex hull} $\cop(\suse)$ is defined as the intersection
of all $\Qp$-convex sets containing $\suse$.
\end{definition}

We will denote by $\ccop(\suse)$ the closure of the set $\cop(\suse)$. Given a $p$-adic Hilbert space $\Hp$
and an orthonormal basis $\Phi\equiv\{\phi_m\menne$ (where $\enne\in\nat$ or $\enne=\infty$) in $\Hp$,
the closed $\Qp$-convex hull $\ccop(\Phi)$ is said to be a \emph{$\Qp$-simplex} in $\Hp$.

One can easily check the following facts:
\begin{itemize}

\item A $\Qp$-convex subset of $X$ is absolutely $\Qp$-convex iff it contains 0.

\item
$\cop(\suse)=\big\{\lambda_1 x_1 + \cdots +\lambda_n x_n \sep
n\in\nat,\ x_1, \ldots ,x_n\in\suse,\ \lambda_1,\ldots ,\lambda_n\in\pint,
\ {\textstyle\sum_{k=1}^n}\lambda_k=1\big\}$.

\end{itemize}

\begin{definition}
A map $g\colon X\rightarrow Y$ --- where $Y$ is a vector space over $\Q$ ---
is said to be \emph{$\Qp$-convex} if
\begin{equation}
g(\lambda_1 x_1 + \lambda_2 x_2)=\lambda_1 g(x_1) + \lambda_2 g(x_2),\
\mbox{$\forall x_1, x_2\in X$, $\forall\lambda_1, \lambda_2\in\pint$
such that $\lambda_1 + \lambda_2=1$.}
\end{equation}
\end{definition}

\begin{theorem} \label{charconv}
Let $(X,\|\cdot\|)$ be a normed space over $\Q$.

For $p\neq 2$, a subset $\coset$ of $X$ is $\Qp$-convex iff
\begin{equation} \label{concort}
\lambda_1 x_1 + \lambda_2 x_2\in\coset,\
\mbox{$\forall x_1, x_2\in\coset$, $\forall\lambda_1, \lambda_2\in\pint$
such that $\lambda_1 + \lambda_2=1$,}
\end{equation}
whereas this condition is (necessary but) \emph{not} sufficient in the case where $p=2$.

For $p=2$, a subset $\coset$ of $X$ is $\Qp$-convex iff
\begin{equation} \label{consta}
\lambda_1 x_1 + \lambda_2 x_2 +\lambda_3 x_3\in\coset,\
\mbox{$\forall x_1, x_2, x_3 \in\coset$, $\forall\lambda_1, \lambda_2, \lambda_3\in\pint$
such that $\lambda_1 + \lambda_2 +\lambda_3=1$.}
\end{equation}
\end{theorem}

\begin{proof}
That a convex combination of three elements of a set belongs to this set is the standard
sufficient condition ensuring convexity in a non-Archimedean normed space; see Theorem~{3.1.15}
in~\cite{Perez-Garcia}. Then, condition~\eqref{consta} is (necessary and) sufficient for the $\Qp$-convexity
of $X$. Moreover, in the case where $p\neq 2$, the residue class field $\rcf$ of $\Qp$ consists of
$p\ge 3$ elements, so that, by Theorem~{3.1.17} in~\cite{Perez-Garcia}, the milder condition~\eqref{concort}
is sufficient too, whereas this condition is \emph{not} sufficient for $p=2$. (We stress that we can apply
the previously mentioned results in non-Archimedean convexity because, by field restriction, $X$ can be
regarded as a vector space over $\Qp$.)
\end{proof}

\begin{corollary} \label{corcoma}
The range $g(X)$ of a $\Qp$-convex map $g\colon X\rightarrow Y$ --- where $Y$ is a vector space over $\Q$ ---
is a $\Qp$-convex subset of $Y$.
\end{corollary}

\begin{proof}
Just note that, if $g$ is $\Qp$-convex, then also
\begin{equation}
g(\lambda_1 x_1 + \lambda_2 x_2 + \lambda_3 x_3)=\lambda_1 g(x_1) + \lambda_2 g(x_2) + \lambda_3 g(x_3),
\end{equation}
for all $x_1, x_2, x_3\in X$, and for all $\lambda_1, \lambda_2, \lambda_3\in\pint$ such that
$\lambda_1 + \lambda_2 + \lambda_3=1$ (assuming that, say, $\lambda_1 + \lambda_2\neq 0$, write
$\lambda_1 x_1 + \lambda_2 x_2 = (\lambda_1 + \lambda_2)(\lambda_1 x_1/(\lambda_1 + \lambda_2)
+ \lambda_2 x_2/(\lambda_1 + \lambda_2))$, where $\lambda_1/(\lambda_1 + \lambda_2),
\lambda_2/(\lambda_1 + \lambda_2)\in\pint$), and apply Theorem~\ref{charconv}.
\end{proof}

\begin{definition} \label{defaf}
A nonempty subset $\affset$ of $X$ is said to be \emph{$\Qp$-affine} if it is of the form
$x+\lis$, for some $x\in X$ and some $\Qp$-linear subspace $\lis$ of $X$.
Given any nonempty subset $\suse$ of $X$, its \emph{$\Qp$-affine hull} $\aff(\suse)$ is the intersection
of all $\Qp$-affine sets containing $\suse$.
\end{definition}

Clearly, a $\Qp$-affine subset of $X$ is also $\Qp$-convex, because every $\Qp$-linear subspace of $X$
is a $\Qp$-absolutely convex subset of $X$. We will denote by $\caff(\suse)$ the closure of the set $\aff(\suse)$.
Given a $p$-adic Hilbert space $\Hp$ and an orthonormal basis $\Phi\equiv\{\phi_m\menne$ in $\Hp$, the closed
$\Qp$-affine hull $\caff(\Phi)$ is said to be a \emph{$\Qp$-hyperplane} in $\Hp$.

One can easily prove the following facts:
\begin{itemize}

\item A $\Qp$-affine subset of $X$ is a $\Qp$-linear subspace iff it contains 0.

\item A nonempty subset $\affset$ of $X$ is $\Qp$-affine iff, for every pair of vectors $x,y\in\affset$,
the \emph{$\Qp$-line} $\{x+\alpha(y-x)\}_{\alpha\in\Qp}=\{y+\alpha(x-y)\}_{\alpha\in\Qp}$ through $x$ and
$y$ is contained in $\affset$; namely, iff $\alpha x + (1-\alpha) y\in\affset$, for all $x,y\in \affset$ and
$\alpha\in\Qp$.

\item if $g\colon X\rightarrow Y$, where $Y$ is a vector space over $\Q$,
is \emph{$\Qp$-affine} --- $g(\alpha x + (1-\alpha) y)=\alpha g(x) + (1-\alpha) g(y)$, for all $x,y\in X$,
and $\alpha\in\Qp$ --- then $g(X)$ is a $\Qp$-affine subset of $Y$.

\item
$\aff(\suse)=\big\{\pi_1 x_1 + \cdots +\pi_n x_n \sep
n\in\nat,\ x_1, \ldots ,x_n\in\suse,\ \pi_1,\ldots ,\pi_n\in\Qp,
\ {\textstyle\sum_{k=1}^n}\pi_k=1\big\}$.

\end{itemize}

\begin{proposition}
Given a map $g\colon X\rightarrow Y$ --- where $Y$ is a vector space over $\Q$ --- the following
facts are equivalent:
\begin{enumerate}[label=\tt{(K\arabic*)}]

\item \label{conl} $g$ is $\Qp$-convex;

\item \label{affl} $g$ is $\Qp$-affine;

\item \label{qual} $g$ is of the form $g(x)=g(0)+h(x)$, for some $\Qp$-linear map $h\colon X\rightarrow Y$.

\end{enumerate}
\end{proposition}

\begin{proof}
Clearly, property~\ref{affl} implies~\ref{conl}. Let us prove that~\ref{conl} implies~\ref{affl}, as well.
In fact, if $g$ is $\Qp$-convex, it is sufficient to show that $g(\alpha x + (1-\alpha) y)=\alpha g(x) + (1-\alpha) g(y)$,
$\forall x,y\in X$, $\forall\alpha\in\Qp\setminus\pint$, i.e., for $|\alpha|>1$. Let us write $y$ as a $\Qp$-convex combination
\begin{equation}
y = \frac{1}{1-\alpha}\,(\alpha x + (1-\alpha) y) + \Big(1-\frac{1}{1-\alpha}\Big) x
= \lambda(\alpha x + (1-\alpha) y) + (1-\lambda) x,
\end{equation}
where $\lambda\equiv 1/(1-\alpha),(1-\lambda)\in\pint$, because $|\lambda|=|1/(1-\alpha)|=1/|(1-\alpha)|=1/|\alpha|<1$
and $|1-\lambda|=1$. Observe, now, that we have:
\begin{align}
g(y) & = g(\lambda(\alpha x + (1-\alpha) y) + (1-\lambda) x)
\nonumber\\
& = \lambda\, g(\alpha x + (1-\alpha) y) + (1-\lambda) g(x)
= \frac{1}{1-\alpha}\,g(\alpha x + (1-\alpha) y) - \frac{\alpha}{1-\alpha}\, g(x).
\end{align}
In conclusion: $g(\alpha x + (1-\alpha) y)=\alpha g(x) + (1-\alpha) g(y)$, for $|\alpha|>1$, and, hence,
for all $\alpha\in\Qp$ ($g$ being $\Qp$-convex).

The equivalence between~\ref{affl} and~\ref{qual} can be shown by a standard argument. We leave the details to the reader
(to prove that~\ref{affl} implies~\ref{qual}, it is enough to show that the map $h\colon X\rightarrow Y$, $h(x)\defi g(x)-g(0)$,
is $\Qp$-homogeneous and, hence, also additive).
\end{proof}

It turns out that $p$-adic probability theory differs significantly from classical probability theory,
because it mainly involves \emph{affine} --- rather than convex --- structures. Nevertheless, both theories arise
in a natural way from a common conceptual background. Again, our exposition will be sketchy; for further details and examples,
see~\cite{khrennikov1993p,khrennikov2002interpretations,khrennikov2013p,khrennikov1997non}, and references therein.

The statistical output of a concrete experiment consists of (relative) frequencies of the form $n/\ttN$, where $\ttN$ is the total
number of measurements performed during the experiment and $n\le\ttN$ counts the number of measurements providing a fixed experimental outcome.
Therefore, the possible statistical outputs of each experiment take values in the following subset of the field of rational numbers:
$\mathscr{O}_\mathbb{Q}=\{q\in\mathbb{Q}\sep 0\le q\le 1\}$. Assuming that a principle of \emph{statistical stabilization}
of frequencies holds (for $\ttN\rightarrow\infty$), it follows that the closure $\cl(\mathscr{O}_\mathbb{Q})$ of $\mathscr{O}_\mathbb{Q}$,
in the completion of $\mathbb{Q}$ w.r.t.\ some suitable topology, should provide the set where all experimental statistical
distributions take their values. Usually, one assumes that this topology is the one induced by the standard valuation on
$\mathbb{Q}$, so obtaining $\cl(\mathscr{O}_\mathbb{Q})=[0,1]\subset\mathbb{R}$. It is a remarkable fact that, if one considers
the topology induced by the $p$-adic valuation, instead, then $\cl(\mathscr{O}_\mathbb{Q})=\Qp$; see Theorem~{1.2}
in Chapt.~{VI} of~\cite{khrennikov2013p}.

Therefore, we can set the following:
\begin{definition}
A (discrete) $p$-adic probability distribution is a countable set $\{\pdj\}_{j\in J}\subset\Qp$
such that $\sum_{j\in J}\pdj=1$.
\end{definition}

It is worth observing the following simple facts:
\begin{itemize}

\item The set $\{1,2,-1,-1\}$ is a legitimate $p$-adic probability distribution, whereas it is not a
standard probability distribution.

\item For every pair $\{\pdj\}_{j\in J}$, $\{\tpdk\}_{k\in K}$ of $p$-adic probability distributions,
$\{\pdj\,\tpdk\}_{j\in J,\,k\in K}$ is a $p$-adic probability distribution too.

\item For every $p$-adic probability distribution $\{\pdj\}_{j\in J}$, $\max_{j\in J} |\pdj|\ge 1$.
Indeed, we have that $1=\big|\sum_{j\in J}\pdj\big|\le\max_{j\in J} |\pdj|$.

\item For every quadratic extension $\Q$ of $\Qp$, the collection of all probability distributions indexed
by $J$ can be identified, in a natural way, with a subset of $c_0(J,\Q)$, i.e.,
\begin{equation}
\prs(J,\Q)\defi\big\{\{\pdj\}_{j\in J}\in c_0(J,\Q) \sep \mbox{$\pdj\in\Qp$, $\forall j\in J$, $\sum_{j\in J}\pdj=1$}\big\}.
\end{equation}
Note that $\prs(J,\Q)$ is a closed $\Qp$-affine subset of $c_0(J,\Q)$ --- the so-called \emph{probability hyperplane}
of $c_0(J,\Q)$ --- which, apart from the trivial case where $\card(J)=1$,
is an \emph{unbounded} subset of $c_0(J,\Q)$, because it is a translate of the (closed) $\Qp$-linear subspace
$\big\{\{x_j\}_{j\in J}\in c_0(J,\Q) \sep \mbox{$x_j\in\Qp$, $\forall j\in J$, $\sum_{j\in J}x_j=0$}\big\}$.
E.g., if $\card(J)\ge 2$, for every $n\in\nat$, the probability distribution
$\pi=\big\{p^{-n}, 1-p^{-n},0,0,\ldots\big\}\subset c_0(J,\Q)$ is such that $\|\pi\|_\infty=p^n$.

\item The probability hyperplane $\prs(J,\Q)$ contains a distinguished $\Qp$-convex subset $\psy(J,\Q)$
of $c_0(J,\Q)$ --- namely,
\begin{equation}
\psy(J,\Q)\defi\big\{\{\pdj\}_{j\in J}\in c_0(J,\Q) \sep \mbox{$\pdj\in\pint$, $\forall j\in J$, $\sum_{j\in J}\pdj=1$}\big\}
\end{equation}
--- that is a \emph{bounded} closed subset of $c_0(J,\Q)$, called the \emph{probability simplex} of $c_0(J,\Q)$.
E.g., the sequence $\big\{\pi_n=p^{n-1}(1-p)\big\}_{n\in\nat}$ belongs to $\psy(\nat,\Q)$.
\end{itemize}


\subsection{States in $p$-adic quantum mechanics}
\label{states}

In the spirit of the algebraic approach to quantum mechanics and taking into account the peculiar features of $p$-adic
probability theory, we now define a state of the unital Banach $\ast$-algebra $\Bpa$ as a suitable element of $\Bpa^\prime$,
where $\Hp$ is a $p$-adic Hilbert space over a quadratic extension $\Q$ of $\Qp$.

\begin{definition}\label{algestat}
A \emph{state}, for the $p$-adic Hilbert space $\Hp$, is a linear functional
\begin{equation}
\sta\colon\Bpa\rightarrow \Q
\end{equation}
satisfying the following conditions:
\begin{enumerate}[label=\tt{(S\arabic*)}]

\item \label{con.s1} $\sta$ is a \emph{bounded functional}, i.e., $\|\sta\|\defi\sup_{\|A\|\neq 0}|\sta(A)|/\|A\|<\infty$.

\item \label{con.s2} $\sta$ is \emph{involution-preserving}, i.e., $\sta(A^*)=\overline{\sta(A)}$, for all $A\in\Bpa$.

\item \label{con.s3} $\sta$ is \emph{normalized}, i.e., $\sta(\mathrm{Id})=1$.

\end{enumerate}
\end{definition}

By comparison with the complex setting, it is clear that a distinguishing feature of the $p$-adic case is contained in
condition~\ref{con.s2} (recall that an analogous property follows from the positivity condition in the complex case).
Also note that we have:
\begin{equation} \label{inqp}
\sta(A)=\sta(A^\ast)=\overline{\sta(A)}\implies\sta(A)\in\Qp, \quad \forall A\in\Bps.
\end{equation}

\begin{proposition}
Given any state $\sta\colon\Bpa\rightarrow \Q$ for $\Hp$, there is a bounded functional $\stext\colon\Bp\rightarrow\Q$
such that
\begin{enumerate}[label=\tt{(E\arabic*)}]

\item \label{con.e1} $\stext$ agrees with $\sta$ on $\Bpa$, i.e., $\stext(A)=\sta(A)$, for all $A\in\Bpa$;
whence, $\stext$ is involution-preserving on $\Bpa$ --- $\stext(A^*)=\overline{\stext(A)}$, $A\in\Bpa$ ---
and $\stext(\mathrm{Id})=1$.

\item \label{con.e2} $\|\stext\|=\|\sta\|$, where the norms are defined on $\Bp^\prime$ and $\Bpa^\prime$, respectively.

\end{enumerate}
\end{proposition}

\begin{proof}
Given any state $\sta$ for $\Hp$, by Theorem~\ref{th.330}, it is sufficient to take a Hahn-Banach extension
$\stext\colon\Bp\rightarrow\Q$ of the bounded functional $\sta\colon\Bpa\rightarrow\Q$ (see Theorem~\ref{th.330}).
\end{proof}

Let us denote by $\states$ the set of all states for $\Hp$. By the preceding result, $\states$ can be identified with
the quotient $\statex/\sim$ of the set
\begin{equation} \label{defstatex}
\statex\defi\{\Theta\in\Bp^\prime \sep \mbox{$\Theta$ is involution-preserving on $\Bpa$ and $\stext(\mathrm{Id})=1$}\}
\end{equation}
w.r.t.\ the equivalence relation defined by
\begin{equation}
\Theta_1\sim\Theta_2 \iffdef \Theta_1(A)=\Theta_2(A),\ \forall A\in\Bpa.
\end{equation}
Moreover, in each equivalence class of $\statex$ modulo $\sim$, there is a functional $\Theta\in\Bp^\prime$ such that
$\|\Theta\|=\|\Theta\vert_{\Bpa}\|$.

\begin{proposition}
$\states$ and $\statex$ are $\Qp$-affine subsets of $\Bpa^\prime$ and $\Bp^\prime$, respectively.
\end{proposition}

\begin{proof}
Just observe that, by Definition~\ref{algestat}, if $\sta_1,\sta_2\in\states$, then $\alpha\sta_1 + (1-\alpha)\sta_2\in\states$,
for all $\alpha\in\Qp$, and, by~\eqref{defstatex}, an analogous property holds for $\statex$ too.
\end{proof}


\subsection{The trace-induced states}
\label{tracial}

Just like in the standard complex setting, the algebraically defined $p$-adic states are somewhat too general
and vaguely characterized objects to be useful for most practical applications, and it is natural to restrict
to the much more concrete class of \emph{trace-induced} states.

For the sake of notational simplicity, we will assume that $\dim(\Hp)=\infty$, but the subsequent results
and their proofs remain valid --- with obvious adaptations --- in the finite-dimensional case (say, neglecting
the trivial case where $\dim(\Hp)=1$, for $2\le\dim(\Hp)<\infty$). In particular, if $\dim(\Hp)<\infty$,
$\Tp=\Bpa=\Bp$ and all states for $\Hp$ are trace-induced.

Let us first consider the subset $\tcst$ (where the subscript stands for `statistical') of $\Tp$ defined by
\begin{equation} \label{deftrstates}
\tcst\defi\{S\in\Tps \sep \tr(S)=1\}.
\end{equation}
We endow $\tcst$ with the relative topology w.r.t.\ the $p$-adic Hilbert-Schmidt space $\Tp$ (the norm-topology).

\begin{theorem} \label{theotcst}
$\tcst$ is a closed $\Qp$-affine subset of $\Tp$. A linear operator $S$ belongs to $\tcst$ iff it is of the form
\begin{equation}\label{sydecst}
S=\sum_{j\in J}(\sj\, |\osj\rangle\langle\fj|+\overline{\sj}\, |\fj\rangle\langle\osj|),
\end{equation}
where $J=\{1,2,\ldots\}$ is a countable index set and
\begin{itemize}

\item $\{\sj\}_{j\in J}\subset\Qa$ and, if $J=\nat$, $\lim_j\sj=0$;

\item $\{\osj\}_{j\in J}$ is a normalized norm-orthogonal system in $\Hp$,
and $\|\fj\|=1$, for all $j\in J$;

\item $\sum_{j\in J}(\sj\,\langle\fj,\osj\rangle+\overline{\sj}\,\langle\osj,\fj\rangle)=1$;

\item the sum in~\eqref{sydecst} --- whenever $J$ is not finite --- converges w.r.t.\ the norm topology.

\end{itemize}

Every $S\in\tcst$ admits a decomposition of the previous form where, in particular, the norm-orthogonal system
$\{\osj\}_{j\in J}$ is contained in any orthonormal basis in $\Hp$.

For every $S\in\tcst$, the functional $\tr(\argo S)\colon\Bpa\ni A\mapsto\tr(AS)\in\Q$ is a state for $\Hp$, and
the map
\begin{equation}
\inj\colon\tcst\ni S\mapsto\big(\tr(\argo S)\colon\Bpa\rightarrow\Q\big)\in\states
\end{equation}
is a continuous $\Qp$-affine injection of $\tcst$ into $\states$, where $\states$ is
endowed with the relative topology w.r.t.\ $\Bpa^\prime$; moreover, $\|\inj(S)\|=\|S\|$,
for all $S\in\tcst$.
\end{theorem}

\begin{proof}
The fact that $\tcst$ is a $\Qp$-affine subset of $\Tp$ is straightforward from~\eqref{deftrstates}.
By Corollary~\ref{corbotra}, the linear functional $\tr\argo\colon\Tp\rightarrow\Q$
is bounded, so that $\tcst$ --- which is the intersection of the closed subset
$\Tps$ of $\Tp$ with the (closed) pre-image w.r.t.\ $\tr\argo$ of the singleton set $\{1\}\subset\Q$
--- is closed in $\Tp$.

The second and the third assertion of the statement follow directly from the definition of $\tcst$
and Proposition~\ref{prosydec}.

Moreover, given any $A\in\Bpa$ and $T\in\Tp$, the estimate
\begin{equation} \label{simpest}
|\tr(AT)|\le\|A\|\, \|T\|,
\end{equation}
see the first inequality in~\eqref{trinequa}, shows that $\tr(\argo T)\colon\Bpa\rightarrow\Q$ is
a bounded functional.
In particular, for every $S\in\tcst\subset\Tps$, the bounded functional $\tr(\argo S)$ is both involution preserving ---
$\tr(A^\ast S)=\overline{\tr(SA)}=\overline{\tr(AS)}$, where we have used relation~\ref{conjtrace} in
Proposition~\ref{prop.tnew} and the cyclic property of the trace --- and normalized. Therefore,
$\tr(\argo S)\colon\Bpa\rightarrow\Q$ is a state for $\Hp$.

It is clear that the map $\inj$ is $\Qp$-affine. Consider, next, the linear map
\begin{equation}
\tinj\colon\Tp\ni T\mapsto\big(\tr(\argo T)\colon\Bpa\rightarrow\Q\big)\in\Bpa^\prime.
\end{equation}
Here, as previously noted, for every $T\in\Tp$ the linear functional $\tinj(T)=\tr(\argo T)$ is bounded
and, by~\eqref{simpest}, $\|\tinj(T)\|\le \|T\|$.  Now, taking $A\in\Bpa$ of the form $A=|\phi\rangle\langle\psi|$,
where $\phi,\psi$ are arbitrary vectors in $\Hp$, we have that
$\tr(AT)=\tr(|\phi\rangle\langle\psi|T)=\tr(|T\phi\rangle\langle\psi|)=\langle\psi,T\phi\rangle$.
Thus, given any orthonormal basis $\{\phi_m\}_{m\in\nat}$, we also have: $\|T\|=\sup_{m,n\in\nat}
|\langle\phi_m,T\phi_n\rangle|=\sup_{m,n\in\nat}|\tr(|\phi_m\rangle\langle\phi_n|T)|\le\|\tinj(T)\|$
($\|\,|\phi_m\rangle\langle\phi_n|\,\|=1$).
We conclude that $\|\tinj(T)\|=\|T\|$; i.e., $\tinj$ is a linear isometry, and then the $\Qp$-affine map $\inj$,
which can be regarded as a restriction of $\tinj$, is injective and continuous.
Moreover, obviously, $\|\inj(S)\|=\|S\|$, for all $S\in\tcst$.
\end{proof}

\begin{definition}
We call the operators in the $\Qp$-affine subset $\tcst$ of $\Tp$ the \emph{statistical operators} in $\Hp$.
Moreover, we call the states in the set $\trstates\defi\inj(\tcst)\subset\states$ the \emph{trace-induced states}
for $\Hp$.
\end{definition}

Since the map $\inj$ is $\Qp$-affine, then $\trstates$ is a $\Qp$-affine subset
of $\states$ that, by the final assertion of Theorem~\ref{theotcst}, can be isometrically identified with
$\tcst$. Note that, by the second estimate in~\eqref{trinequa} (also see~\eqref{tcnor}), for every trace-induced
state $\sta=\inj(S)$, $S\in\tcst$, we have:
\begin{equation}
\sta(\id)=1=\tr(S)\le\|S\|=\|\sta\|.
\end{equation}
This is a further difference w.r.t.\ the complex case, where every state $\omega$ satisfies $\|\omega\|=\omega(\id)=1$.

Let us then have a closer look at the structure of the $\Qp$-affine set $\tcst$ of all statistical operators.
To highlight this affine structure, let us first observe that --- introducing the set
\begin{equation}
\ztc\defi\{T\in\Tps \sep \tr(T)=0\},
\end{equation}
which is a (closed) $\Qp$-linear subspace of $\Tp$ --- for $S\in\tcst$ and $T\in\ztc$,
\begin{equation}
\mbox{$S+T\in\tcst$, and, if $\|S\|<\|T\|$, then $\|S+T\|=\|T\|$.}
\end{equation}
We will call $S+T$ a \emph{zero-trace perturbation} of the statistical operator $S$ (by $T$),
and, by the previous argument, it is easy to see that there exists a zero-trace perturbation of $S$
of arbitrarily large norm ($\ztc$ being a $\Qp$-linear subspace of $\Tp$). Moreover, considering a
symmetric decomposition~\eqref{sydecst} of $S$, there is a natural partition $J=J_0\sqcup J_1$ of
the index set $J$, where
\begin{equation}
J_0\defi\big\{j\in J \sep \langle\osj,\fj\rangle=0\big\} \quad \mbox{and} \quad
J_1\defi\big\{j\in J \sep \langle\osj,\fj\rangle\neq0\big\},
\end{equation}
and we can write $S=S_0+S_1$, with
\begin{equation}
S_0\defi\sum_{j\in J_0}(\sj\, |\osj\rangle\langle\fj|+\overline{\sj}\, |\fj\rangle\langle\osj|)\in\ztc,
\ S_1\defi\sum_{j\in J_1}(\sj\, |\osj\rangle\langle\fj|+\overline{\sj}\, |\fj\rangle\langle\osj|)\in\tcst.
\end{equation}
(If $J_0=\emptyset$, then we put $S_0\equiv0$.)
Here, we have used formula~\eqref{trasa} to conclude that $\tr(S_0)=0$, whereas $\tr(S_1)=\tr(S)=1$. Therefore,
$S$ is a zero-trace perturbation of $S_1$ by $S_0$.

We will now derive a useful refinement of a symmetric decomposition of a statistical operator.
To this end, we will fruitfully adopt the following convenient notation:
\begin{notation}
For every pair of \emph{nonzero vectors} $\phi,\psi\in\Hp$ --- $\phi\neq 0\neq\psi$ ---
and every $\sigma\in\Qa\equiv\Q\setminus\{0\}$,
we set
\begin{equation}
\hphps(\sigma)\defi
\begin{cases}
(\sigma+\overline{\sigma})^{-1}(\sigma\,|\phi\rangle\langle\psi| +
\overline{\sigma}\,|\psi\rangle\langle\phi|)\in\ztc \quad \mbox{if $\langle\phi,\psi\rangle=0$} \\
(\sigma\,\langle\psi,\phi\rangle + \overline{\sigma}\,\langle\phi,\psi\rangle)^{-1}(\sigma\,|\phi\rangle\langle\psi|
+ \overline{\sigma}\,|\psi\rangle\langle\phi|)\in\tcst
\quad \mbox{if $\langle\phi,\psi\rangle\neq 0$}
\end{cases} .
\end{equation}
Here, we stress that $\hphps(\sigma)\in\Tps$, and $\tr\big(\hphps(\sigma)\big)=1$, if $\langle\phi,\psi\rangle\neq 0$,
whereas $\tr\big(\hphps(\sigma)\big)=0$, otherwise. Moreover, we introduce the sets
\begin{equation}
\sts\defi\big\{S\in\tcst \sep S=\hphps(\sigma),\ \phi,\psi\in\Hp\setminus\{0\},\ \langle\phi,\psi\rangle\neq 0,
\ \sigma\in\Qa\big\},
\end{equation}
\begin{equation}
\ztcs\defi\big\{S\in\ztc \sep S=\hphps(\sigma),\ \phi,\psi\in\Hp\setminus\{0\},\ \langle\phi,\psi\rangle= 0,
\ \sigma\in\Qa\big\}.
\end{equation}
\end{notation}

Note that $\hphps(\sigma)=\hpsph(\overline{\sigma})$. Moreover, $\hphps(\alpha\sigma)=\hphps(\sigma)$, for all $\alpha\in\Qp$,
and --- in the case where $\langle\phi,\psi\rangle\neq 0$ --- an analogous relation holds true if we map $\phi$
or $\psi$ into $\alpha\phi$ or $\alpha\psi$, respectively. We also put $\hphps\equiv\hphps(1)$. In particular, given a
nonzero vector $\psi$, we have that $\hpsps(\sigma)=\hpsps=|\psi\rangle\langle\psi|\in\ztcs$, in the case where the vector
$\psi$ is \emph{isotropic}, whereas $\hpsps(\sigma)=\hpsps=\langle\psi,\psi\rangle^{-1}|\psi\rangle\langle\psi|\in\sts$, otherwise.
In the latter case, the statistical operator $\hpsps$ is a (selfadjoint) \emph{rank-one projection}: $\hpsps\,\hpsps=\hpsps$.

\begin{definition}
Let $\phi,\psi\in\Hp$ be a pair of \emph{nonzero vectors}, and let $\sigma\in\Qa$.  If $\langle\phi,\psi\rangle\neq 0$,
we say that $\hphps(\sigma)\in\sts$ is a \emph{simple statistical operator}. If, instead, $\langle\phi,\psi\rangle= 0$, we say that
$\hphps(\sigma)\in\ztcs$ is a \emph{simple zero-trace operator}.
\end{definition}

By the previously introduced notation, we can suitably re-write the symmetric decomposition
$S=\sum_{j\in J}(\sj\, |\osj\rangle\langle\fj|+\overline{\sj}\, |\fj\rangle\langle\osj|)=S_0+S_1$
of the statistical operator $S$. In fact, by construction, we have that
\begin{equation} \label{spano}
S_0=\sum_{j\in J_0} (\sj +\overline{\sj})\hefj(\sigma_j)\equiv\sum_{j\in J_0} \gamma_j\hefj(\sigma_j)\in\ztc,
\quad \mbox{where: $\gamma_j\in\Qp$, $j\in J_0$,}
\end{equation}
while
\begin{align}
S_1 & =\sum_{j\in J_1}(\sj\,\langle\fj,\osj\rangle +\overline{\sj}\,\langle\osj,\fj\rangle)\hefj(\sigma_j)
\nonumber\\ \label{affu}
& \equiv\sum_{j\in J_1}\pdj\hefj(\sigma_j)\in\tcst,\quad \mbox{where: $\pdj\in\Qp$, $j\in J_1$, and
$\sum_{j\in J_1}\pdj=\tr(S)=1$}.
\end{align}
Therefore, actually, $S_0\in\cspanq\big(\ztcs\big)$ and $S_1\in\caff\big(\sts\big)$.

\begin{remark}
Let us observe explicitly that $\cspanq\big(\ztcs\big)\subset\ztc$, since the $\Qp$-linear space $\ztc$ is closed in $\Tp$.
Moreover, we have that $\aff\big(\sts\big)\subset\tcst$, and hence $\caff\big(\sts\big)\subset\tcst$ too, because $\tcst$
is a \emph{closed} $\Qp$-affine subset of $\Tp$.
\end{remark}

It is clear that, conversely, every linear operator in $\Hp$ of the form $S=S_0+S_1$ --- where $S_0\in\cspanq\big(\ztcs\big)$
and $S_1\in\caff\big(\sts\big)$ --- is a statistical operator, because in such a case $S$ is a zero-trace perturbation of
a statistical operator $S_1$ (by $S_0$).

We eventually get to the following result:

\begin{theorem}
Every statistical operator $S$ in $\Hp$ can be expressed as a zero-trace perturbation of a statistical operator $S_1$
--- contained in the closed $\Qp$-affine hull $\caff\big(\sts\big)\subset\tcst$ generated by all simple statistical
operators --- by an operator $S_0\in\cspanq\big(\ztcs\big)\subset\ztc$. Conversely, every zero-trace perturbation of
a statistical operator --- in particular, of an operator contained in $\caff\big(\sts\big)$ by an operator in
$\cspanq\big(\ztcs\big)$ --- is a statistical operator too.

Therefore, we have that
\begin{equation}
\tcst=\ztc+\tcst=\ztc+\caff\big(\sts\big)=\cspanq\big(\ztcs\big)+\caff\big(\sts\big).
\end{equation}

Moreover, for every $T\in\tcst$, we have:
\begin{equation}
\tcst=T+\ztc;
\end{equation}
otherwise stated, $\tcst$ coincides with the (closed) $\Qp$-affine subset $T+\ztc$ of $\Tp$.
\end{theorem}

\begin{proof}
The first assertion follows form our previous discussion; in particular, for any $S\in\tcst$, from the
decomposition $S=S_0+S_1$, with $S_0,S_1$ expressed as in~\eqref{spano} and~\eqref{affu}, respectively.
For the final assertion, just note that every statistical operator $S\in\tcst$ can be expressed in the
form $S=T+(S-T)$, where $T$ is some (fixed) statistical operator and $(S-T)\in\ztc$.
\end{proof}

We have already noted that $\tcst$ --- being $\Qp$-affine --- is an \emph{unbounded} subset of $\Tp$,
in sharp contrast w.r.t.\ the complex case. It is natural to ask: Is there any $p$-adic parallel
for the \emph{bounded} convex set of density operators in a separable \emph{complex} Hilbert space?

Let us consider the following set of statistical operators:
\begin{equation}
\densr\defi\{S\in\tcst \sep \|T\|\le r\},\quad r\in\|\Tp\|\setminus\{0\}=|\Qa|.
\end{equation}
Here, note that
\begin{itemize}

\item for $r<1$, $\densr=\{S\in\tcst \sep \|T\|\le r\}=\emptyset$, because every statistical operator has norm not
smaller than 1;

\item by the previous point, for $r\ge 1$, $\densr=\{S\in\tcst \sep 1\le\|T\|\le r\}$;

\item for every for $s\in|\Qa|$, $s\ge1$, $\cup_{s\le r\in|\Qa|}\hspace{0.3mm}\densr=\tcst$;

\item expressing a statistical operator $S\in\tcst$ as a matrix operator --- $S=\opPhi(S_{mn})$,
for some orthonormal basis $\Phi\equiv\{\phi_m\}_{m\in\nat}$ --- it is easy to see, by means of
explicit examples, that $\densu=\{T\in\tcst \sep \|T\|=1\}\neq\emptyset$ and, for $r,s\in|\Qa|$, with $r>s\ge1$,
$\densr\supsetneq\denss$, because $\|\densr\|=\{t\in|\Qa| \sep 1\le t\le r\}$ (note, by the way, that
$\|\tcst\|=\{t\in|\Qa| \sep t\ge 1\}$);

\item by Theorem~\ref{charconv} and by the strong triangle inequality for the operator norm,
for every $r\in|\Qa|$, $r\ge1$, $\densr$ is a $\Qp$-convex subset of $\Tp$.

\end{itemize}

\begin{definition}
We call \emph{density operators} those statistical operators in the $p$-adic Hilbert space $\Hp$ belonging to the bounded
subset $\densops\equiv\densu$ of $\Tp$ defined by
\begin{equation} \label{defdensops}
\densops\defi\{T\in\tcst \sep \|T\|\le 1\}=\{T\in\tcst \sep \|T\|=1\}.
\end{equation}
\end{definition}

\begin{example}
If $\psi\in\Hp$, is a non-isotropic nonzero vector, then the statistical operator
$\hpsps=\langle\psi,\psi\rangle^{-1}|\psi\rangle\langle\psi|\in\sts$ is a density operator iff
$|\langle\psi,\psi\rangle|=\|\psi\|^2$. For every orthonormal basis $\{\phi_m\}_{m\in\nat}$
and every probability distribution $\{\pi_m\}_{m\in\nat}$ contained in the probability
simplex $\psy(\nat,\Q)$, $\sum_{m\in\nat} \pi_m\, |\phi_m\rangle\langle\phi_m|\in\densops$.
E.g., $\sum_{m\in\nat}p^{m-1}(1-p)\,|\phi_m\rangle\langle\phi_m|$ is a density operator.
\end{example}

\begin{proposition} \label{charadens}
$\densops$ is a closed $\Qp$-convex subset of $\Tp$. For every trace class operator $T\in\Tp$,
the following facts are equivalent:
\begin{enumerate}[label=\tt{(D\arabic*)}]

\item \label{densa} $T\in\densops$.

\item \label{densb} $T\in\Tps$ and admits a canonical decomposition of the form
$T=\sum_{j\in J}\lambda_j\, |\osj\rangle\langle\fj|$, where $\{\lambda_j\}_{j\in J}$ is contained in the
valuation ring $\vari\equiv\Qba\defi\{z\in\Q \sep |z|\le 1\}$  of $\Q$ and
$\sum_{j\in J}\lambda_j\, \langle\fj,\osj\rangle=1$.

\item \label{densbis} $T\in\Tps$ and admits a canonical decomposition of the form
$T=\sum_{j\in J}\lambda_j\, |\osj\rangle\langle\fj|$, where $\max_{j\in J}|\lambda_j|=1$
and $\sum_{j\in J}\lambda_j\, \langle\fj,\osj\rangle=1$.

\end{enumerate}

If $p\neq 2$, conditions~\ref{densa}--\ref{densbis} are equivalent to the following:
\begin{enumerate}[label=\tt{(D\arabic*)}]

\setcounter{enumi}{3}

\item \label{densc} $T\in\Tps$ and  admits a symmetric decomposition of the form
$T=\sum_{j\in J}(\sj\,|\osj\rangle\langle\fj|+\overline{\sj}\, |\fj\rangle\langle\osj|)$, where
$\{\sj\}_{j\in J}\subset\vari$ and $\sum_{j\in J}(\sj\,\langle\fj,\osj\rangle
+\overline{\sj}\,\langle\osj,\fj\rangle)=1$.

\end{enumerate}

Finally, for every $r\ge 1$ and any density operator $T\in\densops$, we have that
\begin{equation} \label{strdensr}
\densr= T+\ztcr,
\end{equation}
where
\begin{equation} \label{defztcr}
\ztcr\defi\{A\in\Tps\sep \tr(A)=0,\ \|A\|\le r\}
\end{equation}
is an absolutely $\Qp$-convex subset of $\Tp$.
\end{proposition}

\begin{proof}
As previously noted, by Theorem~\ref{charconv} and by the strong triangle inequality for the operator norm,
$\densops$ is a $\Qp$-convex subset of $\Tp$. Moreover, $\densops$ is the intersection of the closed unit ball
(or of the unit sphere) in $\Tp$ with the closed subset $\tcst$ of $\Tp$; hence, it is
closed in $\Tp$.

By relations~\eqref{tctra} and~\eqref{tcnor}, given any canonical decomposition
$T=\sum_{j\in J}\lambda_j\, |\osj\rangle\langle\fj|$ of a trace class operator $T\in\Tp$,
we have that $\tr(T)=\sum_{j\in J}\lambda_j\, \langle\fj,\osj\rangle$ and
$\|T\|=\max_{j\in J} |\lambda_j|$ (so that $\|T\|\le1\iff\{\lambda_j\}_{j\in J}\subset\pint$).
Therefore, conditions~\ref{densa} and~\ref{densb} are equivalent, and, by the second equality
in~\eqref{defdensops}, conditions~\ref{densb} and~\ref{densbis} are equivalent too.

Moreover, if condition~\ref{densb} is satisfied, writing $T=\frac{1}{2}(T+T^\ast)$, we obtain the
symmetric decomposition $T=\sum_{j\in J}(\sj\,|\osj\rangle\langle\fj|+\overline{\sj}\, |\fj\rangle\langle\osj|)$, where
$\sum_{j\in J}(\sj\,\langle\fj,\osj\rangle+\overline{\sj}\,\langle\osj,\fj\rangle)=\tr(T)=
\sum_{j\in J}\lambda_j\, \langle\fj,\osj\rangle=1$ and $\{2\sj=\lambda_j\}_{j\in J}\subset\vari$;
in particular, if $p\neq 2$, then $|2\sj|=|\sj|$ and $\{\sj\}_{j\in J}\subset\vari$.
Thus,~\ref{densb} implies~\ref{densc}.

Next, if~\ref{densc} holds, then $T\in\Tps$ is such that
$\tr(T)=\sum_{j\in J}(\sj\,\langle\fj,\osj\rangle+\overline{\sj}\,\langle\osj,\fj\rangle)=1$ and
$\|T\|=\|\sum_{j\in J}(\sj\,|\osj\rangle\langle\fj|+\overline{\sj}\, |\fj\rangle\langle\osj|)\|
\le\max_{j\in J}\|\sj\,|\osj\rangle\langle\fj|\,\|=\max_{j\in J}|\sj|\le 1$, so that
condition~\ref{densa} is verified too.

Finally, by the linearity of the trace and by the strong triangle inequality for the operator norm,
it is easy to see that $\ztcr$ is a (closed) absolutely $\Qp$-convex subset of $\Tp$. Now, for every $S\in\densr$,
we have that $S=T+(S-T)$, for any $T\in\densops$, where $\tr(S-T)=0$ and $\|S-T\|\le\|S\|\le r$; i.e., $S-T\in\ztcr$.
On the other hand, if $S=T+S_0$, with $T\in\densops$ and $S_0\in\ztcr$ --- i.e., if $S$ is a zero-trace perturbation
of the density operator $T$ by $S_0$ --- then $S\in\tcst$ and $1\le\|S\|\le\max\{\|T\|,\|S_0\|\}=
\max\{1,\|S_0\|\}\le\max\{1,r\}=r$; i.e., $S\in\densr$. Therefore, relation~\eqref{strdensr} holds true.
\end{proof}


\subsection{The SOVMs and the statistical interpretation}
\label{sovms}

The statistical interpretation of trace-induced states is ensured by defining the observables as a suitable
$p$-adic counterpart of the POVMs:

\begin{definition}
A (discrete) \emph{selfadjoint-operator-valued measure} (in short, a SOVM) in $\Hp$ is a countable family
$\{A_i\}_{i\in I}\subset\Bps$, which is norm-bounded --- i.e., $\sup_{i\in I}\|A_i\|<\infty$ --- and
such that $\sum_{i\in I} A_i=\id$. Here, if the index set $I$ is not finite,
the series is supposed to converge w.r.t.\ weak operator topology (i.e., the initial topology induced by the family of maps
$\{\ephps\colon\Bp\rightarrow\Q\}_{\phi,\psi\in\Hp}$, $\ephps(A)\defi\langle\phi,A\psi\rangle$). We call a SOVM
$\{A_i\}_{i\in I}\subset\Bps$ \emph{contractive} if, in particular, $\|A_i\|\le 1$, for all $i\in I$; we say that it is
\emph{trace-class} if $\{A_i\}_{i\in I}\subset\Tps$.
\end{definition}

\begin{lemma} \label{lemsovm}
Let $\{A_i\}_{i\in I}\subset\Bps$ be a SOVM, where the index set $I$ is (countably) infinite.
Then, for every pair of vectors $\phi,\psi\in\Hp$,
\begin{equation}
\lim_i\langle\phi,A_i\psi\rangle =0.
\end{equation}
\end{lemma}

\begin{proof}
In fact, we have that
\begin{align}
\sum_{i\in I} A_i=\id\ \mbox{(w.r.t.\ the weak op.\ topology)} & \implies
\langle\phi,\psi\rangle=\langle\phi,({\textstyle\sum_{i\in I}}A_i)\psi\rangle
=\sum_{i\in I}\langle\phi,A_i\psi\rangle
\nonumber\\
& \implies
\lim_i\langle\phi,A_i\psi\rangle =0,
\end{align}
where the second implication holds by Proposition~\ref{sumlemma}.
\end{proof}

\begin{proposition}
Let $\{A_i\}_{i\in I}\subset\Bps$ be a SOVM. Then, for every trace-induced state $\sta\in\trstates$,
$\{\sta(A_i)\}_{i\in I}$ is a $p$-adic probability distribution. In particular, if
$\sta$ is a \emph{density state} --- i.e., $\sta=\inj(S)$, for some $S\in\densops$ --- and $\{A_i\}_{i\in i}$
is contractive, then $\{\sta(A_i)\}_{i\in I}$ is contained in the probability simplex $\psy(I,\Q)$.
Finally, the SOVM $\{A_i\}_{i\in I}$ is contractive iff $\sup_{i\in I}\|A_i\|=\max_{i\in I} \|A_i\|=1$.
\end{proposition}

\begin{proof}
By relation~\eqref{inqp}, $\{\sta(A_i)\}_{i\in I}\subset\Qp$ and, if the index set $I$ is finite,
the first assertion follows directly from the condition that $\sum_{i\in I} A_i=\id$. Thus, we will
henceforth assume that the index set $I$ is (countably) infinite.

Given any trace-induced state $\sta\in\trstates$, we have that $\sta=\inj(S)$, for some $S\in\tcst$,
and, taking any canonical decomposition $S=\sum_{j\in J}\lambda_j\, |\osj\rangle\langle\fj|$
of the trace class operator $S$, we have that
\begin{align}
\sta(A_i)=\tr(A_i S) =\tr\big(A_i {\textstyle\sum_{j\in J}}\lambda_j\, |\osj\rangle\langle\fj|\big)
& =\tr\big({\textstyle\sum_{j\in J}}\lambda_j\, |A_i\osj\rangle\langle\fj|\big)
\nonumber\\ \label{firela}
& =\sum_{j\in J}\lambda_j\,\tr\big(|A_i\osj\rangle\langle\fj|\big)
=\sum_{j\in J}\lambda_j\, \langle\fj,A_i\osj\rangle.
\end{align}
Here, the third equality follows from the fact that the linear map $\Tp\ni T\mapsto A_i T\in\Tp$ is
bounded, and for obtaining the fourth equality we have used the fact that $\tr\argo\colon\Tp\rightarrow\Q$
is a bounded functional and, if $J=\nat$, the series $\sum_{j\in J}\lambda_j\, |A_i\osj\rangle\langle\fj|$
converges w.r.t.\ the norm topology.

Observe now that, by Lemma~\ref{lemsovm}, $\lim_i\lambda_j\,\langle\fj,A_i\osj\rangle=0$, for all $j\in J$.
Moreover, in the case where $J=\nat$, since $\lim_j\lambda_j=0$ and $\alpha\equiv\sup_{i\in I}\|A_i\|<\infty$
(by the definition of a SOVM), then
\begin{equation}
|\langle\fj,A_i\osj\rangle|\le\alpha\|\fj\|\,\|\osj\|=\alpha,\ \forall i,j\in I  \ \implies \
\lim_j\lambda_j\,\langle\fj,A_i\osj\rangle=0, \ \mbox{\emph{uniformly} in $i\in I$}.
\end{equation}
Therefore, we can freely exchange the sums in the following calculation:
\begin{align}
\sum_{i\in I}\sta(A_i) & =\sum_{i\in I}\sum_{j\in J}\lambda_j\, \langle\fj,A_i\osj\rangle
\nonumber\\
& =\sum_{j\in J}\sum_{i\in I}\lambda_j\, \langle\fj,A_i\osj\rangle
\nonumber\\
& =\sum_{j\in J}\lambda_j\, \langle\fj,({\textstyle\sum_{i\in I}} A_i)\osj\rangle
=\sum_{j\in J}\lambda_j\, \langle\fj,\osj\rangle=\tr(S)=1.
\end{align}
Here, we have used relation~\eqref{firela} and the fact that $\sum_{i\in I} A_i=\id$ (w.r.t.\ the weak operator topology).
In conclusion, it is proven that $\{\sta(A_i)=\tr(A_i S)\}_{i\in I}\subset\Qp$ is a $p$-adic probability distribution.

For the second assertion, just note that, if $\|A_i\|\le 1$, then, for every canonical decomposition
$S=\sum_{j\in J}\lambda_j\, |\osj\rangle\langle\fj|$ of the statistical operator $S\in\tcst$, we have that $|\sta(A_i)|=|\tr(A_i S)|=
|\sum_{j\in J}\lambda_j\, \langle\fj,A_i\osj\rangle|\le\max_{j\in J}|\lambda_j\,\langle\fj,A_i\osj\rangle|
\le\max_{j\in J}|\lambda_j|$. Therefore, if $S\in\densops$, then, by the equivalence of conditions~\ref{densa}
and~\ref{densb} in Proposition~\ref{charadens}, $|\sta(A_i)|\le 1$.

Regarding the final assertion, we only need to show that if $\{A_i\}_{i\in I}$ is a contractive SOVM, then
$\sup_{i\in I}\|A_i\|=\max_{i\in I}\|A_i\|=1$. Indeed, if $\{A_i\}_{i\in I}$ is a contractive SOVM,
we must have: $1=\|\id\|=\|\sum_{i\in I} A_i\|\le\sup_{i\in I} \|A_i\|\le 1$.
(Here, $\sup_{i\in I}\|A_i\|=\max_{i\in I}\|A_i\|$, because the only accumulation point of the subset
$\|\Bp\|=|\Q|$ of $\mathbb{R}^+$ is $0$.)
\end{proof}

\begin{example} \label{exsovms}
Given a trace-induced state $\tr(\argo S)\in\trstates$, with every symmetric decomposition
$S=\sum_{j\in J}(\sj\,|\osj\rangle\langle\fj|+\overline{\sj}\, |\fj\rangle\langle\osj|)$, $J=\{1,2,\ldots\}$,
of the statistical operator $S$ is associated, in a natural way, the $p$-adic probability distribution
$\big\{\pi_0=0,\pi_1=\sigma_1\,\langle f_1,e_1\rangle+\overline{\sigma_1}\,\langle e_1,f_1\rangle,\ldots\big\}$
and the SOVM
\begin{equation}
\big\{A_0=\id-S=\id-{\textstyle\sum_{j\in J}}(\sj\,|\osj\rangle\langle\fj|+\overline{\sj}\, |\fj\rangle\langle\osj|),
A_1=\sigma_1\,|e_1\rangle\langle f_1|+\overline{\sigma_1}\, |f_1\rangle\langle e_1|,\ldots\big\} .
\end{equation}
Note that, for $p\neq 2$, $\{A_0,A_1,\ldots\}$ is contractive if $S\in\densops$, because $\|A_0\|\le\max\{1,\|S\|\}=1$, and
$\|A_j\|=\|\sigma_j\,|\osj\rangle\langle\fj|\,\|=|\sigma_j|\le 1$, for all $j\in J$  (see Proposition~\ref{charadens}).
But (for every prime number $p$, and for $\dim(\Hp)=\infty$) it is \emph{not} trace-class because $A_0=\id-S\not\in\Tp$.
For every orthonormal basis $\{\phi_m\mnat$ in $\Hp$, $\{|\phi_m\rangle\langle\phi_m|\mnat$ is a contractive
trace-class SOVM.
\end{example}


\section{Conclusions and prospects}
\label{conclu}

This work is the first step of a program aimed at developing an abstract approach to quantum mechanics
and quantum information theory over a quadratic extension of the field $\Qp$ of $p$-adic numbers, where
$p$ is a generic prime number.

Let us briefly summarize our main achievements:
\begin{enumerate}

\item We have first introduced a suitable notion of a Hilbert space $\Hp$ over a quadratic extension $\Q$
of $\mathbb{Q}_p$. A salient point is that the existence of an orthonormal basis is part of the definition
of $\Hp$. (The space $\Hp$ will be supposed henceforth to be infinite-dimensional.)

\item We have then derived various properties of the ultrametric Banach space $\Bp$ of bounded operators.
A peculiar fact is that not every element of $\Bp$ admits a \emph{proper adjoint}. The \emph{adjointable}
elements of $\Bp$ form an ultrametric Banach $\ast$-algebra $\Bpa$. Bounded operators in $\Hp$ have been
mostly regarded as \emph{matrix operators}. This turns out to be a far-reaching approach in this
non-Archimedean setting.

\item In fact, the unitary operators $\Up\subset\Bpa$ in $\Hp$ have been introduced as suitable matrix operators.
We have then obtained a complete characterization of the unitary group $\Up$. In particular, $\Up$ is
shown to be the intersection of the group $\ipp(\Hp)$ of all surjective, IP-preserving, all-over operators,
with the group $\nop(\Hp)$ of all surjective, NO-preserving, all-over operators. This result is coherent
with the fact that the inner product and the norm are not as strictly related as in the standard complex
case.

\item  Once again, the \emph{trace class} $\Tp\subset\Bpa$ of $\Hp$ has been defined as a suitable class of
matrix operators and shown to be a two-sided $\ast$-ideal in $\Bpa$. A very peculiar feature of the $p$-adic
setting is that $\Tp$ --- endowed with the standard operator norm and with the Hilbert-Schmidt product --- is a
$p$-adic Hilbert space. Moreover, $\Tp^2=\Tp$, whence $\Tp^n=\Tp$, for all $n\in\nat$. Otherwise stated, $\Tp$
alone plays essentially the role that the trace ideals have in the complex case.

\item As in the complex case, $\Tp$ is contained in the Banach space $\Cp$ of \emph{compact operators}; \emph{but},
in the $p$-adic case we have the somewhat \emph{eccentric} fact that $\Tp=\Cp\cap\Bpa$.

\item Following essentially the same route as in the algebraic formulation of (standard) quantum mechanics,
we get to the conclusion that physical \emph{states} in $p$-adic quantum mechanics should be defined as
(suitably normalized) \emph{involution-preserving} bounded functionals on the unital $\ast$-algebra $\Bpa$.
The role played by the $\sigma$-additive states in the complex case --- the states associated with the density operators
--- is now played by the so-called \emph{trace-induced states} $\trstates$, induced, via the trace functional, by the
\emph{statistical operators} $\tcst$. Once again we have a (two-fold) peculiar feature of the $p$-adic case:
$\tcst$ is a \emph{$\Qp$-affine} subset of $\Tp$ --- coherently with the affine structure of $p$-adic probability
distributions --- whence it is an \emph{unbounded} subset of $\Tp$. Nevertheless, one can still define a \emph{$\Qp$-convex}
(and norm-bounded) subset $\densops$ of \emph{density operators}.

\item The statistical interpretation of the (new) theory is completed by suitably defining the \emph{observables}
in $p$-adic quantum mechanics. We believe that the \emph{selfadjoint-operator-valued measures} (SOVMs) --- a suitable $p$-adic
counterpart of the POVMs associated with a complex Hilbert space --- provide a convenient mathematical tool for this scope.
They allow one to associate with every trace-induced state a $p$-adic probability distribution. Note that with every symmetric decomposition
of a statistical operator is associated, in a natural way, a $p$-adic probability distribution and a SOVM (Example~\ref{exsovms}).
This may be somewhat reminiscent of the probability distribution and the PVM canonically associated, via the spectral decomposition, with
a density operator in a complex Hilbert space.

\end{enumerate}

We now briefly outline some natural continuations of our work. It would be futile, at this preliminary stage,
to try to envisage all potential developments and applications, so we will keep within a reasonably narrow
neighborhood of our present understanding and results. As a first point, we plan to investigate the symmetry
transformations in the $p$-adic setting, which entails studying symmetry (group) actions, a representation theory
oriented towards applications to physics (say, the `projective representations') etc. A related topic is the investigation
of dynamical maps and dynamical (semi-)groups in $p$-adic quantum mechanics. Another interesting prospect is the study of
the fine structure of the set of statistical operators, and the characterization of those states that are \emph{not} trace-induced.
Tensor products and entanglement are central in quantum information theory, and we expect that they will be central in
the $p$-adic setting too. A further intriguing open problem is the description of the `weak trace class' $\Tw$.
We already know that it contains the compact operators $\Cp$; but, \emph{is there anything else?}

\section*{Acknowledgements}

The authors gratefully acknowledge useful discussions with Andrei Yur'evich Khrennikov and Andreas Winter.


\end{document}